\documentclass[twoside,11pt,preprint]{article}

\usepackage{thmtools}
\usepackage{thm-restate} 
\usepackage[utf8]{inputenc} 

\usepackage{amsmath}
\usepackage{amsfonts}       
\usepackage{amssymb}

\usepackage{xcolor}
\usepackage{enumitem} 
\usepackage{nicefrac}       
\usepackage[sort=def]{glossaries} 
\usepackage{graphicx} 
\usepackage{float} 
\usepackage{tikz} 
\usepackage{mathtools} 
\usepackage{algorithmic} 
\usepackage{algorithm}
\usepackage{array} 
\usepackage{caption}
\usepackage{jmlr2e}

\newcommand{\comp}{\sim}
\newcommand{\cc}[1]{\mathcal{#1}}
\newcommand{\bb}[1]{\mathbb{#1}}
\DeclareMathOperator*{\argmax}{arg\,max}
\DeclareMathOperator*{\argmin}{arg\,min}
\DeclareMathOperator*{\eqdist}{\overset{!}{=}}
\DeclareMathOperator*{\neqdist}{\not\overset{d}{\neq}}
\newcommand{\indep}{\perp \!\!\! \perp}

\newcommand{\p}{P}
\newcommand{\pa}{\textnormal{PA}}

\newcommand{\diag}[1]{\textnormal{diag}\left(#1\right)}

\newcommand{\supp}[1]{\textnormal{supp}({#1})}

\newcommand{\norm}[1]{\left\lVert#1\right\rVert}
\newcommand{\hatimec}{\hat{\cc{I}}\text{-MEC}}
\newcommand{\imec}{\cc{I}\mathrm{-MEC}}
\newcommand{\simec}{\cc{I}^\star\text{-MEC}}
\newcommand{\hmec}{\cc{H}\text{-MEC}}
\newcommand{\ymec}{\cc{Y}\text{-MEC}}
\newcommand{\preds}[1][p]{\{1,...,#1\}}
\newcommand{\var}{\textnormal{Var}}

\newcommand{\abs}[1]{\left\lvert {#1} \right\rvert}
\newcommand{\absb}[1]{\big\lvert {#1} \big\rvert}

\renewcommand{\dag}{G}

\declaretheorem[name=Lemma]{lemma} 
\declaretheorem[name=Proposition]{proposition}
\declaretheorem[name=Theorem]{theorem}
\newtheorem{definition}{Definition}
\newtheorem{corollary}{Corollary}
\newtheorem{assumption}{Assumption}
 
\newcommand{\ntext}[1]{#1}
\newcommand\rmj[1]{}
\newcommand{\comm}[1]{}
\newcommand{\juan}[1]{}
\newcommand\Peter[1]{}
\newcommand\Armeen[1]{}
\newcommand\Christina[1]{}

\definecolor{g}{RGB}{106, 201, 151}

\newcommand{\ie}{i.e., }

\renewcommand{\glossarysection}[2][]{\paragraph{Notation} Used throughout the proofs of this section.}
\makeglossaries

\newglossaryentry{p}
{
    name=\ensuremath{p},
    description={the number of observed variables}
}

\newglossaryentry{preds}
{
    name=\ensuremath{[p]},
    description={the set of variable indices, \ie $\{1,...,p\}$}
}

\newglossaryentry{id}
{
    name=\ensuremath{I},
    description={the $p \times p$ identity matrix}
}

\newglossaryentry{unit_vector}
{
    name=\ensuremath{e_i},
    description={the $i^\text{th}$ canonical unit vector in $\bb{R}^p$, \ie the vector that is all zeros except for a 1 at the $i^\text{th}$ coordinate}
}

\newglossaryentry{gdiag}
{
    name=\ensuremath{\diag{.}},
    description={the diagonal elements of a matrix}
}

\newglossaryentry{support}
{
    name=\ensuremath{\supp{.}},
    description={the indices of a vector's non-zero elements}
}

\newglossaryentry{geqdist}
{
    name=\ensuremath{\eqdist},
    description={equality in distribution}
}

\newglossaryentry{row_a}
{
    name=\ensuremath{A_{i:}},
    description={the $i^\text{th}$ row of A}
}

\newglossaryentry{col_a}
{
    name=\ensuremath{A_{:i}},
    description={the $i^\text{th}$ column of A}
}

\newglossaryentry{row_expr}
{
    name=\ensuremath{[.]_{i:}},
    description={the $i^\text{th}$ row of the expression in brackets}
}

\newglossaryentry{col_expr}
{
    name=\ensuremath{[.]_{:i}},
    description={the $i^\text{th}$ column of the expression in brackets}
}

\newglossaryentry{submatrix}
{
    name=\ensuremath{A_{R,C}},
    description={the submatrix of $A$ formed by selecting all elements from rows $R$ and columns $C$}
}

\newglossaryentry{dsep}
{
    name=\ensuremath{i \indep_{\dag} j | S},
    description={$i$ and $j$ are d-separated in graph $\dag$ given $S$}
}

\newglossaryentry{parents}
{
    name=\ensuremath{\pa_\dag(i)},
    description={the parents of node $i$ in the DAG $\dag$}
}

\newglossaryentry{dag_of}
{
    name=\ensuremath{\cc{G}(B)},
    description={the directed graph resulting from the matrix $B$, that is, $B_{ij}\neq 0 \iff j \to i \in \cc{G}(B)$}
}

\newglossaryentry{comp}
{
    name=\ensuremath{B \comp \dag},
    description={the matrix $B$ is compatible with the graph $\dag$, \ie all edges in $\cc{G}(B)$ appear in $\dag$}
}

\glsaddall

\author{Juan L. Gamella\\
\addr Seminar for Statistics\\ETH Zürich\\Zürich, Switzerland
\email juan.gamella@stat.math.ethz.ch
\AND
Armeen Taeb\\
\addr Department of Statistics, University of Washington
\email ataeb@uw.edu
\AND
Christina Heinze-Deml\\
\addr Seminar for Statistics, ETH Zürich
\email heinzedeml@stat.math.ethz.ch
\AND
Peter B{\"u}hlmann\\
\addr Seminar for Statistics, ETH Zürich
\email peter.buehlmann@stat.math.ethz.ch
}

\begin{document}

\title{Characterization and Greedy Learning of Gaussian Structural Causal Models under Unknown Interventions}
\ShortHeadings{Gaussian Causal Models with Unknown Interventions}{Gamella, Taeb, Heinze-Deml and B\"uhlmann}
\editor{}

\maketitle

\begin{abstract}
We consider the problem of recovering the causal structure underlying observations from different experimental conditions when the targets of the interventions in each experiment are unknown. We assume a linear structural causal model with additive Gaussian noise and consider interventions that perturb their targets while maintaining the causal relationships in the system. Different models may entail the same distributions, offering competing causal explanations for the given observations. We fully characterize this equivalence class and offer identifiability results, which we use to derive a greedy algorithm called GnIES to recover the equivalence class of the data-generating model without knowledge of the intervention targets. In addition, we develop a novel procedure to generate semi-synthetic data sets with known causal ground truth but distributions closely resembling those of a real data set of choice. We leverage this procedure and evaluate the performance of GnIES on an array of synthetic\rmj{, real,} and semi-synthetic data sets, \ntext{and real data from a biological system and a tightly controlled physical system}.
\rmj{Despite the strong Gaussian distributional assumption, GnIES is robust to an array of model violations and competitive in recovering the causal structure in small- to large-sample settings.} We provide, in the Python packages \href{https://github.com/juangamella/gnies}{\texttt{gnies}} and \href{https://github.com/juangamella/sempler}{\texttt{sempler}}, implementations of GnIES and our semi-synthetic data generation procedure. 
\end{abstract}%
\section{Introduction}
\label{s:intro}
Knowledge of the causal relationships underlying a system or phenomenon allows us to predict its behavior as it undergoes manipulations. Examples of such knowledge are ubiquitous in human cognition starting from an early age \citep{gopnik2004theory, muentener2018development}. It is also of particular importance in science, where understanding a phenomenon usually means understanding the causal relationships which govern it.
Scientists are typically interested in causal relationships because they want to intervene on the system \citep{didelez2005statistical}, with ambitious goals ranging from curing diseases to limiting the effects of climate change.

Causal models \citep{pearl2009causal} aim at approximating causal relations. They can be seen as abstractions of more accurate physical or mechanistic models that retain enough power to answer interventional or counterfactual queries \citep{peters2017elements}. A causal understanding also benefits traditional prediction and regression problems, as it allows for robust predictions which generalize to new and unseen environments \citep{Subbaswamy2018LearningPM, buhlmann2018invariance, Rothenhausler2018AnchorRH, pfister2019stabilizing, heinze2021conditional}.

Inferring causal models from data is commonly referred to in the literature as \emph{causal learning} or \emph{causal discovery}. Akin to statistical learning, this task is challenged by the difficulty of trying to infer properties of a distribution from a finite sample. However, even full knowledge of the underlying distribution does not make causal learning trivial, as different causal models can entail the same distribution, and thus offer competing explanations for the observations gathered from the system. Under additional assumptions about the model class or noise distributions, identifiability of the true causal model is possible \citep{shimizu2006linear, hoyer2009nonlinear, peters2011identifiability, peters2014identifiability, buhlmann2014cam}. Without making such assumptions,
identifiability can
be improved through observations of the system under additional experimental conditions, where some variables have received interventions. This gain in identifiability strongly depends on the nature and targets of the experiments \citep{eberhardt2007causation, eberhardt2007interventions, hauser2014two, gamella2020active}.

\subsection{Related Work}

Motivated by the gain in identifiability, several works in the causal discovery literature deal with data collected from different interventional environments, for example, observations gathered under different experimental conditions \citep{hauser2012characterization, magliacane2016ancestral, zhang2017causal, yang2018characterizing}. However, they require full knowledge of the intervention targets, an assumption that may not be satisfied in many settings, for example, when interventions may have off-target effects.
Methods based on the invariance principle \citep{peters2016causal,meinshausen2016methods,ghassami2017learning,heinze2018invariant,pfister2019invariant} do not require such knowledge but still place restrictions on which targets are allowed; as a result, they often estimate only parts of the causal graph, such as the direct causes of a variable of interest. Moreover, while some of these methods allow control against false causal selection, they are often overly conservative.

Other recently proposed methods place weaker assumptions on the intervention targets. Joint causal inference\footnote{This paper also contains a detailed (but non-exhaustive) summary of methods that deal with data from multiple environments and their assumptions \citep[table 4]{mooij2020joint}.}
\citep{mooij2020joint} treats environments as additional nodes in the causal graph, with some flexible assumptions on their structure; this allows extending standard methods based on conditional independence testing to this setting, such as the PC \citep{spirtes2000causation}, IC \citep{pearl2009causal}, FCI \citep{spirtes1999algorithm, zhang2008completeness}, and GSP \citep{solus2017consistency} algorithms. An extension of the latter, the recent UT-IGSP algorithm additionally employs invariance tests to estimate the causal graph without knowledge of the intervention targets, under flexible assumptions on the nature of the interventions \citep[Assumptions 1 and 2]{squires2020permutation}. \citet{jaber2020causal} study the same setting in the presence of hidden confounders; the authors provide an algorithm based on the same kind of tests but provide no empirical evaluation of its performance. A drawback of methods based on conditional independence and invariance tests is that these usually require large samples, making them ill-suited for situations where only a few observations of the system per environment are available.

Recent works dealing with a small-sample setting employ gradient-based optimization on a continuous relaxation of the causal discovery problem \citep{zheng2018dags,yu2019dag,lachapelle2019gradient,zheng2020learning}. Extensions have been made to the setting with interventional data with unknown targets \citep{ke2019learning,brouillard2020differentiable,faria2022differentiable,hagele2022bacadi}. However, there are indications that the remarkable performance of some of these methods relies on an artifact of the synthetic evaluation data, and can be matched by a simple algorithm that directly exploits it \citep{reisach2021varsort}. Furthermore, \citet{seng2022notears} found this reliance to be a vulnerability to adversarial attacks.

\subsection{Outline and Contributions}

We consider linear Gaussian structural causal models for interventional or perturbation data as detailed in \autoref{s:model}. Such models, often used for different environments or domain sources, have become increasingly popular in the advent of heterogeneous data \citep{hauser2012characterization,hauser2015jointly,peters2016causal,Rothenhausler2018AnchorRH,taeb2021perturbations}. Our model is a special case of the one used in \citet{taeb2021perturbations}, but without considering the setting with hidden variables. The simpler model allows us to derive a precise characterization of identifiable structures (\autoref{s:equivalence}). Based on this result, we build a greedy algorithm (\autoref{s:greedy}) which we call \emph{greedy noise-interventional equivalence search} (GnIES) and which is less heuristic-based than the ad-hoc optimization in \citet{taeb2021perturbations}. Yet, our model in \autoref{s:model} is still much more flexible than the ones for observational data only \citep{chickering2002optimal, van_de_Geer_2013} or when interventions have known targets \citep{hauser2012characterization,hauser2015jointly,wang2017permutation}.

Our proposed algorithm GnIES can be seen as an extension of the celebrated greedy equivalence search (GES) of the penalized maximum likelihood estimator for observational data \citep{chickering2002optimal} to the setting of interventional data with unknown intervention targets. By aggregating statistical strength across environments, the maximum likelihood estimator becomes particularly attractive when there are many different environments or domain sources with few observations only. In \autoref{s:experiments} we evaluate the performance of the algorithm on synthetic and real data sets. The results show that despite the strong distributional assumptions (linearity and normality), GnIES is robust to an array of model violations, and is competitive in recovering the causal structure in small-sample settings. We study its computational complexity and find it to be on par with competing methods.

The algorithm is made available in the Python package \href{https://github.com/juangamella/gnies}{\texttt{gnies}}. More details about the package and the code to reproduce the experiments and figures from the paper can be found in the repository \href{https://github.com/juangamella/gnies-paper}{\texttt{github.com/juangamella/gnies-paper}}. A detailed summary of the software contributions of the paper, including Python implementations of the baselines and the package \href{https://github.com/juangamella/sempler}{\texttt{sempler}} to generate semi-synthetic data, can be found in Appendix \ref{s:software}.

\section{Our Model}
\label{s:model}
We consider the setting where we have access to observations of variables under different environments, such as experimental conditions in which some of the variables may have been manipulated, that is, received interventions. To represent this setting, we consider a vector $X^e \in \bb{R}^p$ of random variables observed in environment $e$, and a collection $\cc{E}$ of such environments. We denote the resulting collection of distributions as $\left\{\p^e\right\}_{e \in \cc{E}}$. We model these distributions via a parametric model of the form
\begin{align}
\label{eq:model}
X^e = B X^e + \epsilon^e \quad \forall e \in \cc{E},  
\end{align}
where:
\begin{enumerate}[label=\roman*)]
    \item $B \in \bb{R}^{p \times p}$ corresponds to the adjacency matrix of a directed acyclic graph (DAG), that is, it has zeros on the diagonal and is lower triangular up to a permutation of rows and columns.
    \item $\epsilon^e = (\epsilon^e_1, ..., \epsilon^e_p) \sim \cc{N}(\nu^e, \Omega^e)$ are noise terms with $\epsilon^e_i \indep \epsilon^e_j \; \forall i \neq j$, that is, $\Omega^e$ is a diagonal matrix with positive entries on the diagonal.
\end{enumerate}
\ntext{The independence requirement in (ii) corresponds to the assumption of causal sufficiency.} Furthermore, we assume that the observed variables $X^e$ and noise terms $\epsilon^e$ are independent across environments. These assumptions mean that for each environment $e$, the variables $X^e$ follow a structural causal model (SCM) with the same coefficients $B$ across all environments $\cc{E}$. The effect of interventions is modeled by allowing the distributions of the noise terms to change across environments. We say that
\begin{center}
    variable $X_i$ has received an intervention if $\epsilon^e_i \neqdist \epsilon_i^{f}$ for some $e,f \in \mathcal{E}$,
\end{center}
that is, if the distribution of its noise term is different for at least two environments. The distribution of the observed variables in each environment is the multivariate Gaussian distribution, that is, $X^e \sim \p^e$ with $\p^e = \cc{N}(\mu^e, \Sigma^e)$
where
\begin{align}
\label{eq:obs_mean_cov}
    \mu^e = (I-B)^{-1}\nu^e \ \text{ and } \ \Sigma^e = (I-B)^{-1}\Omega^e(I-B)^{-T}.
\end{align}
\noindent
The matrix $B$ is the adjacency matrix of the graph underlying the SCM, which we denote by
\begin{align}
    \label{eq:graph}
    \cc{G}(B) := (V, E) \text{ where } V = \preds \text{ and } E = \{(j \to i) : B_{ij} \neq 0\}.
\end{align}
\noindent
Each distribution $P^e$ satisfies the Markov property with respect to the graph $\cc{G}(B)$ \citep[Proposition 6.31]{pearl2009causality, peters2017elements}.\rmj{ No assumption of faithfulness is made unless otherwise stated.}

\subsection{Equivalent Models}

Throughout the paper, we refer to the term ``model'' as the one from \eqref{eq:model}, which contains a causal structure reflected in the DAG $\cc{G}(B)$ \eqref{eq:graph}. Under our modeling assumptions, we find that different models may entail the same set of distributions $\left\{\p^e\right\}_{e \in \cc{E}}$. Since the distribution over the observed variables is Gaussian, it suffices for two models to entail distributions with the same mean and covariance. Together with \eqref{eq:obs_mean_cov} we arrive at the following definition:

\begin{definition}[distribution equivalent models]
\label{def:d_equiv}
Given a collection of environments $\cc{E}$, we call two models $\big(B,\{\nu^e, \Omega^e\}_{e \in \cc{E}}\big)$ and $\big(\tilde{B},\{\tilde{\nu}^e, \tilde{\Omega}^e\}_{e \in \cc{E}}\big)$ distribution equivalent if, for all $e \in \cc{E}$,
\begin{enumerate}[label=\roman*)]
    \item $(I-B)^{-1}\nu^e = (I-\tilde{B})^{-1}\tilde{\nu}^e$, and
    \item $(I-B)^{-1}\Omega^e(I-B)^{-T} = (I-\tilde{B})^{-1}\tilde{\Omega}^e(I-\tilde{B})^{-T}$.
\end{enumerate}
\end{definition}
\ntext{The definition does not include intervention targets because these are assumed to be unknown. They are part of the graphical characterization in \autoref{s:equivalence}}. The true data-generating model, and importantly, its causal structure, may not be identifiable without additional assumptions. As we will show in the next section, if the collection of environments is a singleton ($\abs{\cc{E}} = 1$), the class of distribution equivalent models is a superset of the Markov equivalence class. Under faithfulness, \ntext{and without additional assumptions---e.g.\ equal-noise variances \citep{peters2014identifiability}---}the distribution equivalent models with the sparsest connectivity matrices are the same as the Markov equivalence class (see \autoref{prop:imec_faithfulness}). Additionally, two immediate results follow from \autoref{def:d_equiv}. First, the class of equivalent models cannot grow as we add additional environments:

\begin{corollary}
\label{corr:envs_superset}
Let $\cc{E}_1, \cc{E}_2$ be collections of environments such that $\cc{E}_1 \subseteq \cc{E}_2$. If two models are distribution equivalent under $\cc{E}_2$, they are also under $\cc{E}_1$.
\end{corollary}

Second, the class does not shrink under interventions that change only the means of the noise terms. \ntext{To see this, consider that for any $(B, \{\nu_e\}_{e \in \cc{E}})$, condition (i) is satisfied by any $(\tilde{B}, \{\tilde{\nu}^e\}_{e \in \cc{E}})$ where $\tilde{\nu}^e =  (I-\tilde{B})(I-B)^{-1}\nu^e$ for all $e \in \cc{E}$. Without additional assumptions on their targets or size, these interventions do not provide additional identifiability.}  This argument does not hold for interventions that affect the noise-term variances, as we impose the constraint that $\tilde{\Omega}^e$ be diagonal.
Thus, from now on we consider the data to be centered; for a model $\big(B,\{\nu^e, \Omega^e\}_{e \in \cc{E}}\big)$, by \eqref{eq:obs_mean_cov} and the fact that $(I-B)$ is always a full-rank matrix, this assumption results in $\nu^e = 0$ for all $e \in \cc{E}$. Therefore, from this point on we will denote a model by a tuple $(B,\{\Omega^e\}_{e \in \cc{E}})$, and by $\left[\left(B,\{\Omega^e\}_{e \in \cc{E}}\right)\right]$ its class of distribution equivalent models.

\subsection{\ntext{A practical discussion about the intervention model}}
\label{ss:intervention_model}


\ntext{Do-, hard or perfect interventions \citep{pearl2009causality, eberhardt2007interventions, peters2017elements} remove the incoming edges of the target, removing the influence of its parents. Thus, they model settings where interventions isolate the target from the effect of other variables in the system. In many settings, such as the real examples in \autoref{s:experiments}, interventions can only hope to perturb a variable, and the assumption above may be unrealistic. Instead, we follow previous work in modeling interventions that \emph{modify} the causal relationships of the target, retaining its parental set \citep{tian2001causal, squires2020permutation, jaber2020causal, taeb2021perturbations}. For the class of SCMs we consider, this can be achieved by either changing the parental coefficients of the target in $B$, or modifying the distribution of the noise term. We leave the former as future work (see \autoref{ss:future_work} for a discussion) and focus instead on the latter type of intervention. Because the noise term of a variable propagates through the causal path, it models the uncertainty arising from unmeasured causes of the variable, instead of measurement noise---which would not propagate downstream. Therefore, changes in the mean and variance of the noise terms can potentially capture various interventions, as highlighted by the real-data experiments in \autoref{s:experiments} (e.g., \autoref{fig:lt}c).}

\section{Graphical Characterization of Distributional Equivalence}
\label{s:equivalence}

In this section, we give a graphical characterization of distribution equivalent models, which we then use in \autoref{s:greedy} to construct an algorithm that performs a greedy search over the space of distributional equivalence classes. The proofs for these results can be found in Appendix \ref{s:proofs}. We refer to a \emph{graphical} characterization because we describe the graphs underlying the models in the distributional equivalence class. More formally, for a class $\left[\big(B,\{\Omega^e\}_{e \in \cc{E}}\big)\right]$ of distribution equivalent models, we provide results describing the set

\begin{align}
\label{eq:g_operator}
\bb{G}\left(B,\{\Omega^e\}_{e \in \cc{E}}\right) :=
   \left\{
      \dag : 
      \exists ( \tilde{B},\{\tilde{\Omega}^e\}_{e \in \cc{E}})
      \in \left[\left(B,\{\Omega^e\}_{e \in \cc{E}}\right)\right]
      \text{ such that } \tilde{B} \comp \dag
    \right\},
\end{align}
where $\tilde{B} \comp \dag$ denotes that the adjacency matrix $\tilde{B}$ satisfies the sparsity pattern of the graph $\dag$, that is, that $(i \to j) \notin \dag \implies \tilde{B}_{ji} = 0$.
Our results are based on the notion of \emph{transition sequence equivalence}, introduced by \citet{tian2001causal}. For notational convenience and to account for the different terminology, we refer to it here as $\cc{I}$-equivalence, where $\cc{I} \subseteq \preds$ corresponds to a set of intervention targets. 

\begin{definition}[$\cc{I}$-equivalence]
\label{def:i_equiv}
Two DAGs $\dag$ and $\dag'$ are said to be \emph{$\cc{I}$-equivalent} under interventions on targets $\cc{I} \subseteq \preds$ if
\begin{enumerate}[label=\roman*)]
    \item they have the same skeleton,
    \item they have the same v-structures, and
    \item for all $i \in \cc{I}$, node $i$ has the same parents in $\dag$ and $\dag'$.
\end{enumerate}
\end{definition}

We denote $\cc{I}$-equivalence between graphs $\dag$ and $\dag'$ by $\dag \sim_\cc{I} \dag'$. Furthermore, we denote by $\imec(\dag)$ the class of graphs $\cc{I}$-equivalent to $\dag$. This nomenclature is motivated by the fact that two graphs that satisfy conditions $(i)$ and $(ii)$ are Markov equivalent \citep{verma1990equivalence}. As such, a class of $\cc{I}$-equivalent DAGs is a subset of their Markov equivalence class, and both classes are the same when $\cc{I} = \emptyset$. Together, conditions $(i)$ and $(iii)$ imply that two $\cc{I}$-equivalent graphs also have the same children for all variables $i \in \cc{I}$. In other words, an intervention orients all the edges around its target, and then possibly more following an application of the Meek rules \citep[Figure 1]{Meek1995CausalIA}.

In the remainder of this section, we give results on the assumptions for which the $\cc{I}$-equivalence class of the graph $\cc{G}(B)$ underlying a model $\left(B,\{\Omega^e\}_{e \in \cc{E}}\right)$ fully describes the set $\bb{G}\left(B,\{\Omega^e\}_{e \in \cc{E}}\right)$ of its distribution equivalent graphs. We use
\begin{align}
\label{eq:i_operator}
\bb{I}\big(\{\Omega^e\}_{e \in \cc{E}}\big) := \left \{j : \exists e,f \in \cc{E} \text{ s.t. } \Omega^e_{jj} \neq \Omega^f_{jj}\right \}
\end{align} 
to denote the indices of variables that have received an intervention in at least one of the environments in $\cc{E}$. Note that when $\abs{\cc{E}}=1$, then $\bb{I}\big(\{\Omega^e\}_{e \in \cc{E}}\big) = \emptyset$. Without additional assumptions, we have that the $\cc{I}$-equivalence class of the data-generating graph $\cc{G}(B)$ is a subset of the class of distribution equivalent graphs:

\begin{restatable}{lemma}{propimec}
\label{prop:imec}
Consider a model $\big(B,\{\Omega^e\}_{e \in \cc{E}}\big)$ with underlying graph $\cc{G}(B)$. Let $\cc{I} := \bb{I}\big(\{\Omega^e\}_{e \in \cc{E}}\big)$ be the indices of variables that have received interventions in $\cc{E}$. Then,
    $$\imec\left(\cc{G}(B)\right) \subseteq \bb{G}\left(B,\{\Omega^e\}_{e \in \cc{E}}\right).$$
\end{restatable}
\noindent
If we additionally assume faithfulness for the data-generating model, its graph $\cc{G}(B)$ has the minimal number of edges\footnote{See \autoref{lemma:min_edges_imec} in Appendix \ref{s:proofs}.} and the $\cc{I}$-equivalence class is a subset of the sparsest graphs in $\bb{G}\left(B,\{\Omega^e\}_{e \in \cc{E}}\right)$:

\begin{restatable}{proposition}{propimecfaithfulness}
\label{prop:imec_faithfulness}
Consider a model $\big(B,\{\Omega^e\}_{e \in \cc{E}}\big)$ resulting in a set of distributions which are faithful with respect to its underlying graph $\cc{G}(B)$. Let $\cc{I} := \bb{I}\big(\{\Omega^e\}_{e \in \cc{E}}\big)$ be the indices of variables that have received interventions in $\cc{E}$. Then,
$$
    \imec\left(\cc{G}(B)\right) \subseteq \argmin_{\dag \in \bb{G}\left(B,\{\Omega^e\}_{e \in \cc{E}}\right)} \|\dag\|_{0},
$$
where $\|\dag\|_{0}$ is the number of edges in $\dag$.
\end{restatable}
Note that the set $\bb{G}\left(B,\{\Omega^e\}_{e \in \cc{E}}\right)$ contains all graphs obtained by adding edges to $\cc{G}(B)$ (see \autoref{lemma:imaps_idec} in Appendix \ref{s:proofs}). This motivates considering the sparsest graphs, that is, the most succinct causal explanations which (under faithfulness) comprise the true graph and its equivalents. Without additional assumptions on the parameters of the interventions, \autoref{prop:imec_faithfulness} is the best one can do. To see this, consider a data-generating model $(B, \{\Omega^e, \Omega^f\})$ with two interventional environments $e$ and 
 $f$, and suppose that $\bb{I}\big(\{\Omega^e, \Omega^f\}) = [p]$, that is, there has been an intervention on all variables. For $\cc{I} = \bb{I}\big(\{\Omega^e, \Omega^f\})$, it follows from \autoref{def:i_equiv} that $\imec\left(\cc{G}(B)\right)$ is a singleton. However, if the interventions simply scale all noise terms by the same factor, that is, $\forall i \;\Omega^e_{ii}/\Omega^f_{ii} = \alpha$, then the set $\argmin_{\dag \in \bb{G}\left(B,\{\Omega^e, \Omega^f\}\right)} \|\dag\|_{0}$ still contains the complete Markov equivalence class of $\cc{G}(B)$, which is in general not a singleton. To guard against such pathologically selected intervention parameters, we introduce an assumption which we call \emph{intervention-heterogeneity}:

\begin{assumption}[Intervention-heterogeneity]
\label{assm:int_heter}
    A model $\big(B,\{\Omega^e\}_{e \in \cc{E}}\big)$ is said to satisfy \emph{intervention-heterogeneity} if for every pair of environments $e,f \in \cc{E}$, it holds that
    $$\frac{\Omega^e_{ii}}{\Omega^f_{ii}} \neq \frac{\Omega^e_{jj}}{\Omega^f_{jj}} \text{ for all } i\neq j \in \bb{I}\big(\{\Omega^e, \Omega^f\}\big).$$
\end{assumption}
In words, this assumption means that the noise-term variances which change between environments do so by a different factor. This is trivially satisfied when environments contain interventions on a single target. Furthermore, if the intervention variances are independently sampled from an absolutely continuous distribution, violations of \autoref{assm:int_heter} have zero probability. To fully characterize the equivalence class, we need the following additional assumption.

\begin{restatable}[Model truthfulness]{assumption}{modeltruth}
\label{assm:model_truth}
    Let $\big(B,\{\Omega^e\}_{e \in \cc{E}}\big)$ be a faithful model. Then, for any equivalent model $(\tilde{B},\{\tilde{\Omega}^e\}_{e \in \cc{E}})$, it holds that $(I-\tilde{B})(I-B)^{-1}$ has non-zero diagonal entries.
\end{restatable}

\noindent
We conjecture that this technical assumption directly holds as a consequence of faithfulness and the modeling assumptions. A violation would imply that, given some conditioning set, the conditional variance of a variable is not a function of the variance of its noise term. Furthermore, we provide additional empirical support by showing that it holds for all the random SCMs generated in our synthetic experiments. These results and a detailed discussion can be found in Appendix \ref{s:assm}.

Under faithfulness and these additional assumptions, we have that the $\cc{I}$-equivalence class is the set of sparsest distribution equivalent graphs:

\begin{restatable}[]{theorem}{propimecfull}
\label{prop:imec_full}
Consider a model $\big(B,\{\Omega^e\}_{e \in \cc{E}}\big)$ satisfying Assumptions \ref{assm:int_heter} and \ref{assm:model_truth} and resulting in a set of distributions which are faithful with respect to its underlying graph $\cc{G}(B)$. Let $\cc{I} := \bb{I}\big(\{\Omega^e\}_{e \in \cc{E}}\big)$ be the indices of variables that have received interventions in $\cc{E}$. Then,
$$\imec\left(\cc{G}(B)\right) = \argmin_{\dag \in \bb{G}\left(B,\{\Omega^e\}_{e \in \cc{E}}\right)} \|\dag\|_{0}.$$
\end{restatable}

\subsection{Relation to other Equivalence Classes}
\label{ss:other_classes}

\rmj{We compare the $\cc{I}$-equivalence class implied by \autoref{def:i_equiv} to the interventional equivalence classes introduced by \citet{hauser2012characterization} and \citet{yang2018characterizing}. Incidentally, both papers refer to their respective equivalence classes as the ``interventional Markov equivalence class'' and denote it by $\imec$. We modify their notation to avoid conflicts. The proofs for the lemmas of this section can be found in Appendix \ref{s:proofs_other_classes}.}
\ntext{Our graphical characterization is equivalent to the \emph{transition sequence equivalence} introduced by \citet[Theorems 2 and 3]{tian2001causal}, which describes how the Markov equivalence class shrinks under \emph{single-target} interventions which do not alter the parental set of the target \citep[Definition 1]{tian2001causal}. Assumptions \ref{assm:int_heter} and \ref{assm:model_truth} amount to sufficient conditions under which this characterization also holds for our parametric setting (c.f. \autoref{eq:model}) and \emph{multiple-target} interventions. Intuitively, the assumptions ensure that the effects of interventions do not cancel out, allowing us to simply take the union of all targets (c.f. \autoref{eq:i_operator}) for our equivalence class. However, further assumptions \citep[e.g.,][Definition 2]{tian2001causal} may be needed for the general non-linear models treated in \citet{hauser2012characterization, yang2018characterizing, jaber2020causal}.

To compare our $\cc{I}$-equivalence class with the interventional equivalence classes introduced by the works mentioned above, we first show that $\cc{I}$-equivalence can also be expressed in terms of Markov equivalence of augmented graphs. In our case, we construct the augmented graph as follows:

\begin{definition}[Augmented graph---this work]
\label{def:augmented}
Given a graph $\dag$ over $p$ variables, and a set of targets $\cc{I} \subseteq [p]$, the \emph{augmented graph} $\dag^\cc{I}$ is obtained by adding a new node $F_i$ and an edge $F_i \to i$ for all $i \in \cc{I}$.
\end{definition}%
An example is shown in \autoref{fig:equiv_classes}. We can now express $\cc{I}$-equivalence as Markov equivalence between these augmented graphs, as stated in \autoref{lemma:augmented_graphs} below. The proof can be found in Appendix \ref{s:proofs_other_classes}.

\begin{restatable}[]{lemma}{lemmaaugmented}
\label{lemma:augmented_graphs}
Let $\dag_1$ and $\dag_2$ be two DAGs over $p$ variables and let $\cc{I} \subseteq [p]$ be a set of intervention targets. The following statements are equivalent:
\begin{enumerate}
    \item $\dag_1$ and $\dag_2$ are $\cc{I}$-equivalent, and
    \item their augmented graphs $\dag_1^\cc{I}$ and $\dag_2^\cc{I}$ are Markov equivalent.
\end{enumerate}
\end{restatable}

\begin{figure}[h]
\centering
\includegraphics[width=152mm]{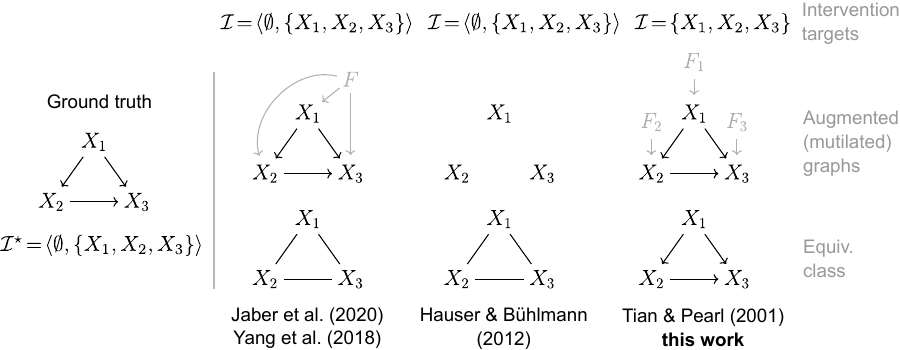}
\caption{\ntext{\comm{[new figure] }Comparison of the equivalence classes obtained for the ground-truth graph and intervention targets on the left for the different definitions given in \citet{tian2001causal, hauser2012characterization, yang2018characterizing, jaber2020causal} and this work. The works differ in their treatment of the intervention targets and the resulting equivalence classes. For the given example, \citet{jaber2020causal} and \citet{yang2018characterizing} result in the same augmented graph and equivalence class, although they are not generally the same \citep[see ][Appendix D.1]{jaber2020causal}.}}
\label{fig:equiv_classes}
\end{figure}

Using this alternative formulation, we compare in \autoref{fig:equiv_classes} the resulting equivalence classes for a concrete example under the various definitions. This highlights how Assumptions \ref{assm:int_heter} and \ref{assm:model_truth} yield further identifiability by making our setting equivalent to the single-target interventions treated in \citet{tian2001causal}. Further details and results for the related equivalence classes are given in Sections \autoref{sss:hauser} to \autoref{sss:jaber} below. The proofs for the lemmas of this section can be found in Appendix \ref{s:proofs_other_classes}.
}

\subsubsection{\citet{hauser2012characterization}}
\label{sss:hauser}
The authors present an interventional equivalence class for the case of \emph{hard} interventions \citep[Section 3.2.2]{pearl2009causal}. Also called \emph{do-} or \emph{perfect}  interventions, these isolate the intervened variable from the effect of its causal parents, effectively removing the incoming edges in the causal graph. One key difference to our formalism is how the intervention targets are represented. While we do so by means of a subset $\cc{I} \subseteq [p]$ of the observed variables, \citeauthor{hauser2012characterization} consider instead a set of subsets, which they refer to as a \emph{family} of targets. They require that this family $\cc{H}$ be \emph{conservative}, meaning that for all $j \in [p]$, there exists $I \in \cc{H}$ such that $j \notin I$ \citep[Definition 6]{hauser2012characterization}. Note that in the presence of an observational environment, that is, $\emptyset \in \cc{H}$, $\cc{H}$ is conservative. To avoid confusion with our notation, we denote their equivalence class by $\hmec$. It arises from the following equivalence relation:

\begin{proposition}[Theorem 10 from \citealp{hauser2012characterization}]
\label{thm:hauser}
Let $\dag_1$, $\dag_2$ be two DAGs and let $\cc{H}$ be a conservative family of targets. Then, $\dag_1$, $\dag_2$ are equivalent if and only if (i) $\dag_1$ and $\dag_2$ have the same skeleton and the same v-structures, and (ii) $\dag_1^{(H)}$ and $\dag_2^{(H)}$ have the same skeleton for all $H \in \cc{H}$, where $\dag_i^{(H)}$ denotes the graph obtained by removing from $\dag_i$ all edges incoming to the nodes in $H$.
\end{proposition}%
The $\cc{I}$-equivalence class is, in fact, a subset of the $\hmec$. This matches the intuition that a hard intervention on two adjacent nodes renders the causal effect between them unidentifiable.

\begin{restatable}[]{proposition}{lemmahauser}
\label{lemma:hauser}
Consider a DAG $\dag$ and a conservative family of targets $\cc{H}$. Let $\cc{I} := \cup_{H \in \cc{H}} H$. Then $\imec(\dag) \subseteq \hmec(\dag).$
\end{restatable}

\subsubsection{\citet{yang2018characterizing}}
The authors give an interventional equivalence class for the case of \emph{general} interventions \citep[Definition 3.3]{yang2018characterizing}. These can be understood as all interventions for which the resulting interventional distribution is still Markov with respect to the data-generating graph. These exclude interventions that add new parents to a variable or affect several variables in a way that induces a new statistical dependence between them, for example, by adding a hidden confounder. Both the hard interventions from \citet{hauser2012characterization} and the noise interventions which we discuss in our paper naturally fall into the category of general interventions. The authors provide a graphical characterization of their equivalence class through the concept of an ``interventional DAG''. Given a graph $\dag$ and a family of targets $\cc{Y}$, they construct it as follows: for every $Y \in \cc{Y} \setminus \{\emptyset\}$ we add a source node to $\dag$, with edges from it to the nodes in $Y$. Then, under the condition that $\emptyset \in \cc{Y}$, two DAGs are $\cc{Y}$-equivalent if and only if their interventional graphs have the same skeleton and v-structures \citep[Theorem 3.9]{yang2018characterizing}. We denote the class of $\cc{Y}$-equivalent models by $\ymec$. Since hard interventions are also a special case of general interventions, \ntext{under causal sufficiency }we expect that our $\cc{I}$-equivalence class is also a subset of the $\cc{Y}$-equivalence class. We state this result in \autoref{lemma:yang} and provide a proof in Appendix \ref{s:proofs_other_classes}.

\begin{restatable}[]{proposition}{lemmayang}
\label{lemma:yang}
Consider a DAG $\dag$ and let $\cc{Y}$ be a family of targets such that $\emptyset \in \cc{Y}$, and let $\cc{I} := \cup_{Y \in \cc{Y}} Y$. Then $\imec(\dag) \subseteq \ymec(\dag).$
\end{restatable}


\ntext{
\subsubsection{\citet{jaber2020causal}}
\label{sss:jaber}

The authors provide a characterization of interventional equivalence ($\Psi$-Markov equivalence) for the case of soft interventions and violations of causal sufficiency, i.e., hidden confounders. The interventions they consider include those in our model class, and for the case where causal sufficiency is satisfied, the authors provide a graphical characterization which directly applies to our setting (\autoref{lemma:jaber}).

\begin{lemma}[Corollary 1 from \citealp{jaber2020causal}]
\label{lemma:jaber}
Given causal graphs without latents $G_1, G_2$ and the corresponding interventional targets $\cc{I}_1, \cc{I}_2$, the pairs $(G_1, \cc{I}_1),  (G_2, \cc{I}_1)$ are $\Psi$-Markov equivalent if and only if their augmented graphs the same skeleton and v-structures.
\end{lemma}%
The augmented graph is constructed as follows: given a causal graph $G = (V, E)$ a set of intervention targets $\cc{I} = \langle I_1, \ldots, I_k\rangle$ with $I_k \subset V$, for each pair of targets $I_i, I_j \in \cc{I}$ we add an additional node to the graph, and edges from it to all nodes $v \in \{w \mid (w \in I_i \land w \notin I_j) \lor (w \notin I_i \land w \in I_j) \}$, that is, nodes which appear in only one of the two targets $I_i$ and $I_j$. An example is shown in \autoref{fig:equiv_classes}.

As for the equivalence class in \citet{yang2018characterizing}, $\cc{I}$-equivalence and $\Psi$-Markov equivalence are the same in the case of single-target interventions (c.f. \autoref{lemma:augmented_graphs} and \autoref{lemma:jaber}). However, under interventions that affect several targets in the same environment, the resulting equivalence classes are different; assumptions \ref{assm:int_heter} and \ref{assm:model_truth} allow treating all interventions as single-target interventions, yielding additional identifiability (e.g., \autoref{fig:equiv_classes}) in the parametric setting we consider.
}
\section{Learning the Class of Equivalent Models}
\label{s:greedy}

Next, we consider the problem of recovering the class of distribution equivalent models
from data.
In particular, we are interested in recovering the causal structure, that is, the graphs underlying each of the models in the distributional equivalence class of the data-generating model.

As a first approach, one might consider an $\ell_0$-penalized maximum likelihood estimator of the form
\begin{equation}
\begin{aligned}
        \left(\hat{B}, \{\hat{\Omega}^e\}_{e \in \cc{E}}\right) \in
    \argmax_{\substack{
    B \comp \text{DAG}\\    
    \text{diagonal pos. def. } \{\Omega^e\}_{e \in     
    \cc{E}}
    }}
    \;
    \sum_{e \in \cc{E}}
    \log \p(\boldsymbol{X}^e; B, \Omega^e)    
    - \lambda \norm{\cc{G}(B)}_0,
\end{aligned}
\label{eq:mle_plain}
\end{equation}
where $\boldsymbol{X}^e$ is the sample from environment $e$, $\lambda > 0$ is a regularization parameter and $\norm{\cc{G}(B)}_0$ is the number of edges in the graph entailed by $B$, which we have restricted to be the connectivity matrix of a DAG. \ntext{The expression for the log-likelihood can be found in \autoref{eq:covariance_reduction} in Appendix~\ref{s:proofs_score_properties}.}
This approach suffers from two important drawbacks. 

Computing the estimator in \eqref{eq:mle_plain} is a highly complex task. This stems from the fact that the DAG constraint over $B$ causes the optimization to be highly non-convex. Since the space of DAGs grows super-exponentially with the number of variables, an exhaustive search is infeasible even for a few variables. Even for the observational setting where $\abs{\cc{E}}=1$, a greedy search over this space, that is, starting with an empty graph and greedily adding and removing edges, will often only find local minima and is known to perform poorly \citep{chickering2002optimal, hauser2012characterization}. This constitutes the first drawback.

As a possible way forward, one may note that for $\abs{\cc{E}}=1$, \autoref{eq:mle_plain} corresponds to the score employed by \citet{chickering2002optimal} in the GES algorithm. In this case, the class of distribution equivalent models corresponds to the Markov equivalence class\footnote{Under faithfulness, see \autoref{prop:mec_equals_dec} in Appendix \ref{s:proofs}.}, and the objective in \eqref{eq:mle_plain} has some desirable properties which are exploited by the GES algorithm to perform a greedy search over the space of Markov equivalence classes. By iterating over this alternative search space---instead of DAGs---GES is able to escape local minima and return the true equivalence class in the large sample limit. This is possible due to the key property of \emph{score equivalence}: for any finite sample, all members of an equivalence class attain the same score.

In the same manner as GES, we would like to greedily iterate over the space of $\cc{I}$-equivalence classes until we arrive at the one containing the data-generating model. Unfortunately, for the setting where $\abs{\cc{E}} > 1$, the property of \emph{score equivalence} is not satisfied by \eqref{eq:mle_plain}; $\cc{I}$-equivalent graphs may attain different scores on the same finite sample\ntext{. See Appendix~\ref{s:score_equivalence_violation} for a discussion}. Motivated by this second drawback, we propose the following score.

\subsection{Score Function}
For a given DAG $\dag$ and a set of intervention targets $\cc{I}$, we consider the score
\begin{align}
    \label{eq:score}
    S(\dag, \cc{I}) = 
    \max_{\substack{
    B \comp \dag\\
    \text{diagonal pos. def. } \{\Omega^e\}_{e \in \cc{E}}\\
    \text{s.t. }\bb{I}(\{\Omega^e\}_{e \in \cc{E}}) \subseteq \cc{I}
    }}
    \quad
    \sum_{e \in \cc{E}}
    \log \p(\boldsymbol{X}^e; B, \Omega^e) - \lambda\text{DoF}(\dag, \cc{I}),
\end{align}
where $\boldsymbol{X}^e$ is the sample from environment $e$. In words, the constraints placed over the parameters mean that $B$ respects the adjacency of the DAG $\dag$, and that the variances of the noise terms remain constant across environments, with the exception of those of the variables in $\cc{I}$. The penalization term $\text{DoF}(\dag, \cc{I}) := \norm{\dag}_{0}+p+|\cc{I}|(|\cc{E}|-1)$ corresponds to the number of free parameters in the model, that is, the number of edges and distinct noise-term variances which need to be estimated.
\newline
\newline
To develop an efficient and accurate algorithm for identifying the equivalence class of DAGs that best fit the data, we want the score function to satisfy the following properties:
\begin{enumerate}[label=\roman*)]
\item \emph{Score equivalence}: for any finite sample, all members of an equivalence class attain the same score, that is, $\dag \sim_\cc{I} \dag' \implies S(\dag,\cc{I}) = S(\dag',\cc{I})$.
\item \emph{Decomposability}: the score can be written as a sum of terms depending only on a variable and its parents.
\item \emph{Consistency}: let $\dag^\star$ be the data-generating DAG and let $\cc{I}^\star$ be the true set of intervention targets. Let
$\{(\dag_1, {\cc{I}}_1),..., (\dag_l, {\cc{I}}_l)\} = \argmin_{\dag,\cc{I}}  S(\dag,\cc{I})$
be the pairs of DAGs and intervention targets that maximize the score $S(\dag,\cc{I})$. In the large sample limit, where $n^e \to \infty$ for every $e \in \mathcal{E}$ and $\lambda$ is chosen appropriately, ${\cc{I}}_j = \cc{I}^\star$ for all $j \in [l]$, and $\{\dag_j\}_{j\in [l]} = \simec(\dag^\star)$ with probability tending to one.
\end{enumerate}

Indeed, the score \eqref{eq:score} satisfies these properties. Score equivalence is attained through the constraint $\bb{I}(\{\Omega^e\}_{e \in \cc{E}}) \subseteq \cc{I}$, that only the noise-term variances of the targets $\cc{I}$ can vary across environments. The proof is presented in Appendix \ref{s:proofs_score_properties}.

 \begin{restatable}[Score properties]{proposition}{propscoreproperties}
 \label{prop:score_properties} Without additional assumptions, the score \eqref{eq:score} satisfies score equivalence \emph{(i)} and decomposability \emph{(ii)}. Under faithfulness and Assumptions \ref{assm:int_heter} and \ref{assm:model_truth} for the data-generating model, and the assumptions that ${\lim_{\mbox{all}\ n_e \to \infty}} \frac{n^e}{\sum_{e \in \mathcal{E}}n^e} > 0$ for all $e \in \mathcal{E}$ and that the parameter space is compact, the score is consistent \emph{(iii)}.
 \end{restatable}
\noindent

\subsection{Our Greedy Algorithm GnIES}

In the remainder of this section, we present the \emph{greedy noise-interventional equivalence search} algorithm (GnIES) to estimate the equivalence class of the data-generating graph under unknown interventions. The algorithm is score-based and composed of two nested, greedy procedures. The inner one proceeds similarly to GES and searches for the optimal equivalence class given a fixed set of intervention targets. The outer procedure greedily searches the space of such intervention targets; its output constitutes the GnIES estimate.

\subsubsection{Inner Procedure}

For a fixed set of intervention targets $\cc{I}$, the inner procedure performs a search over the space of $\cc{I}$-equivalence classes, returning the highest-scoring one. The procedure is a modification of the GES algorithm, which we will now describe at a high level with the purpose of illustrating said modification.

As discussed in the previous section, GES performs a search over the space of Markov equivalence classes to find the one that best fits the data, that is, whose representatives maximize the score of choice. This search is carried out in a greedy manner: starting with the equivalence class of the empty graph, the algorithm considers as neighboring classes those whose representatives only have one more edge than those in the current class, and moves to the one which yields the highest score. This ``forward-step'' is repeated until no transition yields an improvement in the score. With the resulting class as a new starting point, the algorithm then performs the opposite ``backward-steps'', iteratively transitioning to classes with fewer edges until the score can no longer be improved. The resulting class is returned as the estimate, which is, quite remarkably, consistent \citep[Section 4]{chickering2002optimal}.

\begin{algorithm}[h]
\caption{Greedy inner procedure of GnIES}
\begin{algorithmic}[1]
\vspace{0.1in}
\STATE {\bf Input}: data $\boldsymbol{X}^e$ across environments $e \in \mathcal{E}$, intervention set $\cc{I}$.
\vspace{0.05in}
\STATE {\bf Set starting point}: set the current CPDAG $C := \dag_0$ to the empty graph and the current score $s := S(\dag_0, \cc{I})$ to its score.
\vspace{0.05in}
\STATE {\bf Forward phase}: Apply the highest scoring GES \emph{insert} operator until the score cannot be improved:
\begin{enumerate}
    \item[(a)] Compute the change in score for all valid GES \emph{insert} operators which can be applied to the current CPDAG $C$.
    \item[(b)] If the score cannot be improved, end the forward-phase. Otherwise, apply the highest scoring one to $C$, resulting in the PDAG $P$ and updating the score $s$.
    \item [(c)] $C := \texttt{completion\_algorithm}(P, \cc{I})$ and repeat steps (a-c).
\end{enumerate}
\STATE {\bf Backward phase}: Apply the highest scoring GES \emph{delete} operator until the score cannot be improved:
\begin{enumerate}
    \item[(a)] Compute the change in score for all valid GES \emph{delete} operators which can be applied to the current CPDAG $C$.
    \item[(b)] If the score cannot be improved, end the backward-phase. Otherwise, apply the highest scoring operator to $C$, resulting in the PDAG $P$ and updating the score $s$.
    \item [(c)] $C := \texttt{completion\_algorithm}(P, \cc{I})$ and repeat steps (a-c).
\end{enumerate}
\STATE{\bf Output:} The estimated CPDAG $C$\ntext{ and its score $s$}.
\end{algorithmic} \label{algo:inner}
\end{algorithm}

Internally, GES represents an equivalence class by means of a complete, partially directed acyclic graph (CPDAG) \citep[Section 2.4]{chickering2002optimal}. Transitions to neighboring classes are implemented as the application of an operator to the CPDAG representing the current equivalence class. Forward steps constitute an application of the so-called \emph{insert} operator, and backward steps of the \emph{delete} operator, respectively adding and removing edges to the CPDAG. Applying an operator results in a PDAG, or partially directed acyclic graph, representing some of the neighboring equivalence class members. This PDAG is then ``completed'' using a completion algorithm to transform it into the CPDAG representing the complete class. 

The inner procedure of GnIES (Algorithm \ref{algo:inner}) is exactly GES with the exception of two components: the score \eqref{eq:score} and the completion algorithm. All other components, including the operators, remain the same as for GES. The modified completion algorithm takes the PDAG resulting from an operator and the given set of intervention targets, and returns the CPDAG representing the $\imec$,
instead of the CPDAG representing the (observational) Markov equivalence class. We describe this completion procedure in Algorithm \ref{algo:completion}, Appendix \ref{s:completion_algorithm}.

\subsubsection{Outer Procedure}

In Algorithm \ref{algo:outer} we describe the outer component of GnIES, which estimates the list of intervention targets from data and returns the estimated equivalence class.

\begin{algorithm}[h]
\caption{GnIES: Equivalence class recovery with unknown targets}
\begin{algorithmic}[1]
\vspace{0.1in}
\STATE {\bf Input}: data $\boldsymbol{X}^e$ for $e \in \mathcal{E}$
\vspace{0.05in}
\STATE {\bf Compute original score}: initialize the current intervention set $\cc{I}_\text{curr} := \emptyset$ and compute its score $s_c:=S(\cc{I}_\text{curr}) := \max_{\dag}S(\dag,\cc{I}_\text{curr})$ using Algorithm \ref{algo:inner}.
\vspace{0.05in}
\STATE {\bf Greedy addition of target interventions}: add intervention targets until score cannot be improved
\begin{enumerate}
    \item[(a)] add the intervention target $i = \arg\max_{j}S(\cc{I}_\text{curr} \cup \{j\})$ that maximizes the score and let $s_\text{next}$ be the resulting score.
    \item[(b)] if $s_\text{next} > s_\text{curr}$, set $s_\text{curr} := s_\text{next}$, $\cc{I}_\text{curr} := \cc{I}_\text{curr} \cup \{i\}$ and repeat steps (a-b); otherwise, stop.
\end{enumerate}
\STATE {\bf Greedy deletion of intervention targets}: delete intervention targets until score cannot be improved
\begin{enumerate}
    \item[(a)] remove the intervention target $i = \arg\max_{j \in \cc{I}_\text{curr}}S(\cc{I}_\text{curr} \setminus \{j\})$ which maximizes the score and let $s_\text{next}$ be the resulting score.
    \item[(b)] if $s_\text{next} > s_\text{curr}$, set $s_\text{curr} := s_\text{next}$, $\cc{I}_\text{curr} := \cc{I}_\text{curr} \setminus \{i\}$ and repeat steps (a-b); otherwise, stop. 
\end{enumerate}
\STATE {\bf Estimated intervention targets} $\hat{\cc{I}} := \cc{I}_\text{curr}$
\STATE{\bf Output:} return the intervention targets $\hat{\cc{I}}$ and the resulting equivalence class $\hatimec(\argmax_{\dag}\mathcal{S}(\dag,\hat{\cc{I}}))$, \ie via Algorithm \ref{algo:inner}.
\end{algorithmic} \label{algo:outer}
\end{algorithm}
Similar to the inner procedure, the outer component of GnIES proceeds greedily: we start from the empty set and greedily add targets until the regularized maximum likelihood score \eqref{eq:score} cannot be improved. Then, as in the inner procedure, we move backward by greedily removing an intervention target until the score can no longer be improved. The score is obtained by supplying the data and the set of interventions to the inner procedure (described in Algorithm \ref{algo:inner}).
Background knowledge, in terms of partially known intervention targets, can be incorporated by keeping them in the estimate $\cc{I}_\text{curr}$ at every step of Algorithm \ref{algo:outer}. This functionality is made available through the parameter \texttt{known\_targets} in the documentation of the \href{https://github.com/juangamella/gnies}{\texttt{gnies}} package.

\subsubsection{\ntext{Consistency of GnIES}}

\ntext{
Under the assumption that the inner procedure is consistent, the outer procedure is also consistent. That is, as $n^e \to \infty$ for every $e \in \mathcal{E}$ and $\lambda$ chosen appropriately, the outer procedure returns $\cc{I}^\star$ as the intervention targets and ${\cc{I}}^\star\text{-MEC}(\dag^\star)$ as the equivalence class of graphs with probability tending to one. Formally, for each $\mathcal{I} \subseteq \{1,2,\dots,p\}$ and DAG $\dag$ among $p$ nodes, let
\begin{align*}
    S^\star(\dag, \cc{I}) = 
    \max_{\substack{
    B \comp \dag\\
    \text{diagonal pos. def. } \{\Omega^e\}_{e \in \cc{E}}\\
    \text{s.t. }\bb{I}(\{\Omega^e\}_{e \in \cc{E}}) \subseteq \cc{I}
    }}
    \quad
    \sum_{e \in \cc{E}}
    \mathbb{E}[\log \p(\boldsymbol{X}^e; B, \Omega^e)]
\end{align*}
to be the population analog of $S(\dag,\cc{I})$ in \eqref{eq:score}.
\begin{restatable}[]{theorem}{thmouter}
\label{thm:outer_procedure}
Suppose assumptions of Proposition~\ref{prop:score_properties} are satisfied. Furthermore, suppose that the inner procedure in Algorithm~\ref{algo:inner} is consistent, i.e., for every $\mathcal{I} \subseteq \{1,2,\dots,p\}$, it returns ${\cc{I}}\text{-MEC}(\argmax_{\dag}S^\star(\dag,\cc{I}))$ as $n^e \to \infty$ for every $e \in \mathcal{I}$. Denote by $\hat{\cc{I}}_f$ the intervention set resulting from the forward-phase of algorithm~\ref{algo:outer}, and suppose  $\cc{I}^\star \subseteq \hat{\cc{I}}_f$. Then, the output of Algorithm~\ref{algo:outer} is consistent.
\end{restatable}
The proof is presented in \autoref{s:proof_outer}. Theorem~\ref{thm:outer_procedure} guarantees that the GnIES estimate is consistent, assuming that the inner algorithm \ref{algo:inner} is accurate and the greedy addition step (3) of algorithm \ref{algo:outer} yields a superset of the true intervention targets. In any case, one could forego this first step altogether by starting from the full set of intervention targets $\cc{I} = \{1,...,p\}$, albeit paying a price in terms of computation time when the true targets are few.

We do not provide a proof for the consistency of the inner procedure, which is GES with a modified score and completion algorithm. For plain GES, the proof consists in first showing that the score is locally consistent and decomposable, and then showing that the insert and delete operators reach all neighboring classes \citep[Theorems 15 and 17]{chickering1995}. While we show in \autoref{prop:score_properties} that our score satisfies the required properties, we did not extend the results on the operators to our completion algorithm and leave this as future work.


}

\subsubsection{Computational Cost}

The computational cost of a single run of the inner procedure is that of GES, saving the nominal cost incurred by the modified completion algorithm and score. The worst-case complexity of GES is polynomial in the number of variables and exponential on some measures of the underlying graph, such as the maximum number of parents \citep{chickering2015selective}. However, both the experiments carried out in this and the original GES paper \citep{chickering2002optimal} suggest that, in practice, GES is much more efficient than suggested by this worst-case complexity.

The greedy outer procedure of the algorithm, as described in Algorithm \ref{algo:outer}, incurs a worst-case complexity of $\cc{O}(p^2 \times \text{cost of GES})$. As a faster alternative, we propose a ``ranking-based'' outer procedure whose complexity is instead $\cc{O}(p \times \text{cost of GES})$. Under this approach, we first run the inner procedure with interventions on all variables and obtain an ordering of variables based on the variance of their noise-term variance estimates. We then add (remove) targets in increasing (decreasing) order until the score does not improve. At the price of a slightly reduced performance in terms of edge recovery, the ranking procedure yields a substantial speed-up (see \autoref{fig:performance} in Appendix \ref{s:figs}).

The computation of the GnIES score for interventional data is slightly more involved than for the observational equivalent used in GES, as it requires an alternating optimization routine to find the maximum likelihood estimate in \eqref{eq:score}. However, as in GES, we exploit the decomposability of the score for substantial computational savings. Since the contribution of a node to the score remains constant if its parents do not change, we can cache these ``local scores'' and avoid many unnecessary computations across the steps of the inner procedure of GnIES, resulting in a dramatic speed-up. Furthermore, the cache can be preserved across steps of the outer procedure, as only the local scores of nodes that are added to or removed from the intervention targets need to be recomputed. This results in further savings.

\begin{figure}[H]
\centering
\includegraphics[width=0.6\textwidth]{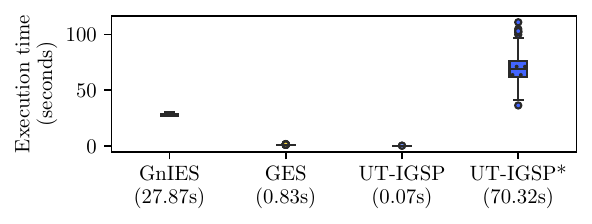}
\caption{\ntext{\comm{[new figure] }Execution times of GnIES and additional baselines on the light-tunnel dataset from \autoref{ss:light_tunnel} ($\mathcal{I}^\star = \{\tilde{\theta}_1, \tilde{\theta}_2\}$). We report the execution times over 100 bootstrap samples of the data, and show the average time in parenthesis under each method. The regularization parameters of each method correspond to the ones shown in \autoref{fig:lt}b.}}
\label{fig:lt_times}
\end{figure}\vspace{-1.5em}%
\ntext{To put the computational cost of GnIES in context with other methods, we compare it in \autoref{fig:lt_times} to the baselines used in the real-data setting of \autoref{ss:light_tunnel}: GES \citep{chickering2002optimal}, UT-IGSP \citep{squires2020permutation} with partial correlation tests, and UT-IGSP with with KCI kernel tests \citep{zhang2011kernel}, denoted by UT-IGSP*. Due to the additional speedups described above, GnIES is about 30 times slower than GES, which is below the time expected by the worst-case complexity given for $p=10$ in this setting. While the method is faster than UT-IGSP with KCI kernel tests, it is much slower than the efficient implementation of UT-IGSP for partial correlation tests \citep{squires2020efficient}.}


\section{Experimental Results}
\label{s:experiments}

Throughout this section, we evaluate the performance of GnIES \rmj{in recovering the equivalence class of the data-generating graph} across a variety of settings\footnote{The code to reproduce the experiments can be found at \href{https://github.com/juangamella/gnies-paper}{\texttt{github.com/juangamella/gnies-paper}}.}.\ntext{ These include synthetic data (\autoref{ss:sim}) and two real-data settings consisting of biological data (\autoref{ss:real}) and data from a controlled physical system (\autoref{ss:light_tunnel}).}

\rmj{As such, t}\ntext{T}he estimates produced by GnIES and \ntext{some of }the baseline methods will be a set of graphs, whose closeness with the \ntext{ground-truth graph or}\rmj{true} equivalence class needs to be measured. We use the metrics introduced in \citet{taeb2021perturbations} for this purpose, which we restate here:

\begin{definition}[Metrics]
\label{def:metrics}
Let $\bb{G}^\star$ be the true $\imec$ and $\hat{\bb{G}}$ its estimate. We define the \emph{true (false) discovery proportion} respectively as
\begin{align*}
    \text{\emph{TDP} : } &\max_{\dag^\star \in \bb{G}^\star} \min_{\hat{\dag} \in \hat{\bb{G}}} \absb{\text{edges in }\dag^\star\text{ and }\hat{\dag}} / \absb{\text{edges in }\dag^\star},\\
    \text{\emph{FDP} : } &\max_{\hat{\dag} \in \hat{\bb{G}}} \min_{\dag^\star \in \bb{G}^\star} \absb{\text{edges in }\hat{\dag}\text{ not in }\dag^\star} / \absb{\text{edges in }\hat{\dag}}.
\end{align*}
\end{definition}

In the case where both the estimate and the truth are singletons, the TDP and FDP correspond to the true and false positive rates in terms of edges. The estimate and truth are equal if and only if FDP = 0 and TDP = 1.

\subsection{Synthetic Data}
\label{ss:sim}

We evaluate the performance of GnIES on simulated data and compare it to other algorithms that can be applied to the same problem setting. \rmj{We consider settings where the data-generating process satisfies all model assumptions (\autoref{fig:model_match}), and where there is a model mismatch in the nature of the interventions (\autoref{fig:model_mismatch}).}\ntext{ Additional results for variations of the generation settings and a misspecified intervention model are given in Appendix \ref{s:additional_experiments}}.

\textbf{Data generation.} For \ntext{the results shown in \autoref{fig:model_match},}\rmj{Figures \ref{fig:model_match} and \ref{fig:model_mismatch}} the data is sampled from 100 randomly generated linear Gaussian SCMs and intervention targets\ntext{, following the model in \autoref{eq:model}}. The underlying DAGs are Erd{\"o}s-Renyi graphs with $p=10$ nodes and an average degree of $2.7$. \ntext{We provide results for denser graphs in Appendix \ref{s:additional_experiments}.} The edge weights are sampled uniformly at random from $[0.5,1]$ to bound them away from zero; the variances of the noise terms are sampled from $[1,2]$ and their mean is set to zero. We generate $10$ different data sets from each SCM consisting of samples of the same observational and interventional distributions, with a sample size of $n=10,100$, or $1000$ observations. Each interventional distribution stems from an environment where a single variable is intervened on at random, with targets being different across environments.\ntext{ In Appendix \ref{s:additional_experiments}, we show the results for multiple-target interventions.} For \autoref{fig:model_match}, the interventions satisfy the modeling assumptions, that is, they change the noise term distribution of the target by setting its variance to a value sampled uniformly at random from \rmj{$[5,10]$}\ntext{$[3,4]$}. \ntext{In Appendix \ref{s:additional_experiments}, we provide results for a misspecified intervention model where the interventions are instead \emph{do-} or hard interventions \citep[Section 3.2.2]{pearl2009causal}}.

While having access to an observational sample is not a requirement for GnIES, it is needed for a fair comparison with the other baselines. In particular, an observational environment is required by UT-IGSP; furthermore, it implies that the family of targets is \emph{conservative}, an assumption required by GIES \citep[Definition 6]{hauser2012characterization}.

\textbf{Varsortability.} As shown by \citet{reisach2021varsort}, synthetic data as generated above incurs a high degree of \emph{varsortability}, that is, that marginal variance tends to increase along the causal ordering of the data-generating graph. This fact can be exploited for causal discovery, and the authors find that some recent continuous-optimization methods \citep{zheng2018dags, ng2020role} inadvertently rely on this property of synthetic data. However, this reliance means that their synthetic performance may not carry over to real-world data, where varsortability may be moderate and depends on the measurement scale; indeed, their performance broke down after standardization of the data, and when applied to the real data set from \citet{sachs2005}. To check to what extent the performance of GnIES depends on varsortability, we perform our experiments on both the standardized (Figures \ref{fig:model_match}, \ref{fig:model_mismatch}) and raw (\autoref{fig:varsort}, Appendix \ref{s:figs}) synthetic data, as advised by \citet{reisach2021varsort}. As a sanity check, we additionally run the \emph{sortnregress} algorithm proposed by \citet{reisach2021varsort} on the \rmj{pooled data---t}\ntext{data pooled across environments. T}he method shows remarkable performance in recovering the data-generating graph when there is a high varsortability in the data, and breaks down when there is not. We find that neither GnIES nor the other algorithms employed in the experiments show a significant change in performance after standardization of the data (see \autoref{fig:varsort} in Appendix \ref{s:figs}), an indication that they do not exploit varsortability in their inference.

\begin{figure}[h]
\centering
\includegraphics[width=0.85\textwidth]{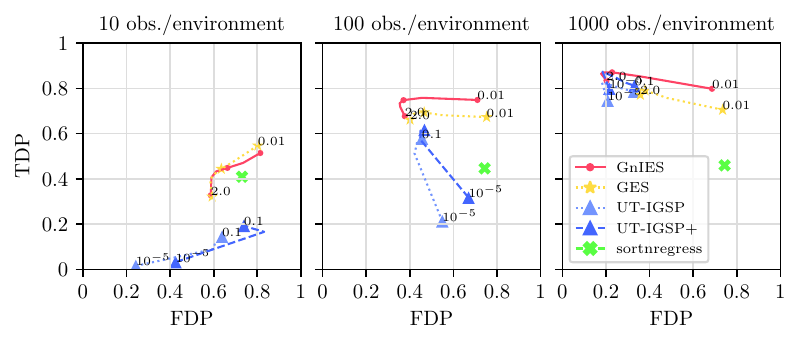}
\caption{\ntext{\comm{[updated figure] }}(correctly specified model) Recovery of the true equivalence class for different sample sizes, as measured by the metrics described in \autoref{def:metrics}; the point \rmj{$(1,0)$}\ntext{$(0,1)$} corresponds to perfect recovery. The results are averaged over 1000 data sets originating from 100 linear Gaussian SCMs\ntext{ following the model in \autoref{eq:model}}. Because UT-IGSP returns a single DAG as estimate, we use its estimated targets and the completion algorithm of GnIES to produce an estimate of the equivalence class (displayed as UT-IGSP+).\rmj{ GES and sortnregress are run on the data pooled across environments.} For each method, the regularization path is shown, together with the smallest and largest values of the regularizers\rmj{. For GnIES and GES, the regularizer is set to $\lambda := \lambda'\log(N)$ where $N$ is the total number of observations and $\lambda' \in [0.01, 2]$}; the additional point shown on the regularization path \ntext{for GnIES and GES } corresponds to the BIC score.\rmj{ For UT-IGSP, the regularizers are the levels, between $10^{-5}$ and $0.1$, of the conditional independence and invariance tests. The poor performance of sortnregress reflects the low varsortability in the data after it is standardized.}}
\label{fig:model_match}
\end{figure}

Due to its flexibility in the assumptions placed on the data-generating process, we employ UT-IGSP \citep{squires2020permutation} as a baseline in all our experiments. An extension of the greedy sparse permutation algorithm \citep{solus2017consistency}, UT-IGSP works by measuring the conditional independencies and invariances in the data through statistical hypothesis tests, and finds the sparsest DAG and intervention targets that satisfy these constraints. The choice of tests can vary according to the setting; for the synthetic experiments of this section, we use standard partial correlation independence and invariance tests provided by the authors in their implementation of the algorithm. A drawback of the method is that it returns a single graph as estimate, when in practice there may be competing equivalent graphs, as happens in the setting that we study here. To ensure a fair comparison, we take the graph and targets estimated by UT-IGSP and compute the corresponding $\cc{I}$-equivalence class with the GnIES completion algorithm (see Appendix \ref{s:completion_algorithm}).

As a baseline for the setting where the model assumptions are satisfied (\autoref{fig:model_match}), we additionally run GES \citep{chickering2002optimal} on the \rmj{pooled data}\ntext{data pooled across environments}. The motivation is that this could be an initial approach employed by a practitioner to recover the class of equivalent models. Its performance suggests that this would not be entirely misguided, particularly at the smallest sample size.
At this size, a possible explanation for the similar performance of GnIES and GES is that the effect of the interventions is less pronounced; due to its penalization, GnIES includes fewer intervention targets in its estimate and thus searches over classes that are often observational, essentially solving the same problem as GES. The poor performance of UT-IGSP in this sample size is not due to the diminished effect of interventions, but rather its reliance on conditional independence and invariance tests, which perform better with larger samples. Indeed, for $n=10$, in about $9\%$ of the runs the correlation matrix computed internally by UT-IGSP becomes singular, and the method fails to produce an estimate; we exclude these cases when computing the average results in the figures. As expected of the setting where all its assumptions are satisfied, GnIES performs competitively in recovering the true equivalence class and does so better than GES due to the additional identifiability.

\ntext{For the results shown in \autoref{fig:model_match} and the additional figures in Appendix \ref{s:additional_experiments}, we evaluate UT-IGSP at different levels of the conditional independence and invariance tests, taken from the range $[10^{-5}, 0.1$]. We evaluate GnIES and GES with different values of the regularization parameter $\lambda := \lambda'\log(N)$, where $N$ is the total number of observations and $\lambda'$ is taken from the range $[0.01, 2]$.}

\subsection{Biological and Semi-synthetic Data}
\label{ss:real}

We now consider the protein expression data set\footnote{The data set was downloaded from\\\href{https://www.science.org/doi/suppl/10.1126/science.1105809/suppl_file/sachs.som.datasets.zip}{\texttt{science.org/doi/suppl/10.1126/science.1105809/suppl\_file/sachs.som.datasets.zip}}.} from \citet{sachs2005}. It consists of measurements of the abundance of 11 phosphoproteins and phospholipids in human immune cells, recorded via a multiparameter flow cytometer under different experimental conditions. These conditions consist of adding reagents to the cell medium \citep[table 1]{sachs2005}, which inhibit or activate different nodes of the protein signaling network. As has been done in previous works \citep{squires2020permutation, meinshausen2016methods, mooij2013cyclic}, we take the observations under different conditions as originating from different interventional distributions, resulting in a data set with a total of 7466 observations across nine interventional samples, each ranging from 707 to 913 observations. Furthermore, because the observations correspond to independent measurements of individual cells, we assume them to be independent and identically distributed. Since GnIES and the other baselines aim to estimate the causal structure underlying the observations, we require a ground truth to evaluate their performance. To this end, we consider the ``consensus network'' presented in \cite[Figure 2]{sachs2005}.

The data set is an interesting challenge because of an array of violations of the GnIES modeling assumptions. The functional relationships between variables appear to be non-linear and the distributional assumption of normality is strongly violated (see Figures \ref{fig:sachs_normality} and \ref{fig:sachs_nonlinear} in Appendix \ref{s:drfs}). Furthermore, the conventionally accepted structure of the protein signaling network contains cycles \citep[see nodes PIP2, PIP3 and PLC$\gamma$ in][Figure 2]{sachs2005} and, as for any measurements of real-world systems, the presence of unmeasured confounders cannot be discarded.

\begin{figure}[h]
\centering
\includegraphics[width=0.75\textwidth]{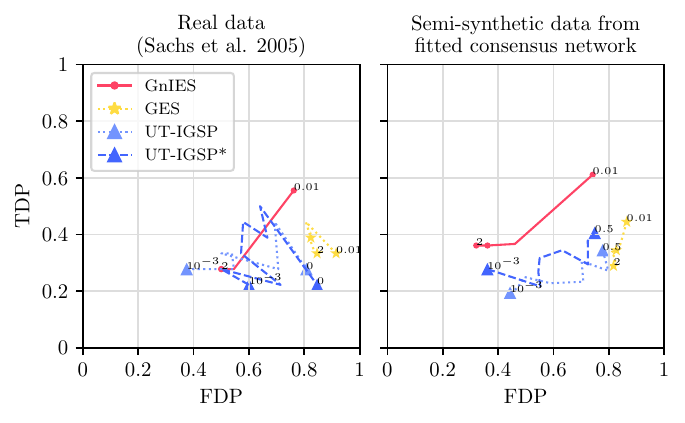}
\caption{(Sachs data set) Recovery of the ``consensus network'' from the original Sachs data (left) and data sampled from a non-parametric SCM fitted to the original data with the consensus network structure (right). For the semi-synthetic data, the results are the average over 10 different samples from the same fitted interventional distributions. UT-IGSP corresponds to the method run with partial correlation conditional independence and invariance tests, and UT-IGSP* to the method with KCI kernel tests \citep{zhang2011kernel}. Because GnIES and UT-IGSP estimate interventions on all variables, their class estimates are a singleton; thus, the TDP and FDP metrics for these methods correspond to the true and false positive rates in terms of directed edge recovery. GES is run on the pooled data across environments and returns the Markov equivalence class estimated from the data. For each method, the regularization path is shown. For GnIES and GES,\rmj{ the regularizer is set to $\lambda := \lambda'\log(N)$ where $N$ is the total number of observations and $\lambda' \in [0.01, 2]$;} the additional point shown on the regularization path corresponds to the BIC score. \rmj{For UT-IGSP, the regularizers are the levels of the conditional independence and invariance tests.}}
\label{fig:sachs}
\end{figure}
\textbf{Evaluation on semi-synthetic data.} Another issue is that the consensus network, which is constructed from effects that appear at least once in the literature \citep[table S1]{sachs2005}, may not be a complete or accurate description of the causal relationships underlying the observations of the protein signaling network. As a safeguard, and to better understand the effect that the different model violations have on their performance, we evaluate the methods on both the original data set and semi-synthetic data sampled from a non-parametric Bayesian network, fitted to each interventional sample according to the structure of the consensus network\footnote{\label{note:acyclic}After removing the edge PIP2 $\to$ PIP3 to satisfy the acyclicity constraint.}. The sample sizes match those in the real data set. The motivation is the following: the marginal and conditional distributions of the semi-synthetic data closely resemble those of the real data and preserve the linearity and Gaussianity violations (Figures \ref{fig:sachs_normality} and \ref{fig:sachs_nonlinear}, Appendix \ref{s:drfs}). However, it respects acyclicity, causal sufficiency (the absence of hidden confounders), and the conditional independence relations implied by the given ``ground-truth'' graph. Thus, we can evaluate GnIES and the other baselines on two data sets of increasing difficulty, and still receive a certain degree of validation in the case of a faulty ground truth. The non-parametric SCM is fitted and sampled from using distributional random forests \citep{cevid2020distributional}; the procedure is described in detail in Appendix \ref{s:drfs} and can be accessed through the Python package \href{https://github.com/juangamella/sempler}{\texttt{sempler}}.

As for the synthetic experiments, we run GnIES and the other methods over a range of their regularization parameters. Again, we run GES on the observations pooled across all samples. For UT-IGSP, we run the method with the same Gaussian tests as in the synthetic experiments and kernel-based tests \citep{zhang2011kernel} provided in the implementation of the algorithm; we respectively label each approach as UT-IGSP and UT-IGSP*. Due to the nature of the experiments, it is not straightforward to select one as the observational environment for UT-IGSP. For a fair comparison, we compare the performance over every option and display the best result in \autoref{fig:sachs}; the others can be found in \autoref{fig:sachs_obs_real} of Appendix \ref{s:figs}.

Taking the consensus network as ground truth, there is no uniformly best performer for the real data. The performance of UT-IGSP seems to depend heavily on the level of its tests. GES performs uniformly worse than the other methods. For the semi-synthetic data and the selected regularization parameters, GnIES performs uniformly better than the other methods. The performance of UT-IGSP with Gaussian tests remains relatively unchanged and decreases with kernel tests. Several competing explanations exist for this change in performance between real and semi-synthetic data. In principle, the potential differences in the data stem from (1) the acyclicity of the data-generating model, (2) the absence of hidden confounders, or (3) a different ground truth. The interplay between these differences and the inner workings of the algorithms is poorly understood at the moment of writing. We acknowledge that without precise knowledge of the causal structure and data collection process which generated the Sachs data set, our statements about this issue are at best speculative.

\subsection{\ntext{Real Data from a Controlled Physical System with Causal Ground Truth}}
\label{ss:light_tunnel}
\ntext{As a real-data setting with a well-justified ground truth, we consider one of the light-tunnel datasets from the causal chambers \citep{gamella2024chamber}. The light tunnel (\autoref{fig:lt}a) is an elongated chamber with a controllable light source at one end, two linear polarizers mounted on rotating frames, and sensors to 
measure different physical variables such as the light intensity at different locations of the tunnel. A detailed diagram of the tunnel can be found in \citet[Figure 2]{gamella2024chamber}. As a ground truth for causal inference tasks, the authors provide a directed graph summarizing the causal effects in the system \citep[Figure 3a]{gamella2024chamber}, built from background knowledge and further validated using randomized experiments.

The dataset\footnote{The data and its documentation can be found under the \href{https://github.com/juangamella/causal-chamber/tree/main/datasets/lt_interventions_standard_v1}{\nolinkurl{lt_inteventions_standard_v1}} dataset at \href{https://github.com/juangamella/causal-chamber}{\nolinkurl{github.com/juangamella/causal-chamber}}.} contains data from experiments in which the light-source colors ($R,G,B$) and polarizer positions ($\theta_1, \theta_2$) are sampled independently from different distributions, and measurements are collected under different configurations of the sensor parameters (e.g., oversampling rate, measurement duration, etc.). For our validation task, we also consider the measurements of the polarizer angles ($\tilde{\theta}_1, \tilde{\theta}_2$), the current drawn by the light source ($\tilde{C}$), and the infrared-light intensity before ($\tilde{I}_1$), between ($\tilde{I}_2$) and after the polarizers ($\tilde{I}_3$). As observational data, we take observations from the dataset's reference experiment. As interventional data, we consider two types of interventions: shifts in the distributions of the polarizer positions $\theta_1$ and $\theta_2$, and changes in the reference voltages of the angle sensors $\tilde{\theta}_1$ and $\tilde{\theta}_2$. When compared to the reference environment, the former results in the relative angle between the polarizers being close to $90^\circ$, dimming the light reaching the third sensor and effectively scaling the effect that the light-source colors ($R,G,B$) have on its reading $\tilde{I}_3$ (\autoref{fig:lt}d). Reducing the reference voltage of the angle sensors increases their sensitivity, affecting the variance of the resulting measurements $\tilde{\theta}_1$ and $\tilde{\theta}_2$ (\autoref{fig:lt}c). For more details about these effects, we refer the reader to \citet[Appendix IV.2.1 and Figure 14, respectively]{gamella2024chamber}.

\begin{figure}[h]
\centering
\includegraphics[width=152.5mm]{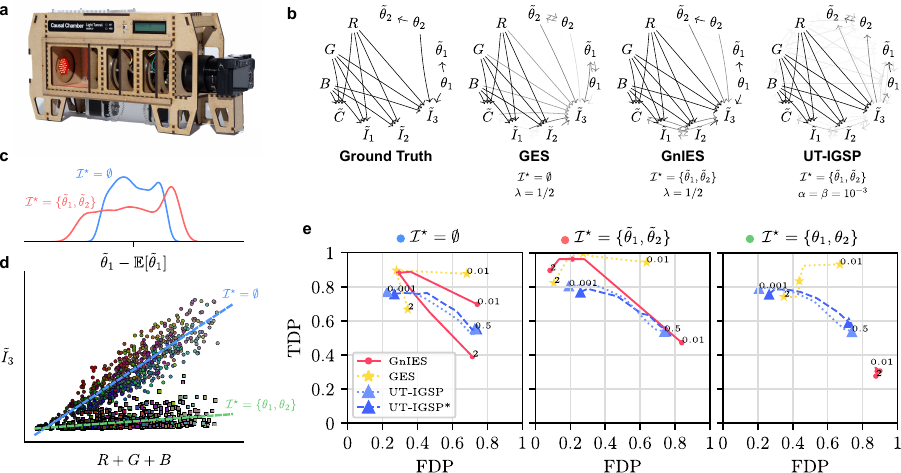}
\caption{\ntext{\comm{[new figure]}(Causal chamber dataset) (a) The light tunnel. (b) The ground-truth graph and a visualization of the estimates produced by each method for 100 bootstrap samples of the data. The color of each edge denotes the frequency with which it appears in the method's estimates, with darker edges appearing at a higher frequency. For GES and GnIES, which return a set of graphs, we pool all graphs to compute the edge frequencies. (c) Effect of the interventions on $\tilde{\theta}_1, \tilde{\theta}_2$: plot of the (centered) marginal density of $\tilde{\theta}_1$ in the reference experiment (blue) and the experiment where we lower the reference voltage of the corresponding sensor (red). (d) Effect of the interventions on the polarizer positions $\theta_1, \theta_2$: the change in the relative angle between the polarizers reduces the amount of light reaching the light sensor $\tilde{I}_3$, scaling the effect of the light-source colors $R,G,B$. (e) Recovery of the ground-truth graph as measured by the metrics in \autoref{def:metrics}. For each method, we show the regularization path averaged over 100 bootstrap samples of the data.}}
\label{fig:lt}
\end{figure}

We run GES, GnIES, and UT-IGSP on the light tunnel data and evaluate their recovery of the ground-truth graph (\autoref{fig:lt}b) for different values of their hyperparameters (\autoref{fig:lt}e). The effects of the light-source colors ($R,G,B$) on the drawn current ($\tilde{C}$) and light intensities ($\tilde{I}_1, \tilde{I}_2, \tilde{I}_3$) are approximately linear and are reliably recovered by all three methods from the observational data alone ($\cc{I}^\star = \emptyset$ in \autoref{fig:lt}). However, both GES and GnIES fail to orient the edges between the polarizer positions ($\theta_1, \theta_2$) and the corresponding angle measurements ($\tilde{\theta}_1, \tilde{\theta}_2$). As predicted by the identifiability results in \autoref{s:equivalence}, including data from the interventions on the reference voltages allows GnIES to correctly orient these edges and reliably recover the ground-truth graph ($\cc{I}^\star = \{\tilde{\theta}_1, \tilde{\theta}_2\}$ in \autoref{fig:lt}b). On the other hand, including data from the interventions on the polarizer positions ($\cc{I}^\star = \{\theta_1, \theta_2\}$) has a catastrophic effect on the performance of GnIES. While the interventions on $\theta_1$ and $\theta_2$ should also allow the method to orient the edges around these variables, the effect of these interventions---which scale the effect (coefficients) of $R,G,B$ on  $\tilde{I}_3$ (\autoref{fig:lt}d)---cannot be expressed by the GnIES model \eqref{eq:model}.
As a result, the method reverses many of the edges in its initial estimate to find an optimal scoring structure, yielding poor results.
UT-IGSP is not affected by the non-linear effect of these interventions, even when partial correlation tests are used (UT-IGSP). However, its performance is virtually unchanged when interventional data is added. The results are similar when the method is run with KCI kernel tests \citep{zhang2011kernel}, shown in \autoref{fig:lt} as UT-IGSP*.}

\section{Discussion}
\label{s:discussion}

This paper considers a parametric approach to the problem of causal structure discovery from interventional data when the targets of interventions are unknown. For linear SCMs with additive Gaussian noise, we provide identifiability results under noise interventions and use these to develop a greedy algorithm to recover the equivalence class of the data-generating model. \ntext{We provide an easy-to-use Python implementation in the package \href{https://github.com/juangamella/gnies}{\texttt{gnies}}. A summary of our software contributions can be found in Appendix \ref{s:software}.}

When evaluated on synthetic data, the algorithm performs competitively with other algorithms built for the same setting, particularly when the available samples are small.
\rmj{For data arising from hard interventions, the performance matches that of the GIES algorithm \citep{hauser2012characterization} albeit lacking knowledge of the intervention targets. Thus, we believe that in situations where a practitioner may consider applying GIES, they may well apply our method and forego the need to precisely know the targets of the experiments that generated the data. To this end, we provide an easy-to-use Python implementation in the package \href{https://github.com/juangamella/gnies}{\texttt{gnies}}. A summary of our software contributions can be found in Appendix \ref{s:software}.} \ntext{We tested our algorithm real data from one of the controlled physical systems described in \citet{gamella2024chamber}; the results show that identifiability improves with interventional data, following the theoretical results in \autoref{s:equivalence}.} We \ntext{also }evaluated the algorithm on the biological data set from \citet{sachs2005} and found its performance in recovering the ``consensus network'' relatively poor, but in line with other methods despite the stronger distributional assumptions. To tease apart the issues underlying this result, we developed a semi-synthetic data generation procedure, which we believe to be of general interest to the causal discovery community. The procedure allows generating semi-synthetic data sets with known causal ground truth but distributions resembling those in real data sets, providing a more realistic challenge than purely synthetic data. We detail the procedure in Appendix \ref{s:drfs} and make it publicly available through the Python package \href{https://github.com/juangamella/sempler}{\texttt{sempler}}.


\subsection{Future Work}
\label{ss:future_work}

While we have employed $\cc{I}$-equivalence to characterize sets of distribution equivalent Gaussian models under noise interventions, the notion is an instance of \emph{transition sequence equivalence}, which was introduced by \citet[Theorem 2]{tian2001causal}. As such, it applies beyond the Gaussian models and noise interventions we consider here. This suggests that an extension of our algorithm to more general models is justified and may be possible by appropriately adapting its score. Within the realm of Gaussian models, one could instead consider other intervention types still covered by the framework of \citet{tian2001causal}. \ntext{In particular, these include interventions that modify the parental coefficients of their target, as we pointed out in \autoref{ss:intervention_model}. Given that the graphical characterization would remain the same, GnIES could be extended to this setting by allowing the coefficients of a target to vary between environments. However, we leave this task and deriving the necessary assumptions for the equivalent of \autoref{prop:imec_full} as future work.}

The similarity between GES and the inner procedure of GnIES suggests another line of work, based on successful extensions of GES. Under certain assumptions, constraining the maximum degree of nodes in the GES search space yielded significant improvements in sample efficiency \citep{chickering2020statistically}; a hybrid approach, using background knowledge or conditional independence tests to restrict the available operators in each step adaptively, has led to a successful scaling of GES to high-dimensional settings with thousands of variables \citep{nandy2018high}. Adapting these constraints to the $\cc{I}$-equivalence class may produce similar improvements for GnIES.

\acks{We would like to thank Jonas Peters, Jinzhou Li, and Malte Londschien for their valuable discussions and comments on the manuscript. J.L. Gamella, A. Taeb, and P. B\"uhlmann have received funding from the European Research Council (ERC) under the European Union’s Horizon 2020 research and innovation program (grant agreement No. 786461).}


\newpage
\appendix
\section{Proofs}
\label{s:proofs}

This section contains the proofs for all results presented in the main text of the paper. An emphasis has been placed on readability and completeness, arguably at the expense of brevity. To aid in the reading of the proofs, we summarize here the notation used and provide a graph of how the results relate to each other.

\textbf{Proof map.} How the presented results relate to each other. An edge $i \to j$ means that result $i$ is used to construct the proof of result $j$.

\begin{figure}[H]
    \centering
    \tikzstyle{block} = [text centered, text width=4em, node distance=4em, font=\small] 
    \def\r{2}    
    \begin{tikzpicture}[]

        \node [block, draw=gray] ("prop_imec_full") at (0:0) {\autoref{prop:imec_full}};        
        \node [block, draw=gray, text width=5.5em] ("assm_int_heter") at (215:\r) {\autoref{assm:int_heter}};
        \node [block, draw=gray, text width=5.5em] ("assm_model_truth") at (270:\r) {\autoref{assm:model_truth}};
        \node [block, draw=gray] ("corr_envs_superset") at (350:1.5*\r) {\autoref{corr:envs_superset}};
        \node [block, above of="prop_imec_full", text width=5.3em, draw=gray] ("prop_imec_faithfulness") {\autoref{prop:imec_faithfulness}};
        \node [block, above of="prop_imec_faithfulness", draw=gray, text width=3.5em] ("prop_imec") {\autoref{prop:imec}};

        \node [block] ("lemma_imaps_dec") at (-3,6) {\autoref{lemma:imaps_dec}};
        \node [block, below right of="lemma_imaps_dec"] ("lemma_mec_dec") {\autoref{lemma:mec_dec}};
        \node [block, below left of="lemma_imaps_dec", text width=3.5em] ("corr_coefficients") {\autoref{corr:coefficients}};                
        \node [block, below of="lemma_imaps_dec", node distance=5.5em] ("lemma_imaps_idec") {\autoref{lemma:imaps_idec}};
        \node [block, above right of="lemma_imaps_dec"] ("thm_meek") {\autoref{thm:meek}};
        \node [block, above left of="lemma_imaps_dec"] ("lemma_inverse") {\autoref{lemma:inverse}};
        \node [block, left of="lemma_inverse", node distance=6.5em] ("lemma_orthogonal_rows_rank") {\autoref{lemma:orthogonal_rows_rank}};
        \node [block, below of="lemma_orthogonal_rows_rank", node distance=4.5em] ("lemma_structure") {\autoref{lemma:structure}};   \node [block, left of="lemma_imaps_idec", node distance=5em] ("lemma_structure_diff") {\autoref{lemma:structure_diff}}; 
        \node [block, below of="lemma_imaps_idec", node distance=3em] ("lemma_scaling") {\autoref{lemma:scaling}};
        \node [block, below left of="lemma_scaling", node distance=3em] ("lemma_targets_parents") {\autoref{lemma:targets_parents}}; 
        
        \node [block, right of="lemma_mec_dec", node distance=6.5em] ("corr_chickering") {\autoref{corr:chickering}};
        \node [block, above left of="corr_chickering", node distance=3.5em] ("lemma_chickering") {\autoref{lemma:chickering}};
        \node [block, above right of="corr_chickering", node distance=3.5em] ("thm_chickering") {\autoref{thm:chickering}};

        \node [block, below left of="corr_chickering", node distance=3em] ("prop_mec_equals_dec") {\autoref{prop:mec_equals_dec}};
        \node [block, right of="thm_chickering", node distance=6em] ("lemma_min_edges_mec") {\autoref{lemma:min_edges_mec}};
        \node [block, below left of="lemma_min_edges_mec", node distance=6em] ("lemma_min_edges_imec") {\autoref{lemma:min_edges_imec}};        
                                
        \node [block] ("faithfulness") at (5,2) {Faithfulness};

        \node[block] ("lemma_graphs_dists") at (-8,1) {\autoref{lemma:graphs_dists}};
        \node[block,below of="lemma_graphs_dists"] ("score_equivalence") {equivalence};        
        \node[block, draw=gray, text width=7em, below of="score_equivalence"] ("prop_score_properties") {\autoref{prop:score_properties}\\(score properties)};
        \node[block, text width=5em, right of="prop_score_properties", node distance=8em] ("consistency") {consistency};
        \node[block, below of="prop_score_properties"] ("decomposability") {decomposability};

        \node[block, text width=6em, draw=gray, right of="decomposability", node distance=30em] ("legend") {in main text};

        \draw [->] ("lemma_structure") to [] ("lemma_graphs_dists");
        \draw [->, out=180, in=60] ("corr_coefficients") to [] ("lemma_graphs_dists");
        \draw [->] ("lemma_imaps_idec") to [] ("lemma_graphs_dists");
        \draw [->] ("lemma_graphs_dists") to [] ("score_equivalence");
        \draw [->] ("score_equivalence") to [] ("prop_score_properties");
        \draw [->] ("decomposability") to [] ("prop_score_properties");
        \draw [->] ("consistency") to [] ("prop_score_properties");
        
        \draw [->, bend right=15] ("prop_imec_full") to [] ("consistency");
        \draw [->] ("assm_int_heter") to [] ("consistency");
        \draw [->] ("assm_model_truth") to [] ("consistency");

        \draw [->] ("prop_imec") to [] ("prop_imec_faithfulness");
        \draw [->] ("prop_imec_faithfulness") to [] ("prop_imec_full");
        \draw [->] ("lemma_mec_dec") to [] ("prop_mec_equals_dec");
        \draw [->, bend left=8] ("thm_meek") to [] ("lemma_min_edges_mec");
        \draw [->, bend left=8] ("lemma_min_edges_mec") to [] ("prop_mec_equals_dec");
        \draw [->, bend left=8] ("lemma_min_edges_imec") to [] ("prop_imec_faithfulness");
        \draw [->] ("lemma_imaps_idec") to [] ("prop_imec_faithfulness");
        \draw [->, bend right=8] ("lemma_imaps_idec") to [] ("prop_imec");
        \draw [->] ("lemma_min_edges_mec") to [] ("lemma_min_edges_imec");
        \draw [->] ("thm_meek") to [] ("prop_mec_equals_dec");
        \draw [->] ("corr_chickering") to [] ("prop_mec_equals_dec");        
        \draw [->] ("lemma_orthogonal_rows_rank") to [] ("lemma_imaps_dec");
        \draw [->, bend right=25] ("lemma_orthogonal_rows_rank") to [] ("lemma_imaps_idec");
        \draw [->] ("lemma_orthogonal_rows_rank") to [] ("lemma_structure");
        \draw [->] ("lemma_inverse") to [] ("lemma_structure");
        \draw [->] ("lemma_inverse") to [] ("lemma_imaps_dec");        
        \draw [->, bend right=20] ("lemma_inverse") to [] ("lemma_structure_diff");
        \draw [->, bend right=5] ("lemma_structure") to [] ("lemma_imaps_idec");
        \draw [->] ("thm_meek") to [] ("lemma_imaps_dec");
        \draw [->] ("lemma_imaps_dec") to [] ("lemma_mec_dec");
        \draw [->] ("lemma_imaps_dec") to [] ("corr_coefficients");
        \draw [->] ("lemma_imaps_dec") to [] ("lemma_imaps_idec");
        \draw [->] ("corr_coefficients") to [] ("lemma_imaps_idec");        
        \draw [->] ("lemma_chickering") to [] ("corr_chickering");
        \draw [->] ("thm_chickering") to [] ("corr_chickering");

        \draw [->] ("faithfulness") to [] ("lemma_min_edges_mec");
        \draw [->] ("faithfulness") to [] ("lemma_min_edges_imec");
        \draw [->] ("faithfulness") to [] ("prop_mec_equals_dec");
        \draw [->] ("faithfulness") to [] ("prop_imec_faithfulness");
        \draw [->, out=250, in=0] ("faithfulness") to [] ("prop_imec_full");
    
        \draw [->, out=260, in=5] ("lemma_min_edges_mec") to [] ("prop_imec_full");
        \draw [->, bend left=50] ("lemma_min_edges_imec") to [] ("prop_imec_full");
        \draw [->, bend left=65] ("prop_mec_equals_dec") to [] ("prop_imec_full");                
        \draw [->] ("assm_int_heter") to [] ("prop_imec_full");
        \draw [->] ("assm_model_truth") to [] ("prop_imec_full");
        \draw [->] ("corr_envs_superset") to [] ("prop_imec_full");
        \draw [->] ("lemma_targets_parents") to [] ("prop_imec_full");

        \draw [->] ("lemma_targets_parents") to [] ("consistency");
        \draw [->, out=220, in=175] ("lemma_orthogonal_rows_rank") to [] ("lemma_targets_parents");
        \draw [->] ("lemma_structure_diff") to [] ("lemma_targets_parents");
        \draw [->] ("lemma_scaling") to [] ("lemma_targets_parents");
        
    \end{tikzpicture}
\end{figure}

\printglossary[title=Notation,type=main,style=long,nonumberlist]

\subsection{Restated Results}

We restate results from \citet{chickering1995, chickering2002optimal} which are used in the proofs of this section. The notation is slightly adapted to fit ours.

\begin{definition}[Covered edge, from \citealp{chickering2002optimal}]
\label{def:covered}
For any DAG $\dag$, we say an edge $j \to i$ in $\dag$ is covered if
$$\pa_\dag(i) = \pa_\dag(j) \cup \{j\}.$$
\end{definition}

\ntext{
\begin{definition}[Independence Map (I-MAP) from \citealp{chickering2002optimal}]
\label{def:imap}
A DAG $H$ is an independence map of a DAG $G$ if every d-separation statement in $H$ also holds in $G$, that is, $i \indep_{H} j \mid S \implies i \indep_{G} j \mid S$.
\end{definition}
}

\begin{theorem}[Restated Theorem 4 from \citealp{chickering2002optimal}]
\label{thm:meek}
Let $\dag$ and $H$ be any pair of DAGs such that $H$ is an independence map of $\dag$. Let $r$ be the number of
edges in $H$ that have opposite orientation in $\dag$, and let $m$ be the number of edges in $H$ that
do not exist in either orientation in $\dag$. There exists a sequence of at most $r + 2m$ edge
reversals and additions in $\dag$ with the following properties:
\begin{enumerate}
    \item Each edge reversed is a covered edge.
    \item After each reversal and addition, $\dag$ is a DAG and $H$ is an independence map of  $\dag$.
    \item After all reversals and additions, $\dag$ = $H$.
\end{enumerate}
\end{theorem}

\begin{theorem}[Restated Theorem 2 from \citealp{chickering1995}]
\label{thm:chickering}
Let $\dag$ and $\dag'$ be any pair of DAGs that are Markov equivalent and for which there are $\delta$ edges in $\dag$ that have opposite orientation in $\dag'$. Then there exists a sequence of $\delta$ distinct edge reversals in $\dag$ with the following properties:
\begin{enumerate}
    \item Each edge reversed in $\dag$ is covered.
    \item After each edge reversal, $\dag$ is a DAG and is Markov equivalent to $\dag'$.
    \item After all reversals, $\dag = \dag'$.
\end{enumerate}
\end{theorem}

\begin{lemma}[Restated Lemma 1 from \citealp{chickering1995}]
\label{lemma:chickering}
Let $\dag$ be any DAG model, and let $\dag'$ be the result of
reversing the edge $X \to Y$ in $\dag$. Then $\dag'$ is a DAG that is Markov equivalent to $\dag$ if and only if
$X \to Y$ is covered in $\dag$.
\end{lemma}

\begin{corollary}[Markov equivalence and covered edge reversals]
\label{corr:chickering} Let $\dag$ and $\dag'$ be two DAGs. From \autoref{thm:chickering} and \autoref{lemma:chickering}, it follows that they are Markov equivalent if and only if there exists a sequence of covered edge reversals that yield $\dag'$ when applied to $\dag$.
\end{corollary}

\subsection{Supporting Results from Linear Algebra}
\label{sec:lemmas}


\begin{lemma}
\label{lemma:orthogonal_rows_rank}
Let $A \in \bb{R}^{p \times p}$. The following two statements are equivalent,
\begin{enumerate}
    \item $AA^T$ is diagonal with positive entries, and
    \item $A$ has full rank and orthogonal row vectors.
\end{enumerate}
\end{lemma}
\begin{proof}
To show that $1 \implies 2$, assume $AA^T$ is diagonal with positive entries but \begin{enumerate}[label=(\roman*)]
    \item $\text{rank}(A) < p$, or
    \item $\exists i,j$ such that $i \neq j$ and $\langle A_{i:}, A_{j:} \rangle \neq 0$.
\end{enumerate}
If (i) is true, then $\text{rank}(A) = \text{rank}(AA^T) < p$, which means that $\text{det}(AA^T) = 0$ and thus, since $AA^T$ is diagonal, $\exists i$ such that $A_{ii}=0$, which is a contradiction. If (ii) is true, then $(AA^T)_{ij} = \langle A_{i:}, A_{j:} \rangle \neq 0$, which is a contradiction.

That $2 \implies 1$ is trivial, as for $i\neq j$,  $(AA^T)_{ij} = \langle A_{i:}, A_{j:} \rangle = 0$ by definition, and $(AA^T)_{ii} = \langle A_{i:}, A_{i:} \rangle > 0$ as $A$ has full rank and $A_{i:} \neq \Vec{0}$ for all $i$.
\end{proof}

\begin{lemma}
\label{lemma:inverse}
Let $A,B \in \bb{R}^{p\times p}$ with $B$ invertible, $\alpha \in \bb{R}$ and let $e_i \in \bb{R}^p$ be the $i^\text{th}$ canonical unit vector.
$$A_{i:} = \alpha B_{i:} \iff (AB^{-1})_{i:} = \alpha e_i^T.$$
\end{lemma}
\begin{proof}
"$\implies$" $(AB^{-1})_{i:} = A_{i:}B^{-1} = \alpha B_{i:}B^{-1} = \alpha e_i^T.$
\\
"$\impliedby$" $(AB^{-1})_{i:} = \alpha e_i^T \implies A_{i:}B^{-1} = \alpha e_i^T \implies A_{i:} = \alpha e_i^TB = \alpha B_{i:}.$
\end{proof}

\begin{lemma}[Scaling of a matrix with orthogonal rows]
\label{lemma:scaling}
Let $M \in \bb{R}^{p \times p}$ be a full-rank matrix with orthogonal row vectors, such that every row has more than one non-zero entry. Let $D \in \bb{R}^{p \times p}$ be a diagonal matrix with positive diagonal entries for which $\exists i \in [p]$ such that $D_{ii} \neq D_{jj}$ for all $j \neq i$. Then $MD$ does not have orthogonal row vectors.
\end{lemma}
\begin{proof}
Assume that $MD$ does in fact have orthogonal row vectors, \ie that $MDD M^T$ is a diagonal matrix. Denote $\tilde{M} := MDD$, and pick\footnote{Such $k$ exists as $M$ has full rank.} $k$ such that $M_{ki} \neq 0$. We have that
\begin{align}
\label{eq:prod_1}
\langle M_{k:}, M_{j:}\rangle = 0 \text{ and } \langle\tilde{M}_{k:}, M_{j:}\rangle = 0 \text{ for } j \neq k.
\end{align}
Since $\{M_{j:}\}_{j \neq k}$ are $p-1$ linearly independent vectors in $\bb{R}^p$, \eqref{eq:prod_1} implies that $\tilde{M}_{k:}$ and $M_{k:}$ belong to the same one-dimensional subspace. Thus, for some $\alpha \in \bb{R}$ we have that
\begin{align}
\label{eq:constraint}
    \tilde{M}_{k:} =& M_{k:}DD = \alpha M_{k:}\nonumber\\
    \implies & M_{kl}D_{ll}^2 = \alpha M_{kl} \text{ for all } l \in [p]\nonumber \\
    \implies & D_{ll}^2 = \alpha \text{ for all } l : M_{kl} \neq 0.
\end{align}
Note that we chose $k$ s.t. $M_{ki} \neq 0$, and because the rows of $M$ contain more than one non-zero entry, $\exists j \neq i$ such that  $M_{kj} \neq 0$. Thus, \eqref{eq:constraint} implies that
$D_{ii}^2 = \alpha = D_{jj}^2$.
Since the diagonal entries of $D$ are positive, the above means that $D_{ii} = D_{jj}$, arriving at a contradiction with our constraints on $D$. Thus, $MD$ cannot have orthogonal row vectors, completing the proof.
\end{proof}

\subsection{Supporting Results for the Observational Case ($\abs{\cc{E}} = 1$)}
\label{ss:support_mec}

While the results of this subsection are known in the literature, for completeness we prove them again here for the models we consider.

\begin{lemma}[I-MAPs are distribution equivalent]
\label{lemma:imaps_dec}
Let $(B, \Omega)$ be a model with a single environment, which entails a distribution $P \sim \cc{N}(0, \Sigma)$ with $\Sigma := (I-B)^{-1}\Omega(I-B)^{-T}$. Let $\dag$ be an independence map of $\cc{G}(B)$, that is, all the d-separation statements in \rmj{$\cc{G}(B)$}\ntext{$\dag$} hold also in \rmj{$\dag$}\ntext{$\cc{G}(B)$}. Then,
\begin{center}
    $\exists$ a model $(B',\Omega')$ s.t. $B' \comp \dag$ and $\Sigma = (I-B')^{-1}\Omega'(I-B')^{-T}$.
\end{center}
\end{lemma}
\begin{proof}
\newline
\newline
\underline{Case I:} \rmj{All edges in $\cc{G}(B)$ appear in $\dag$}\ntext{The edges in $\cc{G}(B)$ are a subset of the edges in $G$}. Then, trivially $B':=B$ and $\Omega':=\Omega$, completing the proof.
\newline
\newline
\underline{Case II:} Let \rmj{all the edges in $\cc{G}(B)$ appear in $\dag$} \ntext{the edges in $\cc{G}(B)$ be a subset of the edges in $G$}, except for a single edge $j \to i$ which has been reversed to $j \leftarrow i$ in $\dag$. Furthermore, assume that $j \leftarrow i$ is covered in $\dag$. This means that
\begin{align}
    \label{eq:same_parents}
    &\pa_{\cc{G}(B)}(k) \subseteq \pa_\dag(k) \text{ for all } k \notin \{i,j\},\\
    \label{eq:parents_i}
    &\pa_{\cc{G}(B)}(i) \setminus \{j\} \subseteq \pa_\dag(i),\\
    \label{eq:parents_j}
    &\pa_{\cc{G}(B)}(j) \cup \{i\} \subseteq \pa_\dag(j), \text{ and}\\
    \label{eq:covered}
    &\pa_\dag(i) \cup \{i\} = \pa_\dag(j).
\end{align}
We proceed to construct a connectivity matrix $B'$ and noise term covariance $\Omega'$ such that the result holds. First, set all but the $i^\text{th}$ and $j^\text{th}$ rows of $B'$ to be the same as in $B$, \ie $B'_{k:} = B_{k:} \ \forall k \notin \{i,j\}$. We construct the remaining $i^\text{th}$ and $j^\text{th}$ rows of $B'$ by the following rule:
\begin{align}
\label{eq:lin_comb}
    (I - B')_{i:} := \beta^i \Omega^{-1/2} (I-B) \text{ and } (I - B')_{j:} := \beta^j \Omega^{-1/2} (I-B)
\end{align}
for some $\beta^i, \beta^j \in \bb{R}^p$ such that
\begin{enumerate}[label=\roman*)]
    \item $\beta^i_k = \beta^j_k = 0$ for all $k \notin \{i,j\}$,
    \item $[\beta^i \Omega^{-1/2} (I-B)]_i = 1$ and $[\beta^j \Omega^{-1/2} (I-B)]_j = 1$.
\end{enumerate}
Additionally, we impose the constraints that (iii) $\beta^i {\beta^j}^T = 0$ and (iv) $(I-B')_{ij} = [\beta^i \Omega^{-1/2} (I-B)]_j = 0$; their purpose will become clear later. We can express all the constraints in matrix form. For the vector $\beta^i$ we have
\begin{align}
    \label{eq:beta_i}
    \beta^iC^i = (0,...,0,1) \text{ where } C^i := \left[ \{e_k\}_{k \notin \{i,j\}}, \Omega^{-1/2}(I-B)_{:j}, \Omega^{-1/2}(I-B)_{:i} \right],
\end{align}
and for $\beta^j$
\begin{align}
    \label{eq:beta_j}
    \beta^jC^j = (0,...,0,1) \text{ where } C^j := \left[ \{e_k\}_{k \notin \{i,j\}}, {\beta^i}^T, \Omega^{-1/2}(I-B)_{:j}\right],
\end{align}
where $e^k$ is the $k^\text{th}$ canonical unit vector\footnote{That is, $e^k_k = 1$ and $e^k_i=0 \ \forall i \neq k$.}.
Up to a permutation of columns, $C^i$ is lower triangular with non-zero diagonal entries, and thus it has full rank; therefore $\beta^i$, as a solution to \eqref{eq:beta_i}, exists and is unique. Plugging it into \eqref{eq:beta_j} yields a full rank matrix; to see this, consider that the last two columns are linearly independent from the first $p-2$ columns\footnote{They are unit vectors with zero $i^\text{th}$ and $j\text{th}$ entries, whereas $\beta^i$ and $(I-B)_{:j}$ have non-zero elements in these positions.}. The last two columns are orthogonal to each other by construction (iv). Thus, $\beta^j$ also exists and is unique. Now we show that the resulting $B'$ satisfies the requirements.

\paragraph{$B'$ is compatible with the graph $\dag$, \ie $B' \comp \dag$.} We proceed by showing that for each node $l$,
$\supp{B'_{l:}} \subseteq \pa_\dag(l)$,
which implies $B' \comp \dag$. Since $B'_{k:} = B_{k:} \ \forall k \notin \{i,j\}$, by \eqref{eq:same_parents} it holds that $\supp{B'_{k:}} = \supp{B_{k:}} = \pa_{\cc{G}(B)}(k) \subseteq \pa_{\dag}(k)$ for all $k \notin \{i,j\}$. Now we show that it also holds for $i$ and $j$. By \eqref{eq:lin_comb}, we have that
$B'_{i:} = \beta^i\Omega^{-1/2}B - \beta^i\Omega^{-1/2}-I_{i:},$
subject to the constraints that $B'_{ii}=0$ (ii) and $B'_{ij}=0$ (iv). Thus,
\begin{align*}
\supp{B'_{i:}} &= \left ( \supp{\beta^i\Omega^{-1/2}B} \cup \supp{\beta^i\Omega^{-1/2}} \cup \supp{I_{i:}} \right) \setminus \{i,j\} \\
    \text{ \emph{by constraint} (i)} \quad
    &\subseteq \left ( \supp{B_{j:}} \cup \supp{B_{i:}} \cup \{i,j\} \cup \{i\} \right )\setminus \{i,j\}\\
    &= \supp{B_{j:}} \setminus \{i,j\} \cup \supp{B_{i:}} \setminus \{i,j\}\nonumber\\
    &= \pa_{\cc{G}(B)}(j) \setminus \{i\} \cup \pa_{\cc{G}(B)}(i) \setminus \{j\} \\
    \text{\emph{by} (\ref{eq:parents_i}, \ref{eq:parents_j})} \quad
    &\subseteq \pa_\dag(j) \setminus \{i\} \cup \pa_\dag(i)\\
    \text{\emph{by} \eqref{eq:covered}} \quad
    &= \pa_\dag(i) \cup \pa_\dag(i)\\
    &= \pa_\dag(i) \nonumber.
\end{align*}
Similarly, we have that
$B'_{j:} = \beta^j\Omega^{-1/2}B - \beta^j\Omega^{-1/2}-I_{j:},$
subject to $B'_{jj}=0$ (ii), and thus
\begin{align*}
\supp{B'_{j:}} &= \left (\supp{\beta^j\Omega^{-1/2}B} \cup \supp{\beta^j\Omega^{-1/2}} \cup \supp{I_{j:}} \right )\setminus \{j\} \\
    \text{ \emph{by constraint} (i)} \quad
    &\subseteq \left ( \supp{B_{j:}} \cup \supp{B_{i:}} \cup \{i,j\} \cup \{i\} \right )\setminus \{j\}\\
    & = \left ( \supp{B_{j:}} \cup \{i\} \cup \supp{B_{i:}} \right )\setminus \{j\} \\
    &= \left (\pa_{\cc{G}(B)}(j) \cup \{i\} \cup \pa_{\cc{G}(B)}(i) \right) \setminus \{j\} \nonumber\\
    &= \pa_{\cc{G}(B)}(j) \cup \{i\} \cup \pa_{\cc{G}(B)}(i) \setminus \{j\} \nonumber\\
    \text{\emph{by} (\ref{eq:parents_i}, \ref{eq:parents_j})} \quad
    &\subseteq \pa_\dag(j) \cup \pa_\dag(i)\\
    \text{\emph{by} \eqref{eq:covered}} \quad
    &= \pa_\dag(j) \nonumber.
\end{align*}

\paragraph{$(B', \Omega')$ is a model such that $\Sigma = (I-B')^{-1}\Omega'(I-B')^{-T}$.} Since $\cc{G}(B')$ corresponds to a DAG\footnote{$\cc{G}(B')$ has its edges contained in the DAG $\dag$, and is thus a DAG.}, it holds that $B'$ is lower triangular up to a permutation of rows and columns and has zero diagonal entries. Now we look at the noise term variances; let $\Omega' := (I-B')(I-B)^{-1}\Omega(I-B)^{-T}(I-B')^T$. By \autoref{lemma:orthogonal_rows_rank}, $\Omega'$ is a diagonal matrix with positive entries if and only if $(I-B')(I-B)^{-1}\Omega^{1/2}$ has full rank and orthogonal row vectors. The first condition holds as it is the product of three invertible matrices. For the orthogonality of the row vectors, note that because $B'_{k:} = B_{k:} \ \forall k \notin \{i,j\}$, by \autoref{lemma:inverse} we have that $[(I-B')(I-B)^{-1}]_{k:} = e_k^T$, and thus these rows are orthogonal to each other. Furthermore, since by construction $[(I-B')(I-B)^{-1}\Omega^{1/2}]_{ik} = \beta^i_k = 0$ and $[(I-B')(I-B)^{-1}\Omega^{1/2}]_{jk} = \beta^j_k = 0 \ \forall k \notin \{i,j\}$, they are also orthogonal to the $i^\text{th}$ and $j^\text{th}$ rows. These two rows are also orthogonal to one another, as by constraint (iv) we have $\beta^i{\beta^j}^T = 0$. Thus $(B', \Omega')$ is a valid model for which $\Sigma = (I-B')^{-1}\Omega'(I-B')^{-T}$. This completes the proof for case II.
\newline
\newline
\underline{Case III:}
Suppose now $\dag$ is any independence map of $\cc{G}(B)$. By \autoref{thm:meek}, there exists a sequence of graphs
$(\cc{G}(B) = \dag_1, \dag_2, \ldots, \dag_k = \dag)$,
where consecutive graphs differ in the reversal of a single covered edge or the addition of a single edge to $G_i$. We let $B^1:= B$, and proceeding iteratively, for each matrix $B^i$ and graph $G_{i+1}$ we apply the result of case I or II to obtain a new model $(B^{i+1}, \Omega^{i+1}) \in [(B, \Omega)]$ with $B^{i+1} \sim G_{i+1}$. At the end of this process, we can set $B':=B^k$ and $\Omega':=\Omega^k$, completing the proof.
\end{proof}

\begin{corollary}[Markov equivalence implies distributional equivalence]
\label{lemma:mec_dec}
Let $(B, \Omega)$ be a model with a single environment. It follows from \autoref{lemma:imaps_dec} that
$$\dag \in \text{MEC}(\cc{G}(B)) \implies \dag \in \bb{G}(B, \Omega).$$
\end{corollary}

\begin{corollary}[Same coefficients]
\label{corr:coefficients}
From the construction of $B'$ in the proof of \autoref{lemma:imaps_dec}, it follows that $B_{i:} = B'_{i:}$ for all $i$ such that $\pa_{\cc{G}(B)}(i) \subseteq \pa_\dag(i)$.
\end{corollary}

\begin{lemma}[Under faithfulness, the true model has the minimal number of edges]
\label{lemma:min_edges_mec}
Let $(B, \Omega)$ be a model with a single environment, which entails a distribution $P \sim \cc{N}(0, \Sigma)$ with $\Sigma := (I-B)^{-1}\Omega(I-B)^{-T}$. Let $\cc{G}(B)$ be the corresponding graph and assume that $P$ is faithful with respect to it. Then it holds that
$$
\min_{\dag \in \bb{G} \left (B, \Omega \right)} \|\dag\|_0 = \|\cc{G}(B)\|_0.
$$
\end{lemma}
\begin{proof}
Let $\cc{C}(P)$ denote the set of conditional independence relationships in the distribution $P$, and let $\text{d-sep}(\dag)$ denote the d-separation statements entailed by some graph $\dag$.

Take any $\dag \in \bb{G} \left (B, \Omega \right)$. There exist $(B', \Omega')$ such that $B' \comp \dag$ and $\Sigma = (I-B')^{-1}\Omega'(I-B')^{-T}$; in other words, there exists a linear Gaussian SCM with connectivity $B'$, noise term variances $\Omega'$ and underlying graph $\cc{G}(B')$ which entails the distribution $P$. As such, this distribution $P$ is Markov with respect to $\cc{G}(B')$; because $\dag$ contains all edges in $\cc{G}(B')$, the distribution is also Markov wrt. $\dag$, \ie all d-separation relations in the graph are matched by a conditional independence relationship in $P$. Because $P$ is faithful with respect to $\cc{G}(B)$, it follows that
$\text{d-sep}(\dag) \subseteq \text{d-sep}(\cc{G}(B)) = \cc{C}(P)$.
As such, $\dag$ is an independence map of $\cc{G}(B)$, and by \autoref{thm:meek} there exists a sequence of edge reversals and additions that yield $\dag$ when applied to $\cc{G}(B)$. It follows that $\dag$ cannot have fewer edges than $\cc{G}(B)$.
\end{proof}

\begin{proposition}[Markov and minimal-edge distributional equivalence]
\label{prop:mec_equals_dec}
Let $(B, \Omega)$ be a model as in \eqref{eq:model} with a single environment, which entails a distribution $P \sim \cc{N}(0, \Sigma)$ with $\Sigma := (I-B)^{-1}\Omega(I-B)^{-T}$. Let $\cc{G}(B)$ be the corresponding graph and assume that $P$ is faithful with respect to it. Then, it holds that

$$
\text{MEC}(\cc{G}(B)) = \arg\min_{\dag \in \bb{G} \left (B, \Omega \right)} \|\dag\|_0,
$$
where $\|\dag\|_0$ denotes the number of edges in $\dag$.
\end{proposition}
\begin{proof}
We show that both sets contain each other.
\paragraph{$\subseteq$} Let $\dag'$ be a Markov equivalent graph to $\cc{G}(B)$. By \autoref{lemma:mec_dec} there exists $(B',\Omega')$ in $\left[(B, \Omega)\right]$ such that $B' \comp \dag'$, and thus $\dag' \in \bb{G} \left (B, \Omega \right)$. By \autoref{lemma:min_edges_mec} we have that $\min_{\dag \in \bb{G} \left (B, \Omega \right)} \|\dag\|_0 = \|\cc{G}(B)\|_0$; since $\cc{G}(B)$ and $\dag'$ are Markov equivalent, they have the same number of edges, and thus $\dag' \in \arg\min_{\dag \in \bb{G} \left (B, \Omega \right)} \|\dag\|_0$.

\paragraph{$\supseteq$} Let $\cc{C}(P)$ denote the set of conditional independence relationships in the distribution $P$, and let $\text{d-sep}(\dag)$ denote the d-separation statements entailed by some graph $\dag$. Take $\dag' \in \arg\min_{\dag \in \bb{G} \left (B, \Omega \right)} \|\dag\|_0$; there exist $(B', \Omega')$ such that $B' \comp \dag'$
and $\Sigma = (I-B')^{-1}\Omega'(I-B')^{-T}$; in other words, there exists a linear Gaussian SCM with connectivity $B'$, noise term variances $\Omega'$ and underlying graph $\cc{G}(B')$ which entails the distribution $P$. As such, this distribution $P$ is Markov with respect to $\cc{G}(B')$ and its independence map $\dag'$. Because $P$ is faithful with respect to $\cc{G}(B)$, it follows that
$\text{d-sep}(\dag') \subseteq \text{d-sep}(\cc{G}(B)) = \cc{C}(P)$,
and thus $\dag'$ is an independence map of $\cc{G}(B)$. By \autoref{thm:meek} there exists a sequence $S$ of edge additions and covered-edge reversals that yield $\dag'$ when applied to $\cc{G}(B)$. By \autoref{lemma:min_edges_mec}, $\dag'$ and $\cc{G}(B)$ have the same number of edges; therefore, the sequence $S$ contains only covered-edge reversals and no edge additions. By \autoref{corr:chickering}, this means that $\dag'$ and $\cc{G}(B)$ are Markov equivalent, completing the proof.
\end{proof}

\begin{lemma}[structure of $(I-B')(I-B)^{-1}$ for same parents]
\label{lemma:structure}
Consider a model $(B, \Omega)$ with a single environment and let $(B', \Omega') \in \left[\left(B,\Omega\right)\right]$ be an equivalent model. Let $e_i$ be the $i^\text{th}$ canonical unit vector. The following statements are equivalent:
\begin{enumerate}[label=\roman*)]
    \item $B_{i:} = B'_{i:}$,
    \item $\left[(I-B')(I-B)^{-1} \right]_{i:} = e_i^T$, and
    \item $\left[(I-B')(I-B)^{-1} \right]_{:i} = e_i$.
\end{enumerate}
\end{lemma}
\begin{proof}
We proceed in order.
\paragraph{($i \implies ii$)} Follows immediately from \autoref{lemma:inverse}.

\paragraph{($ii \implies iii$)} Since $(B', \Omega') \in \left[\left(B,\Omega\right)\right]$, it holds that 
$(I-B')(I-B)^{-1}\Omega(I-B)^{-T}(I-B')^T$
is a diagonal matrix with positive entries. By \autoref{lemma:orthogonal_rows_rank}, it follows that $(I-B')(I-B)^{-1}\Omega^{1/2}$ has full rank and orthogonal row vectors. Now, since $(ii)$ is true, the $i^\text{th}$ row vector is $\left[(I-B')(I-B)^{-1}\Omega^{1/2} \right]_{i:} = \Omega_{ii}^{1/2}e_i^T$, and so the remaining vectors are orthogonal to it iff they all have zeros on the $i^\text{th}$ coordinate, \ie $\left[(I-B')(I-B)^{-1} \right]_{:i} = e_i$.

\paragraph{($iii \implies i$)} Note that $(iii)$ can be rewritten as $(I-B')(I-B)^{-1}e_i = e_i$. It follows that
\begin{equation}
\label{eq:unit_column}
    e_i^T = e_i^T(I-B')^{-T}(I-B)^T.
\end{equation}
Now, since $(B', \Omega') \in \left[\left(B,\Omega\right)\right]$, we have that
$$(I-B')(I-B)^{-1}\Omega(I-B)^{-T}(I-B')^T = \Omega',$$
which we can rewrite as
$$(I-B')
= \Omega'(I-B')^{-T}(I-B)^T\Omega^{-1}(I-B).$$
Multiplying on the left by $e_i^T$, and letting $\Omega'_{ii} = \alpha$ and $\Omega^{-1}_{ii} = \beta$ we arrive at
\begin{align*}
e_i^T(I-B) &= e_i^T\Omega'(I-B')^{-T}(I-B)^T\Omega^{-1}(I-B) \\
&= \alpha e_i^T (I-B')^{-T}(I-B)^T\Omega^{-1}(I-B) \\
&= \alpha e_i^T \Omega^{-1}(I-B) \quad \text{(c.f. \autoref{eq:unit_column})} \\
&= \alpha \beta e_i^T (I-B).
\end{align*}
We have arrived at $(I-B)_{i:} = \alpha\beta(I-B)_{i:}$. Since both $B'$ and $B$ have zeros on the diagonal, $(I-B)_{ii} = (I-B')_{ii} = 1 = \alpha\beta(I-B)_{ii}$, which implies that $\alpha\beta = 1$. Therefore, $(I-B')_{i:} = (I-B)_{i:}$ and $B'_{i:} = B_{i:}$.
\end{proof}

\begin{lemma} (independence of noise terms and parents)
\label{lemma:independence}
Let $(B, \Omega)$ be a model and define the random vector
$X := (I-B)^{-1}\epsilon$  with $\epsilon \sim \cc{N}(0, \Omega)$. For any $i \in \preds$, it holds that
$\epsilon_i \indep X_\supp{B_{i:}}$.
\end{lemma}
\begin{proof}
We have that
$$
\begin{pmatrix}
    \epsilon\\
    X
\end{pmatrix}
=
\begin{pmatrix}
I\\
(I-B)^{-1}
\end{pmatrix} \Omega^{1/2}
Z \text{ with } Z=(Z_1,...,Z_p) \sim \cc{N}(0,I),
$$
that is, $\epsilon$ and $X$ are jointly normal with covariance
$$
\Gamma :=
\begin{pmatrix}
\Omega & P\\
P^T & \Sigma
\end{pmatrix}
$$
where $P := \Omega(I-B)^{-T}$ and $\Sigma := (I-B)^{-1} \Omega (I-B)^{-T}$. Denote $S := \supp{B_{i:}}$; the joint normality implies that $\epsilon \mid X_S$ follows a normal distribution with mean
\begin{align}
\label{eq:mean}
    \bb{E}[\epsilon_i] + P_{i,S} \Sigma_{S,S}^{-1} (x - \bb{E}[X_S])
\end{align}
and variance
\begin{align}
    \label{eq:variance}
    \Omega_{ii} - P_{i,S} \Sigma_{S,S}^{-1}P_{S,i}.
\end{align}
Now, without loss of generality, assume $B$ is lower triangular with zeros on the diagonal. It follows that
\begin{align}
\label{eq:support}
    j < i \quad \forall i \in \preds, \forall j \in \supp{B_{i:}}.
\end{align}%
Since $B$ is lower triangular, so are $(I-B)$ and $(I-B)^{-1}$, and thus $P = \Omega (I-B)^{-T}$ is upper triangular. Therefore $P_{ij}=0$ for $j < i$, and together with \eqref{eq:support}, it follows that $P_{i,S} = 0$. Thus, from (\ref{eq:mean}, \ref{eq:variance}) we see that $\epsilon \mid X_S$ has mean zero and variance $\Omega_{ii}$, that is, $\epsilon_i \indep X_S$.
\end{proof}

\begin{lemma}(Same structure implies same coefficients)
\label{lemma:same_support}
Consider a model $(B, \Omega)$ with a single environment and let $(B', \Omega') \in \left[\left(B,\Omega\right)\right]$ be an equivalent model. Then,
$$\supp{B_{i:}} = \supp{B^\star_{i:}} \iff B_{i:} = B^\star_{i:}.$$
\end{lemma}
\begin{proof}
The ``$\impliedby$'' direction is trivial. For the ``$\implies$'' direction, define the random vectors
\begin{center}
    $X := (I-B)^{-1}\epsilon$ and $X' := (I-B')^{-1}\epsilon'$ with $\epsilon \sim \cc{N}(0, \Omega)$, $\epsilon' \sim \cc{N}(0, \Omega')$.
\end{center}
Let $S:=\supp{B_{i:}}=\supp{B'_{i:}}$; we have that
\begin{center}
$X_i = B_{iS}X_S + \epsilon_i$ and $X'_i = B'_{iS}X'_S + \epsilon'_i$.
\end{center}
Because $(B', \Omega') \in \left[\left(B,\Omega\right)\right]$, $X$ and $X'$ follow the same distribution, and it holds that for all $x_S \in \bb{R}^{|S|}$ that
$$\bb{E}[X_i \mid X_S = x_S] = \bb{E}[X'_i \mid X'_S = x_S].$$
Plugging in the expression for $X_i$ and $X'_i$ into the above, we arrive at
$$B_{iS}x_S + \bb{E}[\epsilon_i \mid X_S = x_S]
=
B'_{iS}x_S + \bb{E}[\epsilon'_i \mid X'_S = x_S],
$$
and by \autoref{lemma:independence}, at
$$B_{iS}x_S + \bb{E}[\epsilon_i]
=
B'_{iS}x_S + \bb{E}[\epsilon'_i].$$
Since $\bb{E}[\epsilon_i] = \bb{E}[\epsilon'_i] = 0$ the above results in the condition that $(B_{iS} - B'_{iS})x_S = 0$ for all $x_S \in \bb{R}^{|S|}$,
which is true if and only if $B_{iS} = B'_{iS}$, completing the proof.
\end{proof}

\subsection{Supporting Results for the Case $\abs{\cc{E}}>1$}
\label{ss:support_imec}

\begin{lemma}[Under faithfulness, the true model has the minimum number of edges]
\label{lemma:min_edges_imec}
Consider a model $\big(B,\{\Omega^e\}_{e \in \cc{E}}\big)$ resulting in a set of distributions which are faithful with respect to its underlying graph $\cc{G}(B)$. Then it holds that
$$\min_{\dag \in \bb{G}\left(B,\{\Omega^e\}_{e \in \cc{E}}\right)} \|\dag\|_{0} = \|\cc{G}(B)\|.$$
\end{lemma}
\begin{proof}
Let $\cc{E}^0 \subset \cc{E}$ be a singleton. By \autoref{lemma:min_edges_mec} we have that $\cc{G}(B)$ attains the minimum number of edges in $\bb{G}\left(B,\{\Omega^e\}_{e \in \cc{E}^0}\right)$. It is trivial to see\footnote{See also \autoref{corr:envs_superset}.} that
$\bb{G}\left(B,\{\Omega^e\}_{e \in \cc{E}}\right) \subseteq \bb{G}\left(B,\{\Omega^e\}_{e \in \cc{E}^0}\right)$; since $\cc{G}(B) \in \bb{G}\left(B,\{\Omega^e\}_{e \in \cc{E}}\right)$, it follows that $\cc{G}(B)$ also attains the minimum number of edges in $\bb{G}\left(B,\{\Omega^e\}_{e \in \cc{E}}\right)$.
\end{proof}


\begin{lemma}[structure of $(I-\tilde{B})(I-B)^{-1}$ for different parents]
\label{lemma:structure_diff}
Let $(B, \{\Omega^e\}_{e\in \cc{E}})$ be a model with a single environment, and let $(\tilde{B}, \{\tilde{\Omega}^e\}_{e\in \cc{E}}) \in \left[\left(B,\Omega\right)\right]$ be an equivalent model such that the diagonal entries of the matrix $M:=(I-\tilde{B})(I-B)^{-1}$ are all non-zero. Then, it holds that
$$\supp{B_{i:}} \neq \supp{\tilde{B}_{i:}} \implies \abs{\supp{\left[(I-\tilde{B})(I-B)^{-1} \right]_{i:}}} > 1.$$
\end{lemma}
\begin{proof}
Let $M:=(I-\tilde{B})(I-B)^{-1}$ and assume that
\begin{center}
    $\supp{B_{i:}} \neq \supp{\tilde{B}_{i:}}$ but $\abs{\supp{M_{i:}}} \leq 1$.
\end{center} 
Since $M$ has full rank it cannot be that $\abs{\supp{M_{i:}}} = 0$. Since the diagonal entries of $M$ cannot be zero, it must then be that $M_{i:} = [(I-\tilde{B})(I-B)^{-1}]_{i:} = \alpha e_i^T$. However, by \autoref{lemma:inverse}, this means that $(I-\tilde{B})_{i:} = \alpha (I-B)_{i:}$ and thus $\supp{B_{i:}} = \supp{\tilde{B}_{i:}}$, resulting in a contradiction.
\end{proof}

\begin{lemma}[I-MAPs are distribution equivalent]
\label{lemma:imaps_idec}
Consider a model $\big(B,\{\Omega^e\}_{e \in \cc{E}}\big)$ with underlying graph $\cc{G}(B)$. Let $\cc{I} = \bb{I}\big(\{\Omega^e\}_{e \in \cc{E}}\big)$ be the indices of variables that have received an intervention in at least one of the environments in $\cc{E}$. Let $\dag$ be an independence map of $\cc{G}(B)$ such that $\pa_{\cc{G}(B)}(i) \subseteq \pa_\dag(i)$ for all $i \in \cc{I}$. Then, $\dag \in \bb{G}\left(B,\{\Omega^e\}_{e \in \cc{E}}\right)$.
\end{lemma}
\begin{proof}
We want to show that there exists a model $(\tilde{B}, \{\tilde{\Omega}^e\}_{e \in \cc{E}})$ such that $\tilde{B} \comp \dag$ and
\begin{align}
\label{eq:dist_condition}
(I-B)^{-1}\Omega^e(I-B)^{-T} = (I-\tilde{B})^{-1}\tilde{\Omega}^e(I-\tilde{B})^{-T} \text{ for all }e \in \cc{E}.
\end{align}
That is equivalent to the condition that for all $e \in \cc{E}$
$$\tilde{\Omega}^e := (I-\tilde{B})(I-B)^{-1}\Omega^e(I-B)^{-T}(I-\tilde{B})^T$$
is a diagonal matrix with positive entries. By \autoref{lemma:orthogonal_rows_rank}, in turn this is equivalent to
$(I-\tilde{B})(I-B)^{-1}(\Omega^e)^{1/2}$
being a full-rank matrix with orthogonal row vectors, for all $e \in \cc{E}$. Now, let $\cc{E}^0 \subset \cc{E}$ be a singleton and let $\Omega^0 := \diag{\omega_1^2, ..., \omega_p^2}$ be the noise-term variances associated to it. By \autoref{lemma:imaps_dec}, there exists $\tilde{B}$ such that $\tilde{B} \comp \dag$, and $(I-\tilde{B})(I-B)^{-1}(\Omega^0)^{1/2}$ is a full-rank matrix with orthogonal row vectors. Now, for every $e \in \cc{E}$, we can rewrite
\begin{align}
\label{eq:requirement}
    (I-\tilde{B})(I-B)^{-1}(\Omega^e)^{1/2} = (I-\tilde{B})(I-B)^{-1}(\Omega^0)^{1/2}(\Omega^0)^{-1/2}(\Omega^e)^{1/2}.
\end{align}
If we let $\gamma_1^2,...,\gamma_p^2$ be the diagonal entries of $\Omega^e$, we can write
$$
(\Omega^0)^{-1/2}(\Omega^e)^{1/2} = \diag{\abs{\frac{\gamma_1}{\omega_1}},...,\abs{\frac{\gamma_p}{\omega_p}}},
$$
\ie a diagonal matrix that constitutes a positive scaling of some of the columns of $(I-\tilde{B})(I-B)^{-1}(\Omega^0)^{1/2}$. Note\footnote{Because $\cc{I} := \bb{I}\big(\{\Omega^e\}_{e \in \cc{E}}\big)$ and \eqref{eq:i_operator}.} that this scaling can only occur in those columns with index in $\cc{I}$. Since we have assumed that $\pa_{\cc{G}(B)}(i) \subseteq \pa_\dag(i)$ for all $i \in \cc{I}$, by \autoref{corr:coefficients} we have that $B_{i:} = \tilde{B}_{i:}$ for $i \in \cc{I}$, and by \autoref{lemma:structure}, these columns have exactly one non-zero element. Scaling them leaves the matrix \eqref{eq:requirement} with full rank and orthogonal row vectors, that is, for all $e \in \cc{E}$
$$
\tilde{\Omega}^e := (I-\tilde{B})(I-B)^{-1}\Omega^e(I-B)^{-T}(I-\tilde{B})^T
$$
is a diagonal matrix with positive elements. Thus, \eqref{eq:dist_condition} holds and since $\tilde{B} \comp \dag$ we have completed the proof.
\end{proof}

\begin{lemma}[Targets have the same parents]
\label{lemma:targets_parents}
Consider a model $(B,\{\Omega^e\}_{e \in \cc{E}})$ satisfying Assumptions \ref{assm:int_heter} and \ref{assm:model_truth}, and let $\cc{I} := \bb{I}(\{\Omega^e\}_{e \in \cc{E}})$ be the indices of variables that have received interventions in $\cc{E}$. Let $(\tilde{B}, \{\tilde{\Omega}^e\}_{e \in \cc{E}}) \in [(B,\{\Omega^e\}_{e \in \cc{E}})]$ be a distribution equivalent model. Then, $B_{i:} = \tilde{B}_{i:}$ for all $i \in \cc{I}$.
\end{lemma}
\begin{proof}
If $\abs{\cc{E}} = 1$, then $\cc{I} = \emptyset$ and we are done. For $\abs{\cc{E}}>1$, see that because $(\tilde{B}, \{\tilde{\Omega}^e\}_{e \in \cc{E}}) \in [(B,\{\Omega^e\}_{e \in \cc{E}})]$ it holds that, for all $e \in \cc{E}$,
$$\tilde{\Omega}^e = (I-\tilde{B})(I-B)^{-1}\Omega^e(I-B)^{-T}(I-\tilde{B})^T,$$
which implies that
$M^e:= (I-\tilde{B})(I-B)^{-1}(\Omega^e)^{1/2}$
is a full-rank matrix with orthogonal row vectors (\autoref{lemma:orthogonal_rows_rank}). Let $S:= \{i : B_{i:} = \tilde{B}_{i:}\}$; then,
$M^e_{i:} = (\Omega^e)^{1/2}e_i^T$ and $M^e_{:i} = (\Omega^e)^{1/2}e_i$ for all $i \in S$.
In other words, for any $e \in \cc{E}$ we have that $M^e$ is a block-diagonal matrix with two blocks
\begin{enumerate}[label=\roman*)]
\item $M^e_{S S}$, a diagonal matrix; and
\item $M^e_{S^\perp S^\perp}$, where $S^\perp := [p] \setminus S$.
\end{enumerate}
Because of this block-diagonal form, the rows of one block are orthogonal to those of the other. Thus, $M^e$ has orthogonal row vectors if and only if so do both blocks. This is immediately true for the diagonal block (i).
Take some $f\in \cc{E} : f \neq e$ and let $D := (\Omega^f)^{-1/2}(\Omega^e)^{1/2}$. We can rewrite
$$M^e = (I-\tilde{B})(I-B)^{-1}(\Omega^f)^{1/2}(\Omega^f)^{-1/2}(\Omega^e)^{1/2} = M^fD$$
and $M^e_{S^\perp S^\perp} = M^f_{S^\perp  S^\perp}D_{S^\perp S^\perp}.$
Because $(\tilde{B}, \{\tilde{\Omega}^e\}_{e \in \cc{E}}) \in [(B, \{\Omega^e\}_{e \in \cc{E}})]$, by \autoref{lemma:orthogonal_rows_rank} we have that $M^e$ and $M^f$ have full rank and orthogonal row vectors, and because of their block structure, so do $M^e_{S^\perp S^\perp}$ and $M^f_{S^\perp S^\perp}$. From \autoref{assm:model_truth} and \autoref{lemma:structure_diff}, we additionally have that $M^f_{S^\perp S^\perp}$ has more than one non-zero element per row. Now, assume $\cc{I} \cap S^\perp \neq \emptyset$, that is, $\exists i \in \cc{I}$ such that $B_{i:} \neq \tilde{B}_{i:}$. Then, by intervention-heterogeneity (\autoref{assm:int_heter}), $D_{S^\perp S^\perp}$ has at least one diagonal entry different from all others. These facts allow us to apply \autoref{lemma:scaling}, by which $M^e_{S^\perp S^\perp} = M^f_{S^\perp  S^\perp}D_{S^\perp S^\perp}$ does not have orthogonal rows, arriving at a contradiction. This completes the proof.
\end{proof}

\subsection{Main Results of Section \ref{s:equivalence}}
\label{ss:main_results}

Now we provide the proofs for \autoref{prop:imec}, \autoref{prop:imec_faithfulness}, and \autoref{prop:imec_full}.

\propimec*
\begin{proof}
Let $\dag \in \imec(\cc{G}(B))$. By \autoref{def:i_equiv}, this implies that $\dag$ is an independence map of $\cc{G}(B)$, and that $\pa_{\cc{G}(B)}(i) = \pa_\dag(i)$ for all $i \in \cc{I}$. Thus, by \autoref{lemma:imaps_idec}, $\dag \in \bb{G}\left(B,\{\Omega^e\}_{e \in \cc{E}}\right)$, completing the proof.
\end{proof}

\propimecfaithfulness*
\begin{proof}
Let $\dag' \in \imec(\cc{G}(B))$. By \autoref{prop:imec}, $\dag' \in \bb{G}\left(B,\{\Omega^e\}_{e \in \cc{E}}\right)$. By \autoref{lemma:min_edges_imec} we have that $\min_{\dag \in \bb{G} \left (B, \Omega \right)} \|\dag\|_0 = \|\cc{G}(B)\|_0$; since $\cc{G}(B)$ and $\dag'$ are Markov equivalent, they have the same number of edges, and thus $\dag' \in \arg\min_{\dag \in \bb{G}\left(B,\{\Omega^e\}_{e \in \cc{E}}\right)} \|\dag\|_0$, completing the proof.
\end{proof}

\propimecfull*
\begin{proof}
We proceed by showing both sets are contained in each other.
\paragraph{$\subseteq$} Follows from \autoref{prop:imec_faithfulness}.
    
\paragraph{$\supseteq$} Let $\tilde{\dag} \in \arg\min_{\dag \in \bb{G}\left(B,\{\Omega^e\}_{e \in \cc{E}}\right)} \|\dag\|_0$. We begin by showing that $\tilde{\dag}$ and $\cc{G}(B)$ have the same skeleton and v-structures. Let $\cc{E}^0 \subset \cc{E}$ be a singleton and $\Omega^0$ be the noise-term variances associated with it. Clearly (c.f.\ \autoref{corr:envs_superset}), $\tilde{\dag} \in \bb{G}\left(B,\Omega^0\right)$. Because we are under faithfulness, by Lemmas \ref{lemma:min_edges_mec} and \ref{lemma:min_edges_imec}
$$\min_{\dag\in \bb{G}\left(B,\Omega^0\right)} \|\dag\|_0 = \|\cc{G}(B)\|_0 = \min_{\dag\in \bb{G}\left(B,\{\Omega^e\}_{e \in \cc{E}}\right)} \|\dag\|_0,$$
and thus $\tilde{\dag} \in \arg\min_{\dag\in \bb{G}\left(B,\Omega^0\right)} \|\dag\|_0$. By \autoref{prop:mec_equals_dec},
$\arg\min_{\dag\in \bb{G}\left(B,\Omega^0\right)} \|\dag\|_0 = \text{MEC}(\cc{G}(B))$
and thus $\tilde{\dag}$ and $\cc{G}(B)$ have the same skeleton and v-structures.

If $\cc{E}$ is a singleton, then $\cc{I} = \emptyset$ and we have completed the proof. For the remaining cases, we will now show that the variables in the intervention set $\cc{I}$ have the same parents in $\tilde{\dag}$ and $\cc{G}(B)$. Since $\tilde{\dag} \in \bb{G}\left(B,\{\Omega^e\}_{e \in \cc{E}}\right)$, there exists a model $(\tilde{B}, \{\tilde{\Omega}^e\}_{e \in \cc{E}}) \in [(B, \{\Omega^e\}_{e \in \cc{E}})]$ such that $\tilde{B} \comp \tilde{\dag}$. First, by \autoref{lemma:targets_parents}, it holds that $\tilde{B}_{i:} = B_{i:}$ for all $i \in \cc{I}$. Second, since $\tilde{\dag} \in \arg\min_{\dag\in \bb{G}\left(B,\{\Omega^e\}_{e \in \cc{E}}\right)} \|\dag\|_0$, it must hold that $\cc{G}(\tilde{B}) = \tilde{\dag}$ or we would have a contradiction\footnote{Otherwise $\cc{G}(\tilde{B})$ would have less edges than $\tilde{\dag}$ and thus $\tilde{\dag} \notin \arg\min_{\dag\in \bb{G}\left(B,\{\Omega^e\}_{e \in \cc{E}}\right)} \|\dag\|_0$.}. From these two facts, it follows that
$\pa_{\tilde{\dag}}(i) = \pa_{\cc{G}(\tilde{B})}(i) = \pa_{\cc{G}(B)}(i)$ for all $i \in \cc{I}$,
completing the proof.
\end{proof}
\subsection{Proofs for \autoref{ss:other_classes}}
\label{s:proofs_other_classes}

\lemmaaugmented*
\begin{proof}
    ($1 \implies 2$) $\dag_1$ and $\dag_2$ have the same skeleton (c.f. \autoref{def:i_equiv}); thus, so do $\dag_1^\cc{I}$ and $\dag_2^\cc{I}$.
    Furthermore, $\dag_1$ and $\dag_2$ have the same v-structures (c.f. \autoref{def:i_equiv}), and the new v-structures in their augmented graphs are $\{j \to i \leftarrow F_i \mid i \in \cc{I} \land j \in \pa_{\dag_1}(i)\}$ and $\{j \to i \leftarrow F_i \mid i \in \cc{I} \land j \in \pa_{\dag_2}(i)\}$, respectively. Because $\pa_{\dag_1}(i) = \pa_{\dag_2}(i)\; \forall i \in \cc{I}$, both sets of v-structures are the same, and $\dag_1^\cc{I}$ and $\dag_2^\cc{I}$ are Markov equivalent.

    ($2 \implies 1$) Because $\dag_1^\cc{I}$ and $\dag_2^\cc{I}$ have the same skeleton, so do $\dag_1$ and $\dag_2$, satisfying condition (i) of \autoref{def:i_equiv}. Assume now that condition (ii) is false, that is, that $\dag_1$ and $\dag_2$ have different v-structures. Since the new nodes in the augmented graphs $\dag_1^\cc{I}$ and $\dag_2^\cc{I}$ are source nodes, adding them can only create new v-structures. Thus, $\dag_1^\cc{I}$ and $\dag_2^\cc{I}$ would also have different v-structures, which contradicts that they are Markov equivalent. Finally, assume now that there exists $i\in\cc{I}$ such that $\pa_{\dag_1}(i) \neq \pa_{\dag_2}(i)$; then, the addition of a new node and edge $F_i \to i$ would create different v-structures in the augmented graphs $\dag_1^\cc{I}$ and $\dag_2^\cc{I}$, which is again a contradiction; this proves condition (iii) of \autoref{def:i_equiv}, completing the proof.
\end{proof}

\lemmahauser*
\begin{proof}
Let $\tilde{\dag} \in \imec(\dag)$. By definition $\dag$ and $\tilde{\dag}$ are Markov equivalent, satisfying (i) in \autoref{thm:hauser}. Since $\forall i \in \cc{I} : \pa_{\dag}(i) = \pa_{\tilde{\dag}}(i)$, the graphs $\tilde{\dag}^{(H)}$ and $\dag^{(H)}$ have the same skeleton for all $H \in \cc{H}$, satisfying (ii) and completing the proof.
\end{proof}

\lemmayang*
\begin{proof}
Let $\tilde{\dag} \in \imec(\dag)$ and let $\tilde{\dag}^\cc{Y}$ and $\dag^\cc{Y}$ be the corresponding interventional DAGs. By \autoref{def:i_equiv}, $\tilde{\dag}$ and $\dag$ have the same skeleton; because the edges introduced in the interventional DAGs depend only on $\cc{Y}$, $\tilde{\dag}^\cc{Y}$ and $\dag^\cc{Y}$ have the same additional edges and also have the same skeleton. We now show that they also have the same v-structures. Without loss of generality, suppose there exists a v-structure in $\dag^\cc{Y}$ which is not in $\tilde{\dag}^\cc{Y}$. Because the nodes $Y \in \cc{Y}$ do not have any parents, they can only appear in the tails of the v-structure, which can therefore be of three types:
\begin{enumerate}[label=\alph*)]
    \item $i\to k \leftarrow j$ where $i,k,j \in [p]$;
    \item $Y_1 \to k \leftarrow Y_2$ where $k \in [p]$ and $Y_1,Y_2 \in \cc{Y}$; or
    \item $Y \to k \leftarrow j$ where $Y \in \cc{Y}$ and $j,k \in [p]$.
\end{enumerate}
By \autoref{def:i_equiv}, $\tilde{\dag}$ and $\dag$ have the same v-structures, which rules out (a). Furthermore, because the edges from nodes in $\cc{Y}$ to nodes in $[p]$ are the same for both $\tilde{\dag}^\cc{Y}$ and $\dag^\cc{Y}$, and they have the same skeleton, a v-structure like (b) would appear in both graphs. We now consider (c): since both graphs have the same skeleton and nodes in $\cc{Y}$ are always source nodes, the only option is that (c) appears instead as $Y \to k \to j$ in $\tilde{\dag}^\cc{Y}$. However, the edge from $Y \to k$ means that $k \in Y$, and by extension $k \in \cc{I}$; thus, by \autoref{def:i_equiv} $k$ would have the same parents in $\tilde{\dag}$ and $\dag$, arriving at a contradiction. Thus, we have shown that the interventional DAGs $\tilde{\dag}^\cc{Y}$ and $\dag^\cc{Y}$ have the same skeleton and v-structures, which means that $\tilde{\dag} \in \ymec(\dag)$, completing the proof.
\end{proof}

\subsection{Proofs of the Score Properties}
\label{s:proofs_score_properties}

We begin by stating some supporting results.

\begin{lemma}[$\cc{I}$-equivalent graphs entail the same distributions]
\label{lemma:graphs_dists}
Consider a set of intervention targets $\cc{I} \subseteq [p]$ and two graphs $\dag$ and $\dag'$ such that $\dag \sim_\cc{I}\dag'$. Then, for every model $(B, \{\Omega^e\}_{e \in \cc{E}})$ such that $B \comp \dag$ and $\bb{I}(\{\Omega^e\}_{e \in \cc{E}}) = \cc{I}$, there exists a model $(\tilde{B}, \{\tilde{\Omega}^e\}_{e \in \cc{E}})$ such that $\tilde{B} \comp \dag'$, $\bb{I}(\{\tilde{\Omega}^e\}_{e \in \cc{E}}) = \cc{I}$, and
$$(I-B)^{-1}\Omega^e(I-B)^{-T} = (I-\tilde{B})^{-1}\tilde{\Omega}^e(I-\tilde{B})^{-T}.$$
\end{lemma}
\begin{proof}
$\dag$ is an independence map of $\cc{G}(B)$ as it contains all it edges. Because $\dag$ and $\dag'$ are $\cc{I}$-equivalent, $\dag'$ is also an independence map of $\cc{G}(B)$, and since
\begin{equation}
\label{eq:subset_parents}
\pa_{\cc{G}(B)}(i) \subseteq \pa_\dag(i) = \pa_{\dag'}(i),
\end{equation}
by \autoref{lemma:imaps_idec} there exists $(\tilde{B}, \{\tilde{\Omega}^e\}_{e \in \cc{E}})$ such that $\tilde{B} \comp \dag'$ and
\begin{equation}
    \label{eq:covariances}
    (I-B)^{-1}\Omega^e(I-B)^{-T} = (I-\tilde{B})^{-1}\tilde{\Omega}^e(I-\tilde{B})^{-T}.
\end{equation}

Now we will show that $\cc{I} = \bb{I}(\{\tilde{\Omega}^e\}_{e \in \cc{E}}) = \bb{I}(\{\Omega^e\}_{e \in \cc{E}})$. As for notation, let $M:=(I-\tilde{B})(I-B)^{-1}$. Together with \autoref{corr:coefficients}, \eqref{eq:subset_parents} implies that $B_{i:} = \tilde{B}_{i:}$ and therefore\footnote{Pick any $g \in \cc{E}$: see that by \autoref{corr:envs_superset} $(\tilde{B}, \tilde{\Omega}^g) \in [(B, \Omega^g)]$ and apply \autoref{lemma:structure}.} $M_{i:} = e_i^T$ and $M_{:i} = e_i$ for $i \in \cc{I}$. First we will show that $i\in\bb{I}(\Omega^e\}_{e \in \cc{E}} \implies i \in \bb{I}(\{\tilde{\Omega}^e\}_{e \in \cc{E}})$. Note that (c.f. \autoref{eq:i_operator})
\begin{equation}
\label{eq:contained_in}
    i \in \cc{I} = \bb{I}(\Omega^e\}_{e \in \cc{E}}) \iff \exists e,f \in \cc{E} : [\Omega^e(\Omega^f)^{-1}]_{ii} \neq 1.
\end{equation}
We can rewrite \eqref{eq:covariances} as $\tilde{\Omega}^e = M\Omega^eM^T,$
and thus $\tilde{\Omega}^e(\tilde{\Omega}^f)^{-1} = M \Omega^e(\Omega^f)^{-1} M^{-1}$. For $i \in \cc{I}$ we have that
\begin{align*}
[\tilde{\Omega}^e(\tilde{\Omega}^f)^{-1}]_{ii}
&= M_{i:} \Omega^e(\Omega^f)^{-1} M^{-1}_{:i}\\
&= e_i^T \Omega^e(\Omega^f)^{-1} M^{-1}_{:i}  \quad \text{(} M_{i:} = e_i^T, \text{ see above)}\\
&= [\Omega^e(\Omega^f)^{-1}]_{ii} e_i^T M^{-1}_{:i} \quad \text{(} \Omega^e, \Omega^f \text{ are diagonal)}\\
&= [\Omega^e(\Omega^f)^{-1}]_{ii} M_{i:} M^{-1}_{:i}\\
&= [\Omega^e(\Omega^f)^{-1}]_{ii} \neq 1,    
\end{align*}
and thus by \eqref{eq:contained_in} $i \in \bb{I}(\{\tilde{\Omega}^e\}_{e \in \cc{E}})$. Now we will show that $j\notin\bb{I}(\Omega^e\}_{e \in \cc{E}} \implies j \notin \bb{I}(\{\tilde{\Omega}^e\}_{e \in \cc{E}})$. Note that
\begin{equation}
    \label{eq:not_contained_in}
    j \notin \cc{I} = \bb{I}(\Omega^e\}_{e \in \cc{E}}) \iff [\Omega^e - \Omega^f]_{jj} = 0 \; \forall e,f \in \cc{E}.
\end{equation}
Now, pick any $e,f \in \cc{E}$ and see from \eqref{eq:covariances} that $\tilde{\Omega}^e - \tilde{\Omega}^f = M(\Omega^e - \Omega^f)M^T$.
Because $M_{:i} = e_i$ for $i \in \cc{I}$, it follows that $M_{ji} = 0$ for $j \notin \cc{I}$. Thus, we have that
$$[\tilde{\Omega}^e - \tilde{\Omega}^f]_{jj}
= M_{j:}(\Omega^e - \Omega^f)M^T_{:j}
= \sum_{i\in\cc{I}}M_{ij}[\Omega^e - \Omega^f]_{ii}
+ \sum_{k\notin\cc{I}}M_{kj}[\Omega^e - \Omega^f]_{kk}
= 0,$$
as by \eqref{eq:not_contained_in} $[\Omega^e - \Omega^f]_{kk} = 0$ in the second sum.
Because we picked $e,f$ arbitrarily, this means that $[\tilde{\Omega}^e - \tilde{\Omega}^f]_{jj} = 0$ for all $e,f \in \cc{E}$ and thus, by the same argument as in \eqref{eq:not_contained_in}, $j \notin \bb{I}(\{\tilde{\Omega}^e\}_{e \in \cc{E}})$. Together with the previous results, this shows that $\bb{I}(\{\tilde{\Omega}^e\}_{e \in \cc{E}}) = \bb{I}(\{\Omega^e\}_{e \in \cc{E}}) = \cc{I}$, completing the proof.
\end{proof}

\propscoreproperties*
\paragraph{Proof of score equivalence}
Consider a set $\cc{I} \subseteq [p]$ of intervention targets and two graphs $\dag$ and $\dag'$ such that $\dag \sim_\cc{I}\dag'$. We want to show that they attain the same score over any sample, that is $S(\dag, \cc{I}) = S(\dag', \cc{I})$. Note that we can rewrite
\begin{equation}
\label{eq:covariance_reduction}
    \log \p\left(\boldsymbol{X}^e; {B}, \Omega^e\right) = -n_e\ln\mathrm{det}({K}^e) - n_e\mathrm{tr}({K}^e\hat{\Sigma}^e),
\end{equation}
where $\hat{\Sigma}^e$ is the sample covariance for environment $e$, and $K^e := [(I-B)^{-1}\Omega^e(I-B)^{-T}]^{-1}$ is the inverse of the covariance matrix entailed by the model for environment $e$. Thus, we can rewrite the score \eqref{eq:score} as
\begin{equation}
\label{eq:score_trace}
S(\dag, \cc{I}) = 
    \max_{\substack{
    B \comp \dag\\    
    \text{diagonal } \{\Omega^e\}_{e \in \cc{E}}\\
    \text{s.t. }\bb{I}(\{\Omega^e\}_{e \in \cc{E}}) = \cc{I}
    }}
    \quad
    \sum_{e \in \cc{E}} n_e\ln\mathrm{det}({K}^e) - n_e\mathrm{tr}({K}^e\hat{\Sigma}^e) - \lambda\text{DoF}(\dag, \cc{I}).
\end{equation}
Now, let $(B, \{\Omega^e\}_{e \in \cc{E}})$ be the model parameters which maximize the score for $\dag$. By \autoref{lemma:graphs_dists}, there exists $(\tilde{B}, \{\tilde{\Omega}^e\}_{e \in \cc{E}})$ for which $\tilde{B} \sim \dag'$ and
$$(I-B)^{-1}\Omega^e(I-B)^{-T} = (K^e)^{-1} = (I-\tilde{B})^{-1}\Omega^e(I-\tilde{B})^{-T}.$$
Furthermore, because $\dag$ and $\dag'$ have the same number of edges, $\text{DoF}(\dag, \cc{I}) = \text{DoF}(\dag', \cc{I})$ and $S(\dag, \cc{I}) \leq S(\dag', \cc{I})$. If we apply the same reasoning to $\dag'$ and its maximizing parameters, we arrive at $S(\dag, \cc{I}) \geq S(\dag', \cc{I})$. It must hold that $S(\dag, \cc{I}) = S(\dag', \cc{I})$, completing the proof.

\paragraph{Proof of decomposability}
Note that if $B \in \bb{R}^{p \times p}$ has zeros on the diagonal and is lower triangular up to a permutation of rows and columns, then $(I-B)$ can be expressed as $PLP$ for some permutation matrix $P$ and some lower-triangular matrix $L$ with ones on the diagonal. Thus, $\mathrm{det}((I-B)^T) = \mathrm{det}((I-B)) = 1$, and we can rewrite the logarithm term in \eqref{eq:covariance_reduction} as
$$\ln\mathrm{det}({K}^e) = \ln[\mathrm{det}((I-B)^T)\mathrm{det}((\Omega^e)^{-1})\mathrm{det}(I-B)] = \ln(\prod_{i\in[p]} 1 / \Omega^e_{ii}) = -\sum_{i \in [p]} \ln \Omega^e_{ii}$$
For the trace term, we have that
$$\mathrm{tr}({K}^e\hat{\Sigma}^e) = \mathrm{tr}((I-B)^T(\Omega^e)^{-1}(I-B)\hat{\Sigma}^e) = \mathrm{tr}((\Omega^e)^{-1}(I-B)\hat{\Sigma}^e(I-B)^T),$$
where the last equality stems from the invariance of the trace to cyclic permutations. By rewriting the trace of products as the sum of entry-wise products, we have that
$$\mathrm{tr}({K}^e\hat{\Sigma}^e)=\sum_{i \in [p]}\sum_{j \in [p]} 1/\Omega^e_{ij} [(I-B)\hat{\Sigma}^e(I-B)^T)]_{ij} = \sum_{i \in [p]}1/\Omega^e_{ii} [(I-B)\hat{\Sigma}^e(I-B)^T)]_{ii},$$
and thus $\mathrm{tr}({K}^e\hat{\Sigma}^e)=\sum_{i \in [p]}1/\Omega^e_{ii} (I-B)_{i:}\hat{\Sigma}^e(I-B)^T)_{:i}$. Together with the new expression for the logarithm term, we can rewrite the score \eqref{eq:score_trace} as
\begin{align*}
\label{eq:score_trace}
S(\dag, \cc{I}) &= 
    \sum_{i\in[p]}\max_{\substack{
    B_{i:},\{\Omega_{ii}^e\}_{e \in \cc{E}}\\
    \text{s.t.}\\
    \supp{B_{i:}} \subseteq \pa_G(i)\\
    \Omega_{ii} > 0
    }}
    \sum_{e \in \cc{E}} -n_e\left[\ln \Omega^e_{ii} + 1 / \Omega^e_{ii}(I-B)_{i:}\hat{\Sigma}^e(I-B)^T)_{:i}\right] - \lambda\text{DoF}(i,\dag, \cc{I}),
\end{align*}
where $\text{DoF}(i,\dag, \cc{I}) = \abs{\pa_G(i)} + \mathrm{max}(\abs{\cc{E}}\boldsymbol{1}_{\{i \in \cc{I}\}}, 1)$ are the number of parameters estimated for variable $i$. In other words, the score decouples into a sum of terms depending only on each variable and its parents, which completes the proof.

\paragraph{Proof of score consistency}
Throughout, we consider $p$ fixed and $n_e \to \infty$ for every $e \in \mathcal{E}$.\\
For every $\dag \in \simec(\dag^\star)$, let $(B,\{\Omega^e\}_{e \in \cc{E}})$ be the unique connectivity matrix and noise variances that specify the population covariance matrices $\{\Sigma^e\}_{e \in \cc{E}}$. In our analysis, 
the parameter space $(B,\{\Omega^e\}_{e \in \mathcal{E}}$ {is assumed} to be compact. For notational ease, we denote the constraint set for the parameters $(B,\{\Omega^e\}_{e \in \cc{E}})$ as $\Theta$. {Assuming such a compactness constraint}
enables uniform convergence of M-estimators \citep{van2000empirical}. As an example of a compactness constraint, let $\tau_1,\tau_2,\tau_3> 0$ be scalars where for every such $(B,\{\Omega^e\}_{e \in \cc{E}})$, $\|(I-B)\|_{F} \leq \tau_1$, $\min_{i,e}|\Omega^e_{ii}| \geq \tau_2$, and $\max_{i,e}|\Omega^e_{ii}| \leq \tau_3$. Then, 
the space of connectivity matrices and noise variances {has}
the additional constraints $\|(I-B)\|_{F} \leq \tau_1$, $\min_{i,e}|\Omega^e_{ii}| \geq \tau_2$ and $\max_{i,e}|\Omega^e_{ii}| \leq \tau_3$ in the score function. We 
denote the score of a DAG $G$ and intervention set $\cc{I}$ when constrained to the compact parameter space $\Theta$ as $S(\dag,\cc{I})$. 

For any $\dag$ and intervention set $\cc{I}$, we let $S^\star(G,\cc{I})$ be the score in population, namely:
$$S^\star(\dag,\cc{I}) =\max_{\substack{
    B \comp \dag\\    
    \text{diagonal } \{\Omega^e\}_{e \in \cc{E}}\\(B,\{\Omega^e\}_{e \in \cc{E}})\in\Theta\\
    \text{s.t. }\bb{I}(\{\Omega^e\}_{e \in \cc{E}}) = \cc{I}}} \sum_{e\in\mathcal{E}} \pi^{e,\star}\left[-\ln \mathrm{det}(\Omega^e) - \mathrm{tr}((I-B)^T(\Omega^e)^{-1}(I-B)\Sigma^{e,\star})\right],$$
where $\Sigma^{e,\star}$ represents the true covariance in environment $e$ and $\pi^{e,\star} = {\lim_{\mbox{all}\ n_e \to \infty}} \frac{n^e}{\sum_{e \in \mathcal{E}}n^e}$. Notice that the minimizers $\argmin_{\dag,\cc{I}} S^\star(\dag,\cc{I})$ are the distributionally equivalent models to the population model $(B^\star,\{\Omega^{e,\star}\}_{e \in \cc{E}})$.


First, we show that for every $\dag$ and intervention set $\cc{I}$, the score function $S(\dag,\cc{I}) \overset{p}{\to} S^\star(\dag,\cc{I})$ as $n_e \to \infty$ for every $e \in \mathcal{E}$. This follows by the compactness of the parameter space leading to
uniform convergence of M-estimators,
and that $\lambda \to 0$ \citep{van2000empirical}. For more details, note that:
\begin{equation}
\begin{aligned}
S(\dag,\cc{I})+\lambda\text{DoF}(\dag,\cc{I}) = \max_{(B,\{\Omega^e\}_{e \in \cc{E}})\in\Theta}  &\sum_{e \in \mathcal{E}}\pi^{e,\star}\left[-\ln\mathrm{det}(\Omega^e) - \mathrm{tr}((I-B)^T(\Omega^e)^{-1}(I-B)\Sigma^{e,\star})\right]\\&- \mathrm{tr}((I-B)^T(\Omega^e)^{-1}(I-B)(\hat{\pi}^e\hat{\Sigma}^e-{\pi}^{e,\star}\Sigma^{e,\star}))\\&- (\hat{\pi}^e-\pi^{e,\star})\ln\mathrm{det}(\Omega^e)
    \end{aligned}
\label{eqn:temp}
\end{equation}
where $\hat{\pi}^e = \frac{n^e}{\sum_{e \in \mathcal{E}}n^e}$. From there, we have
\begin{equation}
\begin{aligned}
S(\dag,\cc{I}) &\geq  S^\star(\dag,\cc{I}) - \max_{(B,\{\Omega^e\}_{e \in \cc{E}})\in\Theta}
    \sum_{e \in \mathcal{E}} \mathrm{tr}\left((I-B)^T(\Omega^e)^{-1}(I-B)(\hat{\pi}^e\hat{\Sigma}^e-{\pi}^{e,\star}\Sigma^{e,\star}))\right)\\
    &-\max_{(B,\{\Omega^e\}_{e \in \cc{E}})\in\Theta}
    \sum_{e \in \mathcal{E}}(\hat{\pi}^e-\pi^{e,\star})\ln\mathrm{det}(\Omega^e)- \lambda~\text{DoF}(\dag,\cc{I}),\\
S(\dag,\cc{I}) &\leq S^\star(\dag,\cc{I}) + \max_{(B,\{\Omega^e\}_{e \in \cc{E}})\in\Theta}
    \sum_{e \in \mathcal{E}} \mathrm{tr}\left((I-B)^T(\Omega^e)^{-1}(I-B)(\hat{\pi}^e\hat{\Sigma}^e-{\pi}^{e,\star}\Sigma^{e,\star})\right)\\&+\max_{(B,\{\Omega^e\}_{e \in \cc{E}})\in\Theta}
    \sum_{e \in \mathcal{E}}(\hat{\pi}^e-\pi^{e,\star})\ln\mathrm{det}(\Omega^e).
\end{aligned}
\label{eqn:upper_lower_S}
\end{equation}
By the compactness constraint, 
\begin{gather*}
\max_{(B,\{\Omega^e\}_{e \in \cc{E}})\in\Theta}
    \sum_{e \in \mathcal{E}} \mathrm{tr}\left((I-B)^T(\Omega^e)^{-1}(I-B)(\hat{\pi}^e\hat{\Sigma}^e-{\pi}^{e,\star}\Sigma^{e,\star})\right) \to 0 \\
    \max_{(B,\{\Omega^e\}_{e \in \cc{E}})\in\Theta}
    \sum_{e \in \mathcal{E}}(\hat{\pi}^e-\pi^{e,\star})\ln\mathrm{det}(\Omega^e) \to 0
\end{gather*}    
    as $n^e \to \infty$. Furthermore, since $\lambda \to 0$ as $n^e \to \infty$ for every $e \in \mathcal{E}$, and $\text{DoF}(\dag,\cc{I})$ is bounded, we have that $\lambda\text{DoF}(\dag,\cc{I}) \to 0$. We can thus conclude  $S(\dag,\cc{I}) \to S^\star(\dag,\cc{I})$ in the infinite data limit for every environment. 
    

Suppose the regularization parameter $\lambda \to 0$ and is chosen such that for every $e\in \mathcal{E}$, it is above fluctuations due to sampling error:
\begin{equation*}
\begin{aligned}
\lambda &\gg \left|\max_{(B,\{\Omega^e\}_{e \in \cc{E}})\in\Theta}
    \sum_{e \in \mathcal{E}} \mathrm{tr}\left((I-B)^T(\Omega^e)^{-1}(I-B)(\pi^{e,\star}\Sigma^{e,\star}-\hat{\pi}^e\hat{\Sigma}^e)\right)\right|\\&+ \left|\max_{(B,\{\Omega^e\}_{e \in \cc{E}})\in\Theta}
    \sum_{e \in \mathcal{E}}(\hat{\pi}^e-\pi^{e,\star})\ln\mathrm{det}(\Omega^e)\right|.
\end{aligned}
\end{equation*}
With this choice of the regularization, it follows that among any distributionally equivalent models $(\dag,\cc{I})$ and $(\dag',\cc{I}')$, if $\text{DoF}(\dag,\cc{I}) < \text{DoF}(\dag,\cc{I})$, then, there exists $N$ such that for any $n^e \geq N$, $S(\dag,\cc{I}) > S(\dag',\cc{I}')$. This allows us to conclude that:
\begin{equation*}
\mathbb{P}\left(\argmax S(\dag,\cc{I}) = \argmin \text{DoF}(\dag,\cc{I}) \text{  subject-to  } \dag,\cc{I} \in \argmax S^\star(\dag,\cc{I})\right) \to 1,
\end{equation*}
as $n^e \to \infty$ for every $e \in \mathcal{E}$. In other words, in the infinite data limit, the maximizers of $\argmax_{\dag,\cc{I}}S(\dag,\cc{I})$ are given by:

\begin{equation}
\begin{aligned}
    \argmin_{\dag,\cc{I}}&~~\text{DoF}(\dag,\mathcal{I})\\\text{subject-to}&~~ B \comp \dag, \mathcal{I} = \mathbb{I}(\{\Omega^e\}_{e\in\cc{E}}),\text{ and}\\
    &~~\Sigma_e^\star = (I-B)^{-1}\Omega^e(I-B)^{-T} ~~\text{ for every }e\in\mathcal{E}.
    \end{aligned}
    \label{eq:optimal_params}
\end{equation}
Our remaining goal is to show that $\cc{I}^\star\text{-MEC}(\cc{G}(B^\star))$ is the set of minimizers in \eqref{eq:optimal_params}. By definition, $\Sigma_e^\star = (I-B^\star)^{-1}{\Omega^{e,\star}}(I-B^\star)^{-T}$ for every $e\in\mathcal{E}$. Since the model $(B^\star,\{\Omega^{e,\star}\}_{e\in\mathcal{E}})$ satisfies Assumptions \ref{assm:int_heter} and \ref{assm:model_truth}, we appeal to \autoref{lemma:targets_parents} to conclude that a feasible $B$ in \eqref{eq:optimal_params} satisfies $B_{i,:} = B^\star_{i,:}$ for all $i \in \mathcal{I}^\star$. Since  $(I-B^\star)^{-1}{\Omega^{e,\star}}(I-B^\star)^{-T} = (I-B)^{-1}\Omega^e(I-B)^{-T}$, the previous conclusion implies that $\Omega^{e,\star}_{ii} = \Omega^e_{ii}$ for all $i \in \mathcal{I}^\star$ and feasible $\{\Omega^e\}_{e\in\mathcal{E}}$. Thus, any feasible $\mathcal{I}$ satisfies $\mathcal{I} \supseteq \mathcal{I}^\star$. Furthermore, from \autoref{prop:imec_full} we know that:
\begin{equation}
\begin{aligned}
    \simec(\dag^\star) = \arg\min_{\dag}&~~\text{DoF}(\dag)\\\text{subject-to}&~~ B \sim \dag\\
    &~~\Sigma_e^\star = (I-B)^{-1}\Omega^e(I-B)^{-T} ~~\text{ for every }e\in\mathcal{E}.
    \end{aligned}
    \label{eq:optimal_params_b}
\end{equation}
Notice that $\text{DoF}(\dag,\mathcal{I}) = \text{DoF}(\dag)+\text{DoF}(\mathcal{I})$. Combining this the relation in \eqref{eq:optimal_params_b} and the fact that any feasible $\mathcal{I}$ in \eqref{eq:optimal_params} satisfies $\mathcal{I} \supseteq \mathcal{I}^\star$, we arrive at the desired result.

\section{Proof of Theorem~\ref{thm:outer_procedure}}
\label{s:proof_outer}
\begin{proof}
Throughout, we consider $p$ fixed and $n_e \to \infty$ for every $e \in \mathcal{E}$. We consider a similar compactness constraint as the one described in Section~\ref{s:proofs_score_properties}. 
Since $S^\star(\dag,\cc{I})$ is a maximum-likelihood objective, the maximizers $\argmax_{\dag,\cc{I}} S^\star(\dag,\cc{I})$ are the distributionally equivalent models to the population model $(B^\star,\{\Omega^{e,\star}\}_{e \in \cc{E}})$, that is the set: 
$$\{(\dag,\cc{I}): \exists (B,\{\Omega^e\}_{e \in \mathcal{E}}) \text{ such that }\\ \cc{I} = \mathbb{I}(\{\Omega^e\}_{e \in \mathcal{E}}) \text{ and }
\Sigma^\star_e = (I-B)^{-1}\Omega^e(I-B)^{-T}\}.$$
Note here that the tuple $(B^\star,\cc{I}^\star)$ belongs to this set so that $\max_{\dag,\cc{I}} S^\star(\dag,\cc{I}) = S^\star(\dag^\star,\cc{I}^\star)$. Since the model $(B^\star,\{\Omega^{e,\star}\}_{e\in\mathcal{E}})$ satisfies Assumptions \ref{assm:int_heter} and \ref{assm:model_truth}, we appeal to \autoref{lemma:targets_parents} to conclude that a feasible $B$ in \eqref{eq:optimal_params} satisfies $B_{i,:} = B^\star_{i,:}$ for all $i \in \mathcal{I}^\star$. Since  $(I-B^\star)^{-1}{\Omega^{e,\star}}(I-B^\star)^{-T} = (I-B)^{-1}\Omega^e(I-B)^{-T}$, the previous conclusion implies that $\Omega^{e,\star}_{ii} = \Omega^e_{ii}$ for all $i \in \mathcal{I}^\star$ and feasible $\{\Omega^e\}_{e\in\mathcal{E}}$. Thus, all optimal $\mathcal{I}$ that are maximizers of $\argmax_{\dag,\cc{I}} S^\star(\dag,\cc{I})$ satisfy the property $\mathcal{I} \supseteq \mathcal{I}^\star$. Furthermore, it is straightforward to see that all $\cc{I} \supseteq \cc{I}^\star$ are also all optimal with respect to the optimization problem $\argmax_{\dag,\cc{I}} S^\star(\dag,\cc{I})$. 

We are now ready to prove the theorem statement. Let $\cc{I}$ be the set of intervention targets at the initial step of the backward phase of Algorithm~\ref{algo:outer}. From the assumption of Theorem~\ref{thm:outer_procedure}, we have that $\cc{I}^\star \subseteq \cc{I}$. Given the consistency of inner procedure, the backward phase of Algorithm~\ref{algo:outer} finds $j^\star \in \cc{I}$ given by $j^\star = \argmax_{G,j}S(G,\cc{I}\setminus{j})$. If $\max_{G}S(G,\cc{I}\setminus{j^\star}) > \max_{G}S(G,\cc{I})$, then the algorithm removes intervention target $j$ from $\cc{I}$. First, we show that the algorithm never removes $j \in \cc{I}^\star$. Specifically, decomposing $S(G,\cc{I})$ as $S(G,\cc{I}) - S^\star(G,\cc{I})+ S^\star(G,\cc{I})$ and from \eqref{eqn:temp}, we obtain:
\begin{eqnarray*}
\begin{aligned}
\max_{\dag} S(\dag,\cc{I}) &\geq \max_{\dag} S^\star(\dag,\cc{I}) - \lambda\max_{\dag}\mathrm{DoF}(\dag,\cc{I})\\
&-\max_{\dag}\max_{(B,\{\Omega^e\}_{e \in \cc{E}})\in\Theta, B \sim \dag}
    \sum_{e \in \mathcal{E}} \mathrm{tr}\left((I-B)^T(\Omega^e)^{-1}(I-B)(\hat{\pi}^e\hat{\Sigma}^e-{\pi}^{e,\star}\Sigma^{e,\star})\right)\\&-\max_{\dag}\max_{(B,\{\Omega^e\}_{e \in \cc{E}})\in\Theta,B \sim \dag}
    \sum_{e \in \mathcal{E}}(\hat{\pi}^e-\pi^{e,\star})\ln\mathrm{det}(\Omega^e), \\
\max_{\dag} S(\dag,\cc{I} \setminus\{j\}) &\leq \max_{\dag} S^\star(\dag,\cc{I} \setminus \{j\})\\
&+\max_{\dag}\max_{(B,\{\Omega^e\}_{e \in \cc{E}})\in\Theta, B \sim \dag}
    \sum_{e \in \mathcal{E}} \mathrm{tr}\left((I-B)^T(\Omega^e)^{-1}(I-B)(\hat{\pi}^e\hat{\Sigma}^e-{\pi}^{e,\star}\Sigma^{e,\star})\right)\\&+\max_{\dag}\max_{(B,\{\Omega^e\}_{e \in \cc{E}})\in\Theta,B \sim \dag}
    \sum_{e \in \mathcal{E}}(\hat{\pi}^e-\pi^{e,\star})\ln\mathrm{det}(\Omega^e).
\end{aligned}
\end{eqnarray*}
By the compactness constraint, 
\begin{gather*}
\max_\dag\max_{(B,\{\Omega^e\}_{e \in \cc{E}})\in\Theta}
    \sum_{e \in \mathcal{E}} \mathrm{tr}\left((I-B)^T(\Omega^e)^{-1}(I-B)(\hat{\pi}^e\hat{\Sigma}^e-{\pi}^{e,\star}\Sigma^{e,\star})\right) \to 0, \\
    \max_\dag\max_{(B,\{\Omega^e\}_{e \in \cc{E}})\in\Theta}
    \sum_{e \in \mathcal{E}}(\hat{\pi}^e-\pi^{e,\star})\ln\mathrm{det}(\Omega^e) \to 0,
\end{gather*}    
    as $n^e \to \infty$. Furthermore, since $\lambda \to 0$ as $n^e \to \infty$ for every $e \in \mathcal{E}$, and $\text{DoF}(\dag,\cc{I})$ is bounded, we have that $\lambda\text{DoF}(\dag,\cc{I}) \to 0$. From the earlier analysis, we have that $\max_{G} S^\star(G,\cc{I}) = S^\star(\mathcal{G}^\star,\cc{I}^\star)$ and that $S^\star(\mathcal{G}^\star,\cc{I}^\star) > \max_\dag S^\star(\dag,\cc{I}\setminus\{j\})$. Putting everything together, we conclude that there exists an integer $N$ such that for $n^e \geq N$ for every environment $e \in \mathcal{E}$, $\max_\dag S(\dag,\cc{I}) > \max_\dag S(\dag,\cc{I}\setminus\{j\})$. In other words, in the infinite data regime, any node $j \in \cc{I}^\star$ will not be removed by the update on $\cc{I}$.  

    We now consider any potential $j \in \mathcal{I} \setminus \cc{I}^\star$. We will argue that in the infinite data regime, $\max_\dag S(\dag,\cc{I}) < \max_\dag S(\dag,\cc{I}\setminus\{j\})$ so that our algorithm will remove $j$ from $\cc{I}$. From \eqref{eqn:temp}, we can conclude that:

    \begin{equation*}
\begin{aligned}
\max_{\dag}S(\dag,\cc{I}) &\geq \max_\dag S^\star(\dag,\cc{I}) - \max_\dag\max_{(B,\{\Omega^e\}_{e \in \cc{E}})\in\Theta}
    \sum_{e \in \mathcal{E}} \mathrm{tr}\left((I-B)^T(\Omega^e)^{-1}(I-B)(\hat{\pi}^e\hat{\Sigma}^e-{\pi}^{e,\star}\Sigma^{e,\star}))\right)\\
    &-\max_\dag\max_{(B,\{\Omega^e\}_{e \in \cc{E}})\in\Theta}
    \sum_{e \in \mathcal{E}}(\hat{\pi}^e-\pi^{e,\star})\ln\mathrm{det}(\Omega^e)- \max_\dag\lambda~\text{DoF}(\dag,\cc{I}),\\
\max_\dag S(\dag,\cc{I}) &\leq \max_\dag S^\star(\dag,\cc{I}) + \max_{(B,\{\Omega^e\}_{e \in \cc{E}})\in\Theta}
    \sum_{e \in \mathcal{E}} \mathrm{tr}\left((I-B)^T(\Omega^e)^{-1}(I-B)(\hat{\pi}^e\hat{\Sigma}^e-{\pi}^{e,\star}\Sigma^{e,\star})\right)\\&+\max_\dag\max_{(B,\{\Omega^e\}_{e \in \cc{E}})\in\Theta}
    \sum_{e \in \mathcal{E}}(\hat{\pi}^e-\pi^{e,\star})\ln\mathrm{det}(\Omega^e).
\end{aligned}
\label{eqn:upper_lower_S2}
\end{equation*}
Again, appealing to the compactness constraint, and that $\lambda \to 0$ as $n^e \to \infty$ for every $e \in \mathcal{E}$, we conclude that $\max_{\dag} S(\dag,\cc{I})  \overset{p}{\to} \max_\dag S^\star(\dag,\cc{I})$. Thus, as $n^e \to \infty$, the maximizers $\argmax_{\dag} S(\dag,\cc{I}) \subseteq \argmax_\dag S^\star(\dag,\cc{I})$. Since $\cc{I} \supseteq \cc{I}^\star$, recall that the maximizers $\argmax_\dag S^\star(\dag,\cc{I})$ are distributional equivalent models. Further, among distributional equivalent models $(\dag,\cc{I})$ and $(\dag',\cc{I})$, if $\mathrm{DoF}(\dag,\cc{I}) < \mathrm{DoF}(\dag,\cc{I})$, and the regularization $\lambda$ is above fluctuations due to sampling error
\begin{equation*}
\begin{aligned}
\lambda &\gg \left|\max_{(B,\{\Omega^e\}_{e \in \cc{E}})\in\Theta}
    \sum_{e \in \mathcal{E}} \mathrm{tr}\left((I-B)^T(\Omega^e)^{-1}(I-B)(\pi^{e,\star}\Sigma^{e,\star}-\hat{\pi}^e\hat{\Sigma}^e)\right)\right|\\&+ \left|\max_{(B,\{\Omega^e\}_{e \in \cc{E}})\in\Theta}
    \sum_{e \in \mathcal{E}}(\hat{\pi}^e-\pi^{e,\star})\ln\mathrm{det}(\Omega^e)\right|,
\end{aligned}
\end{equation*}
then there exists $N$ such that for any $n^e < N$, $S(G,\cc{I}) > S(G',\cc{I})$. This allows us to conclude that the maximizers of $\argmax_\dag S(\dag,\cc{I})$ are given by the solutions to the following optimization problem
\begin{equation*}
\begin{aligned}
    \argmin_{\dag}&~~\text{DoF}(\dag)\quad \text{subject-to }~~ B \sim \dag~~;~~\Sigma_e^\star = (I-B)^{-1}\Omega^e(I-B)^{-T} ~~\text{ for every }e\in\mathcal{E}.
    \end{aligned}
    \label{eq:optimal_params_b_}
\end{equation*}
From \autoref{prop:imec_full}, we know in the infinite data regime, as $n^e \to \infty$ for every $e \in \mathcal{E}$, $\lambda \to 0$ and $\lambda$ being above the fluctuations due to sampling error (described above), the maximizes $\argmax_\dag S(\dag,\cc{I})$ converge to $\simec(\dag^\star)$. 

Since $j \not\in \cc{I}^\star$, a similar argument as above allows us to conclude that
$\max_\dag S(\dag,\cc{I} \setminus \{j\}) \overset{p}{\to} S^\star(\dag^\star,\cc{I}^\star)$ as $n^e \to \infty$ and the maximizers of
$\argmax_\dag S(\dag,\cc{I} \setminus \{j\})$ converge to $\simec(\dag^\star)$. Let $N$ be the integer such that for $n^e \geq N$, maximizers of $\argmax_\dag S(\dag,\cc{I} \setminus \{j\})$ and $\argmax_\dag S(\dag,\cc{I})$ are both $\simec(\dag^\star)$. Since $\mathrm{DoF}(\dag)$ is the same for every $\dag \in \simec(\dag^\star)$, and that $\mathrm{DoF}(G,\cc{I} \setminus \{j\}) + 1 = \mathrm{DoF}(G,\cc{I})$, we conclude that for $n^e \geq N$, $\max_\dag S(\dag,\cc{I}) < \max_\dag S(\dag,\cc{I}\setminus\{j\})$, as desired. 

We have thus concluded that in the first step of the backward phase of the algorithm, we remove some $j \in \cc{I} \setminus \cc{I}^\star$. We can repeat the same argument to conclude that estimated intervention targets estimated by our algorithm converges to $\cc{I}^\star$ and the estimated equivalence class of graphs converges to $\simec(\dag^\star)$.

\end{proof}

 \section{Discussion of Assumption \ref{assm:model_truth}}
\label{s:assm}

We restate \autoref{assm:model_truth} here for completeness.

\modeltruth*
\noindent
A violation of \autoref{assm:model_truth} would imply that, given some conditioning set, the conditional variance of a variable, is not a function of the variance of its noise term. We formalize and prove this statement in \autoref{lemma:violation} below.

\begin{lemma}[Violation of \autoref{assm:model_truth}]
\label{lemma:violation}
Let $(B,\Omega)$ be a model and $X := (I-B)^{-1} \epsilon$ with $\epsilon \sim \cc{N}(0, \Omega)$ the normal random vector it entails. If there exists an equivalent model $(B', \Omega') \in [(B, \Omega)]$ such that $[(I-B')(I-B)^{-1}]_{ii} = 0$, then there exists $S \subseteq [p] \setminus \{i\}$ for which
the variance of $X_i$ given $X_S$ is not a function of the variance of its noise-term $\epsilon_i$.
\end{lemma}
\begin{proof}
Because $(B', \Omega') \in [(B, \Omega)]$, we have that
\begin{equation}
\label{eq:same_covariance_apx}
    (I-B)^{-1}\Omega(I-B)^{-T} = (I-B')^{-1}\Omega'(I-B')^{-T}.
\end{equation}
Now, let $M := (I-B')(I-B)^{-1}$ and assume $M_{ii} = 0$ for some $i \in [p]$. It must hold that $\supp{B_{i:}} \neq \supp{B'_{i:}}$, or by Lemmas \ref{lemma:same_support} and \ref{lemma:structure} we would have $[(I-B)(I-B')]_{ii} = 1$.
Define $X' := (I-B')^{-1} \epsilon'$ with $\epsilon' \sim \cc{N}(0, \Omega')$ and let $S:= \supp{B'_{i:}}$.
Because the models are equivalent, the random vectors $X$ and $X'$ follow the same distribution. Thus, for all $x_S \in \bb{R}^{\abs{S}}$
\begin{align*}
     \var(X_i \mid X_s = x_S) &= \var(X'_i \mid X'_s = x_S)\\
     \text{by \eqref{eq:model}}\quad &= \var\Big({\sum_{j \in S} B'_{ij} X'_j + \epsilon'_i \mid X'_s = x_S}\Big)\\
     &= \var\left(B'_{iS} x_S + \epsilon'_i \mid X'_s = x_S\right)\\
     &= \var\left(\epsilon'_i \mid X'_s = x_S\right)\\
     \text{by \autoref{lemma:independence}}\quad&= \var(\epsilon'_i) = \Omega'_{ii}\\
     \text{by \eqref{eq:same_covariance_apx}}\quad&= \sum_{j \in [p]} M^2_{ij}\Omega_{jj}\\
     \text{because $M_{ii} = 0$}\quad&= \sum_{j \neq i} M^2_{ij}\Omega_{jj}.
\end{align*}
In other words, the conditional variance of $X_i$ given $X_S$ is not a function of its noise-term variance $\Omega_{ii}$.
\end{proof}
As additional empirical support for the claim that \autoref{assm:model_truth} is a direct consequence of faithfulness and our modeling assumptions, we check that the assumption indeed holds for all random models generated in the synthetic experiments of \autoref{s:experiments}.

\ntext{\section{Violation of score equivalence by the $\ell_0$-penalized score}
\label{s:score_equivalence_violation}

We offer some discussion and an example of how the score corresponding to the objective in \eqref{eq:mle_plain}, i.e., the $\ell_0$-penalized likelihood score from \citet{chickering2002optimal}, does not satisfy score equivalence. Namely, graphs in the same $\cc{I}$-equivalence class may obtain different scores on the same finite sample. For compactness, we state the score corresponding to \eqref{eq:mle_plain} here
\begin{equation}
\label{eq:l0_score}
\begin{aligned}
        \cc{S}(\cc{G}) := \max_{\substack{
    B \comp \text{DAG}\\    
    \text{diagonal pos. def. } \{\Omega^e\}_{e \in     
    \cc{E}}
    }}
    \;
    \sum_{e \in \cc{E}}
    \log \p(\boldsymbol{X}^e; B, \Omega^e)    
    - \lambda \norm{\cc{G}(B)}_0,
\end{aligned}
\end{equation}
The likelihood term can be formulated in terms of the sample covariances and covariances implied by the model for each environment $e \in \cc{E}$, that is
\begin{equation}
    \log \p\left(\boldsymbol{X}^e; {B}, \Omega^e\right) = -n_e\ln\mathrm{det}({K}^e) - n_e\mathrm{tr}({K}^e\hat{\Sigma}^e),
\end{equation}
where $\hat{\Sigma}^e$ is the sample covariance for environment $e$, and $K^e := [(I-B)^{-1}\Omega^e(I-B)^{-T}]^{-1}$ is the inverse of the covariance matrix implied by $B$ and $\Omega^e$ for the graph $\cc{G}$. Thus, a necessary condition for two graphs to attain the same score on any finite sample is that they produce the same sets of covariance matrices under the constraints in \eqref{eq:l0_score}. For a single environment ($|\cc{E}| = 1$, i.e., $\cc{I} = \emptyset$), this is true for any two Markov-equivalent graphs (see proof of \autoref{prop:score_properties} in Appendix~\ref{s:proofs_score_properties}). However this does not hold for $|\cc{E}| > 1$, because---without the additional (sufficient) constraints placed on the noise-term matrices $\{\Omega^e\}_{e \in \cc{E}}$ by the GnIES score in \eqref{eq:dist_condition}---two $\cc{I}$-equivalent graphs may yield disjoint sets of covariance matrices for each of the environments in $\cc{E}$.

As an example, consider the following two graphs: $X_1 \to X_2$, which we denote by $\cc{G}$, and $X_1 \leftarrow X_2$, which we denote by $\tilde{\cc{G}}$. The two are $\cc{I}$-equivalent for $\cc{I}=\emptyset$. Let $(B, \{\Omega, \Omega'\})$ and $(\tilde{B}, \{\tilde{\Omega}, \tilde{\Omega}'\})$ be any two models that satisfy the constraints in \eqref{eq:l0_score} for $\cc{G}$ and $\tilde{\cc{G}}$, respectively. Additionally, let $\Omega_{11} = \Omega'_{11}$, $\Omega'_{22} \neq \Omega_{22}$, and $B_{21} = 1$. For the first environment, the covariance implied by the first model (graph $\cc{G}$) is given by
$$\Sigma = (I-B)^{-1}\Omega(I-B)^{-T} = \begin{pmatrix}
\Omega_{11} & B_{21} \Omega_{11}\\
B_{21} \Omega_{11} & B_{21}^2\Omega_{11} + \Omega_{22}
\end{pmatrix},$$
and the covariance implied by $\tilde{\cc{G}}$ is
$$\tilde{\Sigma} = \begin{pmatrix}
\tilde{B}_{12}^2\tilde{\Omega}_{22} + \tilde{\Omega}_{11} & \tilde{B}_{12} \tilde{\Omega}_{22}\\
\tilde{B}_{12} \tilde{\Omega}_{22} & \tilde{\Omega}_{22}
\end{pmatrix}.$$
The two covariances are equal if and only if (note: $B_{21} = 1$)
\begin{equation}
\label{eq:counter_beta}
\tilde{\Omega}_{11} = \Omega_{11} - \frac{\Omega_{11}}{\Omega_{11} + \Omega_{22}},\; \tilde{\Omega}_{22} = \Omega_{11} + \Omega_{22},\text{ and }\; \tilde{B}_{12} = \frac{\Omega_{11}}{\Omega_{11} + \Omega_{22}}.
\end{equation}
For the second environment, the covariances matrices are given by
$$\Sigma = \begin{pmatrix}
\Omega'_{11} & B_{21} \Omega'_{11}\\
B_{21} \Omega'_{11} & B_{21}^2\Omega'_{11} + \Omega'_{22}
\end{pmatrix}\text{ and }\tilde{\Sigma}' = \begin{pmatrix}
\tilde{B}_{12}^2\tilde{\Omega}'_{22} + \tilde{\Omega}'_{11} & \tilde{B}_{12} \tilde{\Omega}'_{22}\\
\tilde{B}_{12} \tilde{\Omega}'_{22} & \tilde{\Omega}'_{22}
\end{pmatrix}.$$
As before, the two covariances are equal if and only if
$$\tilde{\Omega}'_{11} = \Omega'_{11} - \frac{\Omega'_{11}}{\Omega'_{11} + \Omega'_{22}},\; \tilde{\Omega}'_{22} = \Omega'_{11} + \Omega'_{22},\text{ and }\; \tilde{B}_{12} = \frac{\Omega'_{11}}{\Omega'_{11} + \Omega'_{22}} = \frac{\Omega_{11}}{\Omega_{11} + \Omega'_{22}}.$$
However, this is a contradiction since the value obtained for $\tilde{B}_{12}$ above differs from that in \eqref{eq:counter_beta}, which is not possible as $\tilde{B}_{12}$ is defined as being constant across environments.}
\section{GnIES Completion Algorithm}
\label{s:completion_algorithm}

Below we detail the completion algorithm used at the end of each step of the inner procedure from GnIES. The algorithm is also employed in the experiments to compute an estimate of the equivalence class from the graph and intervention targets returned by UT-IGSP.

\begin{algorithm}[H]
\caption{Completion algorithm for the inner procedure of GnIES}
\begin{algorithmic}[1]
\vspace{0.1in}
\STATE {\bf Input}: PDAG $P$, intervention set $\cc{I}$.\newline
Precondition: The PDAG $P$ has no undirected in- or out-going edges from the nodes in $\cc{I}$; this is guaranteed during the inner procedure of GnIES.
\vspace{0.05in}
\STATE {\bf Obtain the CPDAG} $C$ by completing $P$ using the GES completion algorithm \citep[Appendix C]{chickering2002optimal}
\STATE {\bf Orient edges}: For every node in $\cc{I}$, orient all in- and out-going edges to match the orientation in the original PDAG $P$
\STATE {\bf Apply the Meek rules} \citep{Meek1995CausalIA} iteratively to $C$ until no more edges can be oriented.
\STATE{\bf Output:} The I-CPDAG $C$ representing the $\imec$.
\end{algorithmic} \label{algo:completion}
\end{algorithm}
\section{Software Contributions}
\label{s:software}

The code, data sets and instructions to reproduce the experiments and figures in the manuscript can be found in the repository \href{https://github.com/juangamella/gnies-paper}{\texttt{github.com/juangamella/gnies-paper}}. The Python implementation of the GnIES algorithm is available through the \href{https://github.com/juangamella/gnies}{\texttt{gnies}} package (see details below). During our work on this paper, we developed additional software which we make available through the Python packages listed below. The packages are well documented, actively maintained, and available through the python package index\footnote{See \href{https://pypi.org/}{\texttt{https://pypi.org/}}.}, that is, they can easily be installed with Python's package manager by typing \texttt{pip install <package\_name>} into a suitable terminal. More details and installation instructions can be found in their respective repositories, listed below.

\begin{center}
\begin{tabular}{ | m{1.5cm} | m{5.5cm}| m{7cm} | } 
  \hline
  Package & Description & Details and documentation \\ 
  \hline
  \hline
  \texttt{gnies} & The GnIES algorithm. & \href{https://github.com/juangamella/gnies}{\texttt{github.com/juangamella/gnies}}\\
  \hline
  \texttt{ges} & A pure Python implementation of the GES algorithm \citep{chickering2002optimal}. & \href{https://github.com/juangamella/ges}{\texttt{github.com/juangamella/ges}}\\
  \hline
  \texttt{gies} & A pure Python implementation of the GIES algorithm from \citet{hauser2012characterization}. & \href{https://github.com/juangamella/gies}{\texttt{github.com/juangamella/gies}}\\
  \hline
  \texttt{sempler} & Used to generate the synthetic and semi-synthetic data used in the experiments. & \href{https://github.com/juangamella/sempler}{\texttt{github.com/juangamella/sempler}}\\
  \hline
  
\end{tabular}
\end{center}

\section{Semi-synthetic Data Generation Procedure}
\label{s:drfs}

In this section we summarize the procedure used to generate the semi-synthetic data set from \autoref{s:experiments}, generated by combining the biological data set and consensus graph given in \citet{sachs2005}. The procedure is general in that it can be used with any other data set and graph; an implementation is provided in the python package\footnote{See \href{https://github.com/juangamella/sempler}{\texttt{github.com/juangamella/sempler}} for more details.} \texttt{sempler}.

As input, we are given a directed acyclic graph $\cc{G}$ over $p$ variables and a data set
consisting of their observations under different environments. The procedure then fits a non-parametric Bayesian network which approximates the conditional distributions entailed by the factorization of the graph $\cc{G}$ via distributional random forests \citep{cevid2020distributional}. Once fitted, we can sample from this collection of forests to produce a new, semi-synthetic data set that respects the conditional independence relationships entailed by $\cc{G}$, while its marginal and conditional distributions closely match those of the original data set (Figures \ref{fig:sachs_normality} and \ref{fig:sachs_nonlinear}).

The procedure is motivated by a fundamental concept from graphical models. We say a distribution $P$ satisfies the \emph{global Markov} property with respect to a DAG $\cc{G}$ if and only if it factorizes according to its parental relations \citep{pearl1988probabilistic}, that is, if
$$\p(X_1,...,X_p) = \prod_{j \in [p]} \p(X_j\mid X_{\pa_\cc{G}(j)}).$$
Thus, we can model the joint data-generating distribution by independently modeling the conditional distribution of each variable given its parents in $\cc{G}$. We take this approach and model the collection of data-generating distributions $\{\p^e\}_{e \in \cc{E}}$ as follows: for each environment $e \in \cc{E}$ and each node $j \in [p]$, we estimate $\hat\p^e_{X_j \mid X_{\pa_\cc{G}(j)}=x}$ by fitting a distributional random forest \citep{cevid2020distributional}. Note that while the distributions $\{P^e\}_{e \in \cc{E}}$ of the data may not satisfy the Markov condition with respect to $\cc{G}$, the estimated distributions $\{\hat{P}^e\}_{e \in \cc{E}}$ do. 

We can now generate new data by sampling from the estimated conditional distributions, following the causal ordering of the DAG $\cc{G}$. We proceed as follows:
\begin{enumerate}
    \item First, we obtain the causal ordering $O_\cc{G}$ of $\cc{G}$.
    \item Then, for every environment $e \in \cc{E}$ and the next (or first) node in $O_\cc{G}$:
    \begin{enumerate}
        \item if $j$ is a source node, we generate a sample $\left\{\hat{x}^e_j(i)\right\}_{i=1,...,n_e}$ via bootstrapping, where $n_e$ is the number of observations from environment $e$ in the original data set.
        \item If $j$ is not a source node, for every new observation $i=1,...,n_ e$ we sample  
        $$\hat{X}^e_j(i) \sim \hat{\p}^e_{X_j \mid X_{\pa_\cc{G}(j)} = \hat{x}^e_{\pa_\cc{G}(j)}(i)}$$
        using the technique provided by \citet{cevid2020distributional}, where $\hat{x}^e_{\pa_\cc{G}(j)}(i)$ is the previously sampled observation for the parents of $j$.
    \end{enumerate}
    \item We collect all the samples, resulting in a new semi-synthetic data set.
\end{enumerate}

\begin{figure}[H]
\centering
\includegraphics[width=0.99\textwidth]{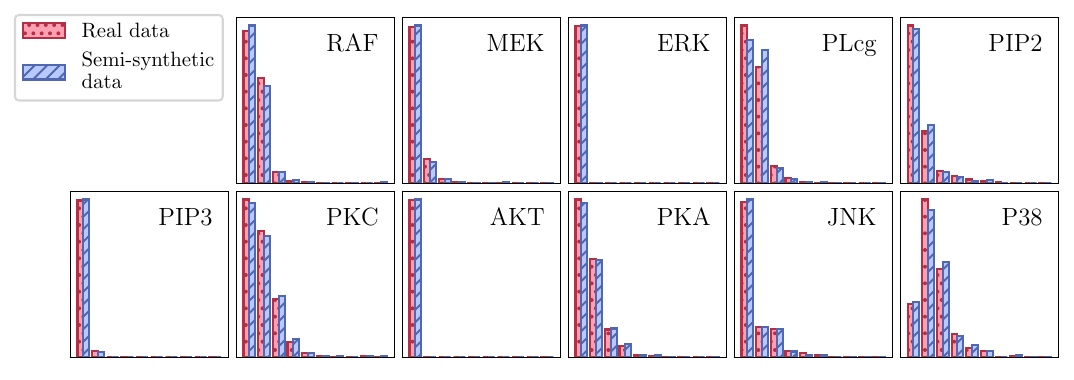}
\caption{Histograms of the observations of each variable measured in the Sachs data set, for the real data (red, dots) and the semi-synthetic data (blue, stripes) sampled from the network fitted to the data set using distributional random forests \citep{cevid2020distributional} following the consensus network. The marginals of the semi-synthetic data closely resemble those of the real data, and both show strong indications of non-normality.}
\label{fig:sachs_normality}
\end{figure}

\begin{figure}[H]
\centering
\includegraphics[width=1.0\textwidth]{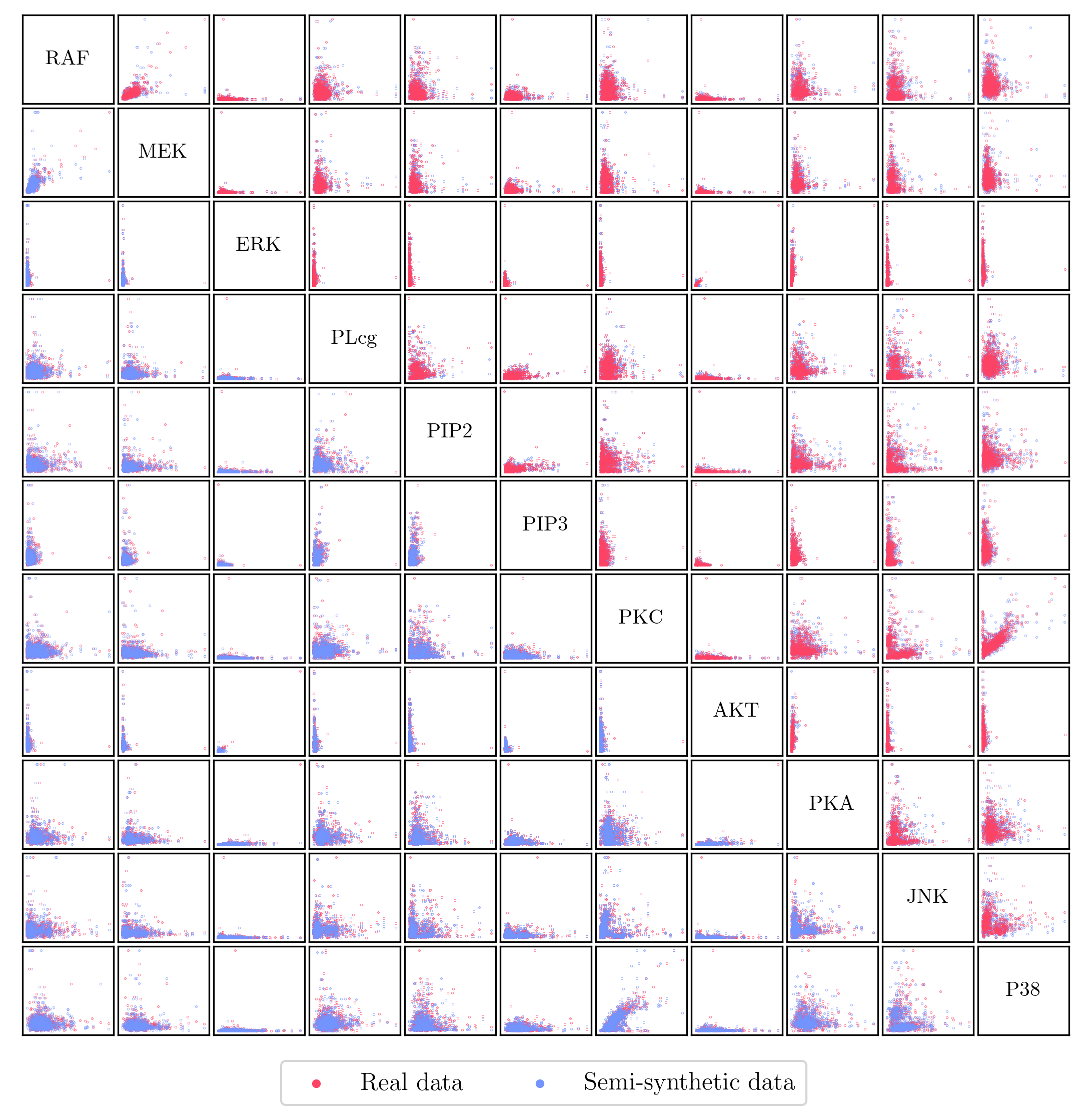}
\caption{Pairwise observations of all variables in the Sachs data set, for the real (red) and semi-synthetic (blue) data. The conditional distributions of the semi-synthetic data set closely follow those of the real data set, and both exhibit a high degree of non-linearity.}
\label{fig:sachs_nonlinear}
\end{figure}


\ntext{ \section{\ntext{Additional Synthetic Experiments}}
\label{s:additional_experiments}

In this section we provide the results of additional synthetic experiments to evaluate GnIES and the compared methods on data from denser ground-truth graphs (\autoref{fig:denser}), multiple target interventions (\autoref{fig:multi}), and a misspecified intervention model (\autoref{fig:model_mismatch}). In all figures, we evaluate the recovery of the ground-truth equivalence class using the metrics described in \autoref{def:metrics}; the point $(0,1)$ corresponds to perfect recovery. The results are an average across 1000 data sets originating from 100 linear Gaussian SCMs, which are otherwise generated following the procedure in \autoref{ss:sim}. We show the regularization path for all methods. For GnIES, GIES and GES, the regularizer is set to $\lambda := \lambda'\log(N)$ where $N$ is the total number of observations and $\lambda' \in [0.01, 2]$; the additional point shown on the regularization path of these methods corresponds to the BIC score ($\lambda'=1/2$). For UT-IGSP, the regularizers are the levels of the conditional independence and invariance tests, taken between $10^{-5}$ and $0.1$. The poor performance of sortnregress in all experiments indicates the low varsortability in the data after it is standardized.

\subsection{Misspecified intervention model}
\label{ss:hard_interventions}

We evaluate the performance of GnIES on data from a misspecified intervention model, where the interventions are \emph{do-} or hard interventions \citep[Section 3.2.2]{pearl2009causal} instead of the noise interventions described in \autoref{s:model}. In particular, the interventions remove the causal effect of the parents on the target and set it to a normal distribution with mean zero and variance sampled uniformly at random from $[3,4]$. The data is otherwise generated as in \autoref{ss:sim}.

We show the results in \autoref{fig:model_mismatch}. We employ GIES \citep{hauser2012characterization} as an additional baseline, as this setting precisely satisfies the modeling assumptions of the method save for one aspect: that knowledge of the intervention targets is required. To overcome this we provide GIES with the true intervention targets, thus acting as an indicator of the ``optimum performance'' we would hope to obtain if we had such knowledge. Furthermore, the equivalence class of the data-generating model is no longer the $\cc{I}$-equivalence class described in this paper, but rather the interventional Markov equivalence class introduced by \citet{hauser2012characterization}. In turn, we now use the completion algorithm from GIES to compute the equivalence class estimated by UT-IGSP, from its estimated DAG and intervention targets; we note that, in practice, this would require knowledge of the model violation. Because UT-IGSP allows for a wider range of intervention types, its performance is essentially undisturbed by the change to hard interventions on the same targets. The results are shown in \autoref{fig:model_mismatch} below. As expected, GnIES is sensitive to the misspecification and performs worse than GIES. The relative advantage of GnIES over UT-IGSP is reduced and then reversed as the sample size increases.

\begin{figure}[h]
\centering
\includegraphics[width=0.85\textwidth]{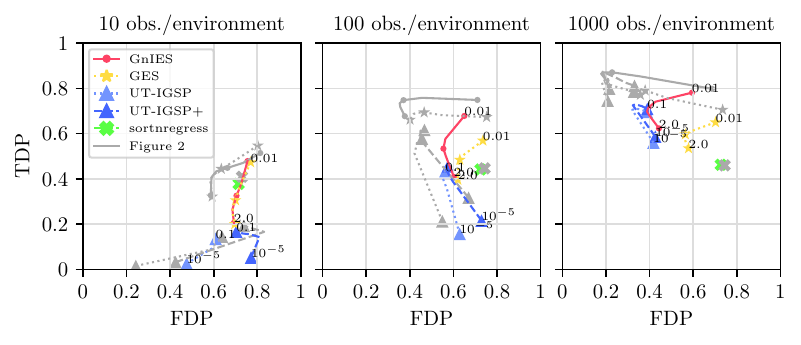}
\caption{\ntext{(denser graphs) Comparing the effect of denser ground-truth graphs with the results from \autoref{fig:model_match} (shown in gray). The ground-truth models and datasets are generated as in \autoref{ss:sim}, but we increase the average degree of the underlying Erd{\"o}s-Renyi graphs from 2.7 to 4. As expected, the performance of all methods deteriorates, but their relative performance remains similar.}}
\label{fig:denser}
\end{figure}

\begin{figure}[h]
\centering
\includegraphics[width=0.85\textwidth]{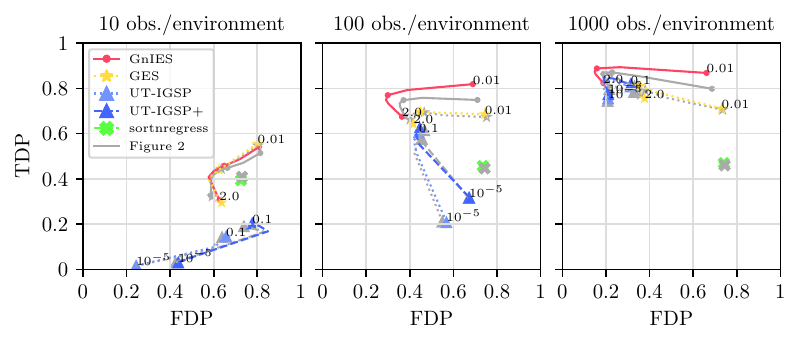}
\caption{\ntext{(multiple target interventions) Comparing the effect of multiple-target interventions with the results from \autoref{fig:model_match} (shown in gray). The datasets are generated as in \autoref{ss:sim}, but the number of intervention targets in each environment is increased to 3. The performance of the compared methods remains virtually unchanged or even improves slightly for GnIES in the larger sample sizes.}}
\label{fig:multi}
\end{figure}

\begin{figure}[h]
\centering
\includegraphics[width=0.85\textwidth]{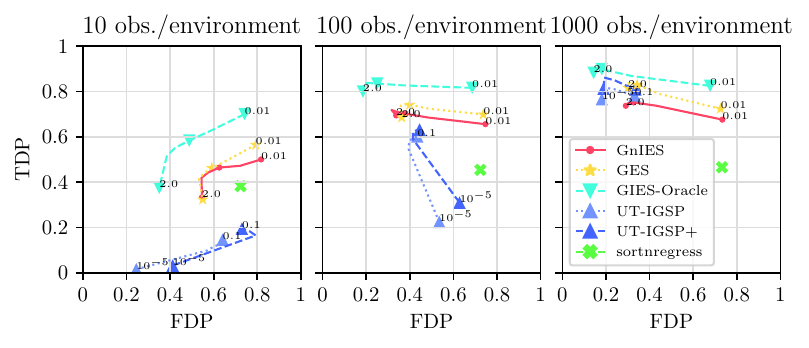}
\caption{\ntext{(misspecified intervention model) In this case, the true equivalence class is the interventional Markov equivalence class given in \citep{hauser2012characterization} for hard interventions (c.f. $\cc{H}$-MEC in \autoref{sss:hauser}). GIES is run with full knowledge of the true intervention targets. Because UT-IGSP returns a single DAG as an estimate, we use its estimated targets and the completion algorithm of GIES to produce an estimate of the equivalence class (displayed as UT-IGSP+). For each method, the regularization path is shown.}}
\label{fig:model_mismatch}
\end{figure}

}
\section{Additional Figures}
\label{s:figs}

\begin{figure}[H]
\centering
\includegraphics[width=0.97\textwidth]{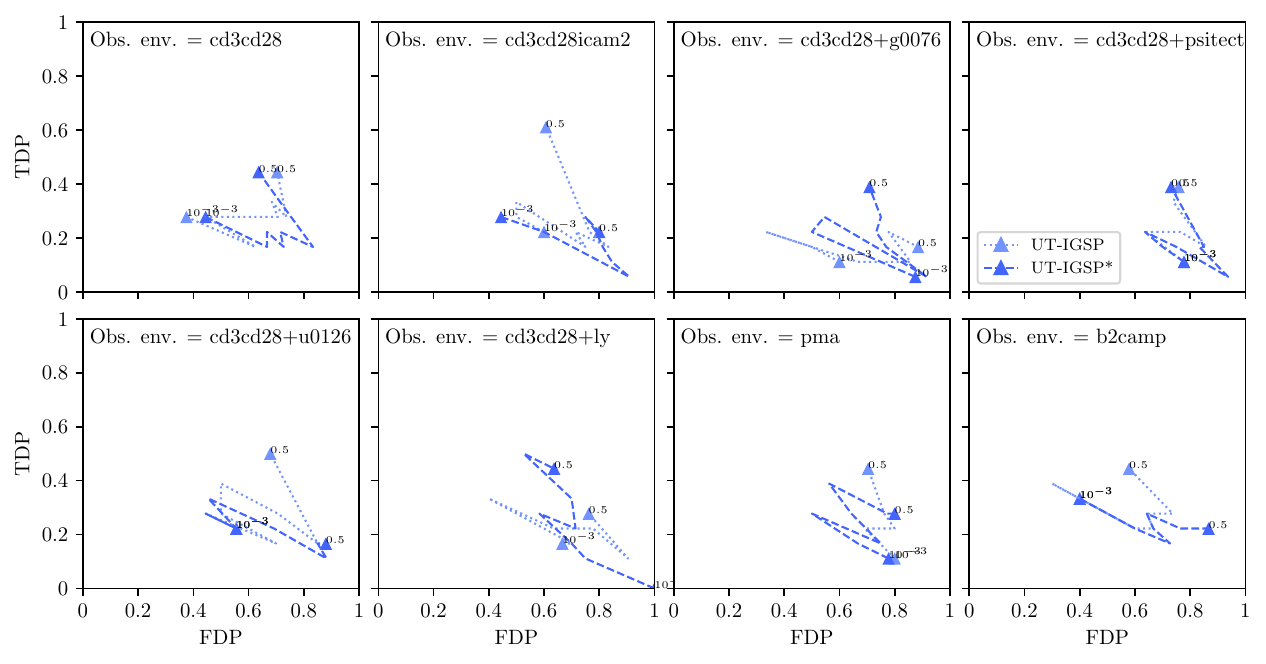}
\caption{Performance of UT-IGSP on the Sachs data set for other choices of the observational environment. The remaining environment \texttt{cd3cd28+aktinhib} is displayed in \autoref{fig:sachs} in the main text.}
\label{fig:sachs_obs_real}
\end{figure}

\begin{figure}[H]
\centering
\includegraphics[width=1\textwidth]{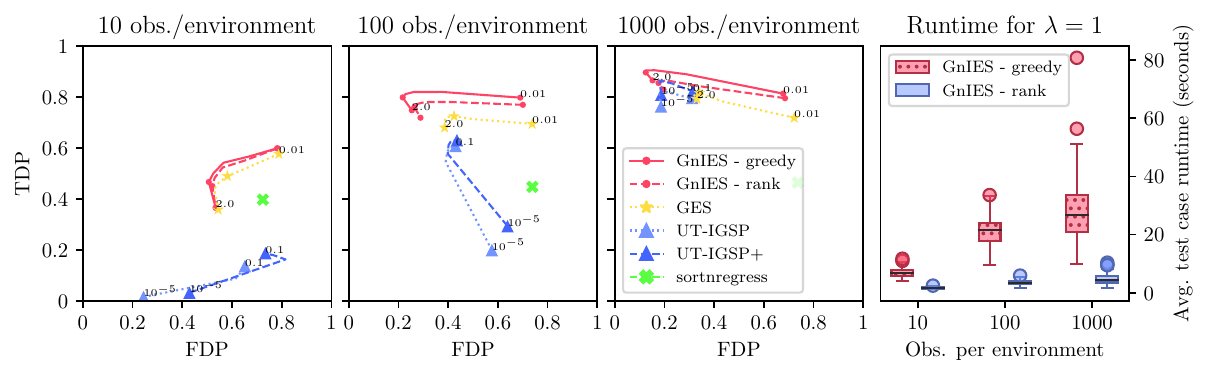}
\caption{Recovery of the interventional Markov equivalence class for GnIES with the ranking outer procedure (red, circles, dashed) compared to the greedy approach and other algorithms. Bottom right: its runtime compared to GnIES with a greedy outer procedure. The data sets, metrics, and other algorithms are the same as in \autoref{fig:model_match}. At the cost of a slightly reduced performance in class recovery, the ranking procedure results in a significant speed-up. The relative weight of the penalization term in the score decreases as the sample size increases, resulting in more targets being added and a higher runtime for both approaches.}
\label{fig:performance}
\end{figure}

\begin{figure}[H]
\centering
\includegraphics[width=0.85\textwidth]{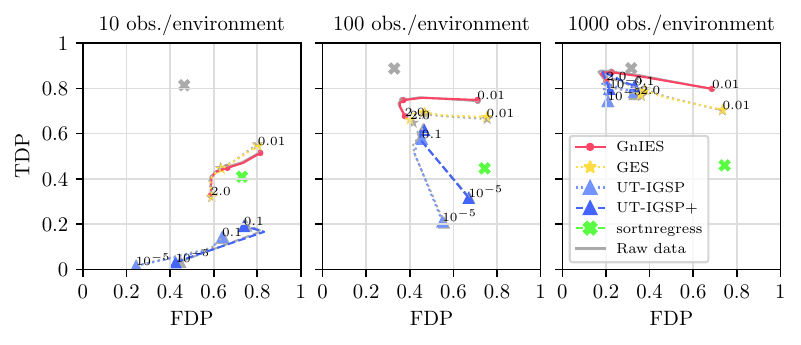}
\includegraphics[width=0.85\textwidth]{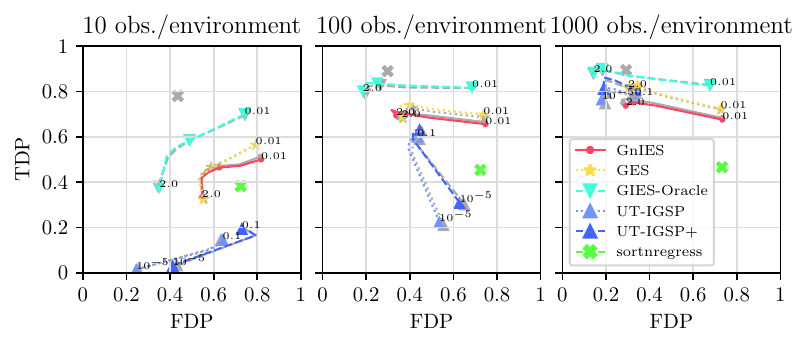}
\caption{Performance of GnIES and the other algorithms when fed standardized (color) and raw (gray) data, for the settings with a model match (top, cf. \autoref{fig:model_match}) and a model mismatch in the form of \emph{hard} interventions (bottom, cf. \autoref{fig:model_mismatch}). The sortnregress algorithm explicitly exploits varsortability, and its performance is heavily affected when it is removed from the data via standardization. The other algorithms, including GnIES, show no significant change in performance; this can be taken as an indication that they do not rely on varsortability for inference. In both settings, the average varsortability score \citep[Section 3]{reisach2021varsort} is around $0.9$ for the raw data and $0.5$ for the synthetic data.}
\label{fig:varsort}
\end{figure}


\bibliography{refs}

\end{document}